\pdfoutput=1
\def\passOptions#1#2{\PassOptionsToPackage{#2}{#1}}

\passOptions{cleveref}{capitalise,nameinlink}
\passOptions{hyperref}{colorlinks,allcolors=blue, linkcolor=red}
\passOptions{xcolor}{table}
\documentclass[dvipsnames,acmsmall,screen,nonacm]{acmart}

\usepackage{\jobname}

\AtBeginDocument{%
  \providecommand\BibTeX{{%
      \normalfont B\kern-0.5em{\scshape i\kern-0.25em b}\kern-0.8em\TeX}}}

\setcopyright{acmcopyright}
\copyrightyear{2020}
\acmYear{2020}
\acmDOI{10.1145/1122445.1122456}

\acmConference[Woodstock '18]{Woodstock '18: ACM Symposium on Neural
  Gaze Detection}{June 03--05, 2018}{Woodstock, NY}
\acmBooktitle{Woodstock '18: ACM Symposium on Neural Gaze Detection,
  June 03--05, 2018, Woodstock, NY}
\acmPrice{15.00}
\acmISBN{978-1-4503-XXXX-X/18/06}

\begin{document}
\title[Ties between Type Systems and Automata]{Ties between Parametrically Polymorphic Type Systems and Finite Control Automata}
\subtitle{Extended Abstract}
\author{Joseph (Yossi) Gil}
\email{yogi@cs.Technion.ac.IL}
\author{Ori Roth}
\email{soriroth@cs.Technion.ac.IL}
\affiliation{%
	\institution{The Technion}
	\city{Haifa 32000} 
	\state{Israel}
}

\begin{abstract}
  We present a correspondence and bisimulation between variants of
parametrically polymorphic type systems and variants of
finite control automata, such as FSA, PDA, tree automata and Turing
machine. Within this correspondence we show that two recent
celebrated results on automatic generation of fluent API are optimal
in certain senses, present new results on the studied type systems, formulate
open problems, and present potential software engineering applications, other
than fluent API generation, which may benefit from judicious use of type
theory.

\end{abstract}


%TODO mandatory: Please choose ACM 2012 classifications from https://dl.acm.org/ccs/ccs_flat.cfm


\maketitle

\section{Introduction}
\label{section:aa}
Computational complexity of type checking is a key aspect of any type system.
Several classical results characterize this complexity in type systems where
the main type constructor is function application: Type checking in the
\emph{\textbf Simply \textbf Typed \textbf Lambda \textbf Calculus} (STLC), in
which function application is the sole type constructor, is carried out in linear
time. In the Hindley-Milner (HM) type
system~\cite{Hindley:69,Milner:78,Damas:82}, obtained by augmenting the STLC
with parametric polymorphism with \emph{unconstrained} type parameters,
type checking is harder, and was found to be deterministic exponential (DEXP) time
complete~\cite{Kfoury:1990}. However, the Girard–Reynolds type system
\cite{Girard:71,Girard:72,Reynolds:74} (System-F) which generalizes HM is
undecidable~\cite{Wells:99}.

\nomenclature[A]{HM}{Hindley-Milner (type system)}
\nomenclature[A]{STLC}{simply typed lambda calculus}
\nomenclature[A]{DEXP}{deterministic exponential}

In contrast, our work focuses in type systems where the main type constructor
is pair (or tuple), i.e., no higher order functions. This type constructor
models \emph{object based programming}, including concepts such as records,
classes and methods, but not inheritance. In particular, we investigate the
computational complexity of such systems in the presence of parametric
polymorphism, also called genericity, allowing generic classes and generic
functions.

We acknowledge significant past work on more general systems modeling the
combination of genericity with the \emph{object oriented programming} paradigm,
i.e., classes with single and even multiple inheritance. Type checking in these
is particularly challenging, since inheritance may be used to
place sub-type and super-type constraints on the parameters to generics. In
fact, Kennedy and Pierce~\citeyear{Kennedy:Pierce:07} showed that, in general,
such type systems are undecidable. Their work carefully analyzed the factors
that may lead to undecidability, and identified three decidable fragments, but
without analyzing their complexity. In fact, the presumed decidability of C＃'s
type system is a result of adopting one particular such fragment. Later,
Grigore \citeyear{Grigore:2017} proved that the combination of inheritance and
genericity in the Java type system makes its type system undecidable (after
Amin and Tate \citeyear{Amin:16} showed that is unsound).

Recent work on the practical fluent API problem, drew our attention to these
kind of systems. However, this work is mostly of theoretical nature. We present
a correspondence and bisimulation between variants of these type systems,
organized in a conceptual lattice~$𝔗$ of type systems, and variants of finite
control automata, such as FSA, PDA, tree automata and Turing machine, organized
in another conceptual lattice~$𝔄$.

\nomenclature[A]{API}{application programming interface}
\nomenclature[A]{FSA}{finite state automaton}
\nomenclature[A]{PDA}{pushdown automaton}
\nomenclature[B]{$𝔗$}{lattice of parametrically polymorphic type systems, see~\cref{table:lattice}}
\nomenclature[B]{$𝔄$}{lattice of finite control automata, see~\cref{table:automata}}

With this correspondence we determine the exact \emph{computational complexity
class} of type checking of many, but not all, type systems in~$𝔗$;
for other type systems, we provide upper and lower bounds, leaving the precise
characterizations as open problems. We also show that two celebrated results on
the fluent API problem, are optimal in certain senses.  The research also has
practical applications for language design, e.g.,
\cref{theorem:TA:capturing:TM} below shows that introducing functions whose
return type is declared by keyword \kk{auto} to C＃, would make its type system
undecidable.
\vspace{-2ex}

\subsection{Background}

Recall†{\Cref{section:fluent} gives more precise definitions and motivation for
the fluent API problem} that a \emph{fluent API generator} transforms~$ℓ$, the
\emph{formal language} that specifies the API, into~$L=L(ℓ)$, a \emph{library}
of type definitions in some \emph{target programming language}, e.g., Java,
Haskell, or {C++}. Library~$L(ℓ)$ is the \emph{fluent API library} of~$ℓ$ if an
expression~$e$ type checks (in the target language) against~$L(ℓ)$ \emph{if an
  only if} word~$w=w(e)$ is in the API language~$ℓ$,~$w(e)∈ℓ⇔\text{$e$ type
checks against~$L(ℓ)$}$. It is required that expression~$e$ is in the form of
a chain of method invocations. The word~$w=w(e)$ is obtained by enumerating, in
order, the names of methods in the chain of~$e$, e.g., a fluent API generator
for {C++} receives a language~$ℓ$ over alphabet~$Σ=❴a,b❵$, i.e.,~$ℓ⊆❴a,b❵^*$
and generates as output a set of {C++} definitions in which an expression such
as
\vspace{-1ex}
\begin{equation} \label{eq:first}
  \cc{new Begin().a().b().a().a().end()}
\vspace{-1ex}
\end{equation}
type checks if, and only if word~$abaa$ belongs in language~$ℓ$.

\nomenclature[B]{$L$}{library of type definitions in a programming language}
\nomenclature[B]{$ℓ$}{formal language}
\nomenclature[C]{$a$}{an example letter in alphabet}
\nomenclature[C]{$b$}{an example letter in alphabet}

The most recent such generator is \TypelevelLR due to Yamazaki, Nakamaru,
Ichikawa and Chiba~\citeyear{Yamazaki:2019}. \TypelevelLR compiles an \emph{LR
  language} †{\textbf Left-to-right, \textbf Rightmost derivation~\cite{Knuth:1965}}~$ℓ$ into a fluent API library~$L(ℓ)$ in either
Scala, {C++}, or, Haskell (augmented with ‟\textit{four GHC extensions:
  \texttt{MultiParamTypeClasses}, \texttt{FunctionalDep\-endencies},
  \texttt{Flex\-ible\-Instances}, and \texttt{Undecidable\-Instances}}”), but neither
  Java nor C＃.

The architecture of \TypelevelLR makes it possible to switch between different
front-ends to translate a context free grammar specification of~$ℓ$ into an
intermediate representation. Different such front-ends are SLR, LALR, LR(1)
grammar processors. Similarly to traditional multi-language compilers,
the front-ends compile the input specification into a library in \Fluent, an
intermediate language invented for this purpose; the back-ends of \TypelevelLR
translates~$L_\Fluent$ into an equivalent library~$L=L(L_\Fluent)$ in the
target languages.

\ifjournal
As shown in \cref{figure:multi} \TypelevelLR's architecture is comparable with
the architecture of traditional multi-language compilers: A multi-language
compiler revolves around a specification of an intermediate language, e.g.,
\emph{three address code}. Different front-ends compile different programming
languages into the intermediate language, and different back-ends translate the
intermediate language to different architectures.

This intermediate representation is then compiled by
different back-ends into the three target languages as per the clever
explanations of Yamazaki et al.

\begin{figure}
  \begin{tikzpicture}
	\node[draw,fill=black!10] at (0,0) (mediator) {\shortstack{Program in \\ \Fluent}};
	\node at ($(mediator)+(0,1)$) {\itshape TypelevelLR};
	
	\node (i1) [above left=of mediator] {SLR};
	\node (i2) [left=of mediator] {LALR};
	\node (i3) [below left=of mediator] {LR(1)};
	\node (o1) [above right=of mediator] {C++};
	\node (o2) [right=of mediator] {Haskell};
	\node (o3) [below right=of mediator] {Scala};
	
	\path[->] (i1) edge (mediator.north west);
	\path[->] (i2) edge (mediator.west);
	\path[->] (i3) edge (mediator.south west);
	\path[->] (mediator.north east) edge (o1);
	\path[->] (mediator.east) edge (o2);
	\path[->] (mediator.south east) edge (o3);
	
	\node[draw,fill=black!10] (tac) at (-7,0) {\shortstack{Three Address \\ Code}};
	\node (gcc) at ($(tac)+(0,1)$) {\itshape \shortstack{Multi Language \\ Compiler}};
	
	\node (c) [above left=of tac] {C};
	\node (pascal) [left=of tac] {Pascal};
	\node (ada) [below left=of tac] {Ada};
	\node (x86) [above right=of tac] {x86};
	\node (mips) [right=of tac] {\shortstack{mips \\ asm}};
	\node (arms) [below right=of tac] {\shortstack{arms \\ asm}};
	
	\path[->] (c) edge (tac.north west);
	\path[->] (pascal) edge (tac.west);
	\path[->] (ada) edge (tac.south west);
	\path[->] (tac.north east) edge (x86);
	\path[->] (tac.east) edge (mips);
	\path[->] (tac.south east) edge (arms);
\end{tikzpicture}
  \caption{Architecture of compiler from multiple programming languages to
    multiple architecture (left) and the architecture of \TypelevelLR
    (right).
}
  \label{figure:multi}
\end{figure}

Similarly, central to \TypelevelLR is a definition of the \Fluent intermediate
language: All front-ends of \TypelevelLR compile an input LR grammar~$G$ into a
library in \Fluent,~$L_{\Fluent}=L_{\Fluent}(G)$; the back-ends of \TypelevelLR
translates~$L_\Fluent$ into an equivalent library~$L=L(L_\Fluent)$ in the
target languages.
\fi
\nomenclature[B]{$G$}{context free grammar}
\nomenclature[O]{\Fluent}{intermediate language used in the implementation of \TypelevelLR}
\nomenclature[O]{\TypelevelLR}{a fluent API generator due to
Yamazaki, Nakamaru, Ichikawa and Chiba~\citeyear{Yamazaki:2019}}
\nomenclature[A]{LR}{left-to-right, right-most derivation}

\TypelevelLR strikes a sweet spot in terms of front-ends: It is a common
belief that most \emph{programming languages} are LR, so
there is no reason for a fluent API generator to support any wider class of
formal languages for the purpose of the \emph{mini programming language} of an
API. On the other hand, \TypelevelLR's client may tune down the generality of
\TypelevelLR, by selecting the most efficient front-end for the grammar of the
particular language of the fluent API.

We show that in terms of computational complexity, \TypelevelLR strikes another
sweet spot in selecting \Fluent, specifically, that~$\Fluent=\|LR|$
(\cref{theorem:fluent}). Equality here is understood in terms of classes of
computational complexity, i.e., for every set of definitions~$L$ in \Fluent,
there exists an equivalent formal language~$ℓ=ℓ(L)∈\|LR|$, and for
every~$ℓ∈\|LR|$ there exists equivalent library~$L=L(ℓ)$. Also term \Fluent in
the equality refers to the computational complexity class defined by the
\Fluent language. This abuse of notation is freely used henceforth.

Is there a similar sweet spot in the back-ends of \TypelevelLR? Why didn't
\TypelevelLR include neither \Fluent-into-Java nor \Fluent
into-C＃ back-ends? And, why were these GHC extensions selected and
not others? These and similar questions made the practical motivation for this
work.
\vspace{-2ex}
\nomenclature[A]{GHC}{Glasgow Haskell Compiler}

\subsection{A Taxonomy of Parametric Polymorphism}

\nomenclature[B]{$T$}{a type system in lattice~$𝔗$}
It follows from \cref{theorem:fluent} that \Fluent can be compiled into a type
system~$T$ only if~$T$ is (computationally wise) sufficiently
\emph{expressive}, i.e.,~$\|LR|⊆T$. But which are the \emph{features} of~$T$
that make it this expressive?
Motivated by questions such as theses, we offer in \cref{section:lattice} a
taxonomy, reminiscent of~$λ$-cube~\cite{Barendregt:91}, for the classification
of parametrically polymorphic type systems. The difference is that~$λ$-cube is
concerned with parametric polymorphism where the main type constructor is
function application; our taxonomy classifies type system built around the
pairing type constructor, as found in traditional imperative and object
oriented languages.

The taxonomy is a partially ordered set, specifically a lattice,~$𝔗$ of points
spanned by six, mostly orthogonal \emph{characteristics}. (See
\cref{table:lattice} below.) A point~$T∈𝔗$ is a combination of features (values
of a characteristic) that specify a type system, e.g., \Fluent is defined by
combination of three non-default features, \textsl{monadic-polymorphism},
\textsl{deep-argument-type}, and \textsl{rudi\-mentary-typeof} of three
characteristics; features of the three other characteristics take their default
(lowest) value, \textsl{linear-patterns}, \textsl{unary-functions}, and,
\textsl{one-type}.

We say that~$T₁$ is less \emph{potent} than~$T₂$, if~$T₁$ is strictly
smaller in the lattice order than~$T₂$, i.e., any program (including type
definitions and optionally an expression to type check) of~$T₁$ is also a
program of~$T₂$. In writing~$T₁=T₂$ ($T₁⊊T₂$) we mean that the computational
complexity class of~$T₁$ is the same as (strictly contained in) that of~$T₂$.
\vspace{0ex}

\paragraph{The~\PP Type System.}
\newcommand\GR{G＆R\xspace}
\nomenclature[A]{\GR}{Gil and Roth \citeyear{Gil:2019}}
\nomenclature[O]{\Fling}{a fluent API generator contributed by \GR}
We employ~$𝔗$ to analyze \Fling, yet another API
generator~\cite{Gil:2019} (henceforth \GR), capable of producing output for
Java and C＃. Although \Fling does not explicitly define an intermediate
language equivalent to \Fluent, type definitions produced by \Fling belong to a
very distinct fragment of type systems of Java and C＃, which \GR call
‟\emph{unbounded unspecialized parametric polymorphism}”, and we call
\PP†{read ‟plain parametric polymorphism”, or, ‟polyadic parametric
polymorphism” here} here.
\nomenclature[A]{\PP}{plain parametric polymorphism, or,
polyadic parametric polymorphism,~\eq{PP} and \cref{figure:PP}}

In plain words,~$\PP$ refers to a type system in which genericity occurs only in
no-parameters methods occurring in generic classes (or interfaces) that take
one or more unconstrained type arguments, as in, e.g., \cref{example:pp}%
†{Although not intended to be executable, Java (and {C++}) code in this
  examples can be copied and pasted as is (including Unicode characters such
  as~$γ$) into-, and then compiled on- contemporary IDEs. 
  Exceptions are expressions included for demonstrating
  type checking failure.%
  \nomenclature[A]{IDE}{interactive development environment}%
}.
In terms of lattice~$𝔗$, type system \PP is defined by feature
\textsl{polyadic-parametric-polymorphism} of the ‟\emph{number of type
arguments"} characteristics (and default, least-potent feature value of all
other characteristics). 

\begin{code}[caption={An example non-sense program in type system \PP},
  label={example:pp}]
class Program {¢¢
  // Type definitions
  interface¢¢ γ1 {} interface¢ ¢ γ2 {}
  interface¢¢ γ3<x1,x2> {¢¢ x1 a(); x2 b();γ4<γ2,γ3<x1,x2>> c(); }
  interface¢¢ γ4<x1,x2> {¢¢γ4<x2,x1> a();γ4<γ3<x1,x2>,γ3<x2,x1>> b();γ3<x1,x2> c(); }
  {¢¢// Initializer with expression(s) to check.
    ((γ3<γ1,γ2>) null).c().a().b().b().a(); // Type checks
    ((γ3<γ1,γ2>) null).c().a().b().b().a(); // Type check error
  }
}
\end{code}
We prove that~$\PP=\|DCFL|$ (\cref{theorem:pp:dcfl}),
i.e., computational complexity of~\PP is the same as \Fluent. Further, we will
see that type systems less potent than \PP reduce its computational complexity.
Other results (e.g., \cref{theorem:PP:deep} and \cref{theorem:PP:nonlinear})
show that making it more potent would have increased
its computational complexity.

\paragraph{Combining Theory and Practice: \Fling+\TypelevelLR architecture.}

As Yamazaki et al.~noticed, translating of \Fluent into mainstream programming
language is not immediate. Curiously, the type systems of all target languages
of \TypelevelLR are undecidable. However, it follows from \cref{theorem:fluent}
that the target language, from the theoretical complexity perspective, is only
required to be at least as expressive as DCFL, as is the case in language
such as Java, ML, (vanilla) Haskell, and C＃. 

\nomenclature[A]{ML}{the ML (‟meta-language”) programming language}

To bring theory into practice, notice that all these languages contain the \PP
type system. We envision a software artifact, the whose architecture combines
\TypelevelLR and \Fling, making it possible to compile of the variety of LR
grammars processors into any programming language which supports code such as
\cref{example:pp}. Front ends of this ‟ultimate fluent API generator” are the
same as \TypelevelLR. However, instead of directly translating \Fluent
introduce a (rather straightforward) implementation, e.g., in Java, of the
algorithm behind the proof of \cref{theorem:fluent}, plugging it as a back end
of \TypelevelLR. Concretely, the artifact compiles \Fluent into a
specification of a DPDA (deterministic pushdown automaton) as in the said
proof. We then invoke (a component of) \Fling to translate the DPDA
specification into a set of \PP type definitions. The back-ends of \Fling are
then employed to translate these type definitions to the desired target
programming language.

{%\small

\textbf{Outline.} \slshape
\Cref{section:automata} presents the lattice~$𝔄$ of finite control automata on
strings and trees, ranging from FSAs to Turing machines, and reminds the reader
of the computational complexity classes of automata in this lattice, e.g., in
terms of families of formal languages in the Chomsky hierarchy. The lattice of
parametrically polymorphic type systems~$𝔗$ is then presented in
\cref{section:lattice}. The presentation makes clear bisimulation between the
runs of certain automata and type checking in certain type systems, whereby
obtaining complexity results of these type systems. 

\Cref{section:polymorphism}, concentrating on parallels between real-time
automata and type systems, derives further complexity results.  In particular,
this section shows that \Fling is optimal in the sense that no wider class of
formal languages is used by \PP, the type system it uses. Non real-time
automata, and their relations to type systems which admit the \kk{typeof}
keywords are the subject of \cref{section:capture}. In particular, this section
sets the computational complexity of \Fluent and several variants. 
\Cref{section:overloading} then turns to discussing the ties between
non-deterministic automata and type systems that allow an expression
to have multiple types.

\Cref{section:zz} concludes in a discussion, open problems and directions for
future research. 

While reading this paper, readers should notice extensive overloading of notation,
made in attempt to highlight the ties between automata and type
systems.  The list of symbols in \cref{section:symbols} should help in
obtaining a full grasp of this overloading.
Appendices also include some of the more technical proofs and other
supplementary material.
}

\section{Finite Control Automata}
\label{section:automata}
This section presents a unifying framework of finite control automata and
formal languages, intended to establish common terminology and foundation for
the description in the forthcoming \cref{section:lattice} of parametrically
polymorphic type systems and their correspondence to automata.

Definitions here are largely self contained, but the discussion is brief; it is
tacitly assumed that the reader is familiar with fundamental concepts of
automata and formal languages, which we only re-present here.

\subsection{Characteristics of Finite Control Automata}

We think of automata of finite control as organized in a conceptual
lattice~$𝔄$. The lattice (strictly speaking, a \emph{Boolean algebrea}) is
spanned by seven (mostly) orthogonal \emph{characteristics}, such as the kind
of input that an automaton expects, the kind of auxiliary storage it may use,
etc. Overall, lattice~$𝔄$ includes automata ranging from finite state automata
to Turing machines, going through most automata studied in the classics of
automata and formal languages (see, e.g., Hopcroft, Motwani and Ullman
\citeyear{Hopcroft:Motwani:Ullman:07}).

Concretely, \cref{table:automata} defines lattice~$𝔄$, by row-wise enumeration
of the said characteristics and the values that each may take. We call these
values \emph{properties} of the lattice.
\ifjournal
In total, the seven characteristics
offer~$2+6+2+2+3+2+2=19$ different automata properties.
\fi

\begin{wraptable}{r}{38ex}
  \def\arraystretch{1.01}
  \scriptsize
  \rowcolors{2}{gray!25}{yellow!25}
  \begin{tabular}{>{\raggedright}p{19ex}>{\slshape}p{30ex}}
    \toprule
     \bfseries Characteristic & \bfseries Values in increasing potence ⏎
    \midrule
     No\@. states \newline (\cref{stateless}) & 
          \begin{celllist}
            \item stateless
            \item stateful
          \end{celllist} ⏎
     Aux\@. storage \newline (\cref{section:auxiliary:storage})
       & 
          \begin{celllist}
            \item no-store
            \item pushdown-store
            \item tree-store
            \item linearly-bounded-tape-store
% \item primitive-recursive-bounded-tape-store
            \item unbounded-tape-store
          \end{celllist} ⏎
           Recognizer kind (\cref{language:recognizer}, \cref{forest:recognizer}).
       & 
          \begin{celllist}
            \item language-recognizer
            \item forest-recognizer
          \end{celllist} ⏎
     $𝜀$-transitions \newline (\cref{real-time})
       & 
          \begin{celllist}
            \item real-time
            \item~$𝜀$-transitions
          \end{celllist} ⏎
     Determinism \newline (\cref{deterministic})
       & 
          \begin{celllist}
            \item deterministic
            \item deterministic-at-end
            \item non-deterministic
          \end{celllist} ⏎
     Rewrite multiplicity (\cref{linear-rewrite})
       & 
          \begin{celllist}
            \item linear-rewrite
            \item non-linear-rewrite
          \end{celllist} ⏎
      Rewrite depth \newline (\cref{shallow-rewrite})
       & 
          \begin{celllist}
            \item shallow-rewrite
            \item deep-rewrite
          \end{celllist} ⏎
    \bottomrule
  \end{tabular}
  \caption{Seven characteristics and 18 properties spanning lattice~$𝔄$ of finite control automata}
    \label{table:automata}
    \vspace{-4em}
\end{wraptable}

Values of a certain characteristics are mutually exclusive: For example, the
first row in the table states that the first characteristic, \emph{number of
states}, can be either \emph{stateless} (the finite control of the automata does
not include any internal states) or \emph{stateful} (finite control may depend
on and update an internal state). An automaton cannot be both stateful and
stateless.

An automaton in~$𝔄$ is specified by selecting a value for each of the
characteristics.
\ifjournal
Overall, there are~$2·6·2·2·3·2·2=576$ different combinations of properties.
However, not all combinations make lattice points. For example, rewrites only
make sense in the presence of auxiliary-storage, so, in the Boolean algebra~$𝔄$,
\begin{equation}\label{rewrites}
  \textsl{no-store}→\textsl{linear-rewrite} ∧ \textsl{shallow-rewrite}.
\end{equation}
We use prefixes of property names when space is scarce and no ambiguity can arise,
e.g., \[
  \textsl{no-store}→\textsl{linear}∧\textsl{shallow}
\] is a succinct version of \cref{rewrites}.
\else
\fi

\ifjournal Characteristics are listed in the table in somewhat arbitrary
order. However, the properties in each characteristic are enumerated in
increasing potence order.\else
The table enumerates properties in each
characteristic in increasing potence order.
\fi
For example, in ``number of states'' characteristic, \textsl{stateful} automata
are more \emph{potent} than stateless automata, in the sense that any
computation carried out by~$A$, a certain stateless
automaton 
in~$𝔄$, can be carried out by automaton~$A'∈𝔄$, where the only
difference between~$A$ and~$A'$ is that~$A'$ is \textsl{stateful}.

Each automaton in the lattice might be fully specified as a set of seven
properties \A{$p₁,…,p₇$}. In the abbreviated notation we use, a property of a
certain characteristic is mentioned only if it is not the weakest (least
potent) in this characteristic. For example, the notation ‟\A{}” is short for
the automaton with least-potent property of all characteristics,
\nomenclature[C]{$pᵢ$}{value of the lattice property~$i$}
\begin{equation}
  \begin{aligned}
    &𝔄_⊥=\A{}=⟨\textsl{stateless},
  \textsl{no-store},
  \textsl{language},
  \textsl{real-time},
  \textsl{determ},
  \textsl{shallow},
  \textsl{linear}⟩,
\end{aligned}
\nomenclature[B]{$𝔄_⊥$}{bottom of lattice~$𝔄$\nomrefeq}
\end{equation}
the \emph{bottom} of lattice~$𝔄$. \Cref{table:common:automata} offers
additional shorthand using acronyms of familiar kinds of automata and their
mapping to lattice points. For example, the second table row maps FSAs to
lattice point~$⟨\textsl{stateful}, \textsl{non-deter}⟩$.

\begin{table}[ht]
\def\B{\textbf}
\def\arraystretch{1.5}
\rowcolors{2}{gray!25}{yellow!25}
\scriptsize
\begin{threeparttable}
\begin{tabularx}\textwidth{p{32ex}lXl}
  \toprule
  \hiderowcolors
  \bfseries Acronym &
  \bfseries Common name &
  \bfseries Lattice point &
  \bfseries Complexity\tnote{1} ⏎
  \midrule
  \showrowcolors

  \B Deterministic \B Finite \B State \newline \B Automaton &
  DFSA &
  \A{stateful} &
  REG ⏎
   \B Finite \B State \B Automaton &
  FSA &
  $\A[DFSA]{non-deterministic}=\A{stateful, non-deterministic}$ &
  REG\tnote{2} ⏎

  \B Stateless \B Real-time \B Deterministic \newline \B Push\B Down \B Automaton &
  SRDPDA &
  $⟨\textsl{pushdown, stateless,~$𝜀$-transitions, non-deterministic,}$
  \newline\mbox\qquad~$\textsl{shallow, linear}⟩$ &
  $⊊$ DCFL ⏎

  \B Real-time \B Deterministic \newline \B Push\B Down \B Automaton &
  RDPDA &
  \A[SRDPDA]{stateful}=$⟨\textsl{pushdown, stateful,~$𝜀$-transitions,}$
  \newline\mbox\qquad~$\textsl{non-deterministic, shallow, linear}⟩$ &
  $⊊$ DCFL\tnote{3} ⏎

  \B Deterministic \B Push\B Down \newline \B Automaton &
  DPDA &
  $\A[RDPDA]{$𝜀$-transitions}=⟨\textsl{pushdown, stateful,~$𝜀$-transitions,}$
  \newline\mbox\qquad~$\textsl{deterministic, shallow, linear}⟩$ &
  DCFL\tnote{4} ⏎

  \B Tree \B Automaton &
  TA &
  \A{tree-store, stateless, real-time, shallow, linear} &
  DCFL ⏎

  \B Push\B Down \B Automaton & PDA &
  \A[DPDA]{non-deterministic}=$⟨\textsl{pushdown, stateful,~$𝜀$-transitions,}$
  \newline\mbox\qquad~$\textsl{non-deterministic, shallow, linear}⟩$ &
  CFL\tnote{5} ⏎

  \B Real-time \B Turing \B Machine &
  RTM &
  $⟨\textsl{linearly-bounded-tape, stateful, real-time, deterministic,}$
  \newline\mbox\qquad~$\textsl{shallow, linear}⟩$ &
  $⊊$ CSL ⏎

  \B Linear \B Bounded \B Automaton &
  LBA &
  \A[FSA]{linearly-bounded-tape, shallow, linear}=
  \newline\mbox\qquad~$⟨\textsl{linearly-bounded-tape, deterministic, shallow, linear}⟩$ &
  CSL ⏎

  %\B Primitive \B Recursive \B Bounded \B Automaton &
  %PRBA &
  %\A[LBA]{primitive-recursive-bounded-tape}=$⟨\textsl{primitive-recursive-}$
  %\newline\mbox\qquad~$\textsl{-bounded-tape, deterministic, shallow, linear}⟩$ &
  %PR ⏎

  \B Turing \B Machine &
  TM &
  \A[LBA]{unbounded-tape}=$⟨\textsl{unbounded-tape, stateful,}$
  \newline\mbox\qquad~$\textsl{deterministic, shallow, linear}⟩$ &
  RE\tnote{6} ⏎
  \bottomrule
\end{tabularx}
\begin{tablenotes}
\item[1]~$\|REG|⊊\|DCFL|⊊\|CFL|⊊\|CSL|⊊\|PR|⊊\|R|⊊\|RE|$
\item[2]~$\A[FSA]{$𝜀$-transitions, non-deterministic}=\|FSA|$
\item[3]~\cite[Example 5.3]{Autebert:1997}
\item[4]~$\A[DPDA]{deep}=\|DPDA|$
\item[5]~$\A[PDA]{deep}=\|PDA|$
\item[6]~$\A[TM]{non-deterministic}=\|TM|$
\end{tablenotes}
\end{threeparttable}
\caption{Selected automata in the lattice~$𝔄$ and their computational complexity classes}
\label{table:common:automata}
  \nomenclature[A]{DFSA}{deterministic finite state automaton}
  \nomenclature[A]{REG}{the set of regular languages}
  \nomenclature[A]{SRDPDA}{stateless real-time deterministic pushdown automaton}
  \nomenclature[A]{DCFL}{deterministic context free language}
  \nomenclature[A]{RDPDA}{real-time deterministic pushdown automaton}
  \nomenclature[A]{TA}{tree automaton, i.e., an automaton employing a tree store}
  \nomenclature[A]{CFL}{context free language}
  \nomenclature[A]{CSL}{context sensitive language}
  \nomenclature[A]{RTM}{real-time Turing machine}
  \nomenclature[A]{LBA}{linear bounded automaton}
  \vspace{-4ex}
\end{table}

Observe that just as the term \emph{pushdown automaton} refers to an automaton that
employs a pushdown auxiliary storage, we use the term
\emph{tree automaton} for an automaton that employs a tree auxiliary storage.
Some authors use the term for automata that receive hierarchical tree rather
than string as input. In our vocabulary, the distinction is found in
the \textsl{language-recognizer} vs\@. \textsl{forest-recognizer} properties of
the ‟recognizer kind” characteristic.

The final column of \cref{table:common:automata} also specifies the
computational complexity class of the automaton defined in the respective row.
In certain cases, this class is a set of formal languages found in the Chomsky
hierarchy. From the first two rows of the table we learn that even though DFSAs
are less potent than FSAs, they are able to recognize exactly the same set of
formal languages, namely the set of regular languages denoted REG. By writing,
e.g.,~$\|DFSA|=\|FSA|=\|REG|$, we employ the convention of identifying
an automaton in the lattice by its computational complexity class. We notice
that a less potent automaton is not necessarily computationally weaker.

\subsection{Finite Control Automata for Language Recognition}
\label{section:language:recognizer}
As usual, let~$Σ$ be a finite \emph{alphabet}, and let~$Σ^*$ denote the set of
all strings (also called words) over~$Σ$, including~$𝜀$, the empty string. A (formal)
\emph{language}~$ℓ$ is a (typically infinite) set of such strings, i.e.,~$ℓ⊆Σ^*$.
\nomenclature[G]{$Σ$}{finite alphabet of symbols}
\nomenclature[G]{$Σ^*$}{set of all strings (words) over~$Σ$, including the empty string}
\nomenclature[H]{$𝜀$}{the empty string}
\nomenclature[c]{$ℓ$}{formal language of strings,~$ℓ⊆Σ^*$}

\begin{definition}\label{language:recognizer}
A \emph{recognizer of language~$ℓ⊆Σ^*$} is a device that takes as
input a word~$w∈Σ^*$ and determines whether~$w∈ℓ$.
\end{definition}
\nomenclature[c]{$w$}{the input word to language recognizer}

Let~$A$ be a finite control automata for language recognition. (Automata for
recognizing forests are discussed below in \cref{section:forest}.) Then,~$A$ is
specified by four finitely described components: states, storage,
consuming transition function, and~$𝜀$-transition function:
\nomenclature[B]{$A$}{a finite control automaton}

\begin{enumerate}
  \item \emph{States.} The specification of these includes \1~a finite set~$Q$
    of \emph{internal states} (or \emph{states}), \2~a designated \emph{initial
      state}~$q₀∈Q$, and, \3~a set~$F⊆Q$ of \emph{accepting states}.
      \nomenclature[B]{$Q$}{the set of internal states of a finite control automaton}
      \nomenclature[c]{$q₀$}{the initial internal state of a finite control automaton}
      \nomenclature[B]{$F$}{the set of accepting states in a finite control automaton}

    \begin{definition}\label{stateless}
      $A$ is \textsl{stateful} if~$|Q|>1$; it is \textsl{stateless}
      if~$|Q|=1$, in which case~$F=Q=❴q₀❵$.
    \end{definition}

  \item \emph{Storage.} Unlike internal state, the amount of data in auxiliary
    storage is input dependent, hence unbounded. The pieces of information
    that can be stored is specified as a finite alphabet~$Γ$ of \emph{storage
    symbols}, which is not necessarily disjoint from~$Σ$.
    \nomenclature[G]{$Γ$}{alphabet of symbols used in auxiliary storage}

    The organization of these symbols depends on the \textsl{auxiliary-storage}
    characteristic of~$A$: In \textsl{pushdown-store} automata,~$Γ$ is known as
    the set of \emph{stack symbols}, and the storage layout is sequential. In
    \textsl{tree-store} automata, the organization is hierarchical and~$Γ$ is a
    ranked-alphabet. In tape automata,~$Γ$ is called the set of \emph{tape symbols},
    and they are laid out sequentially in a uni-directional tape.

    Let~$𝚪$ denote the set of possible contents of the auxiliary storage. In
    pushdown automata~$𝚪=Γ^*$; in tape automata, the storage contents includes
    the position of the head: Specifically, in \textsl{unbounded-tape-store}
    (employed by Turing machines),~$𝚪=ℕ×Γ^*$. We set~$𝚪=ℕ×Γ^*$ also
    for the case of linearly bounded automata.
\ifjournal
, and for Turing machine
    restricted to bounded a tape whose size does not increase faster than a
    primitive recursive function (\textsl{linearly-bounded-tape-store}
    automaton), even though a more precise definition is possible.
\fi
    For \textsl{tree-store} automata,~$𝚪=Γ^▵$, where~$Γ^▵$ is defined below as
    the set of trees whose internal nodes are drawn from~$Γ$.
    
    \nomenclature[G]{$𝚪$}{set of possible contents of auxiliary storage}
    \nomenclature[B]{$ℕ$}{the set of non-negative integers~$❴0,1,2,…❵$}

  \begin{definition}\label{ID} An \emph{\textbf Instantaneous \textbf
    Description} (ID, often denoted~$ι$) of~$A$ running on input word~$w∈Σ^*$
    includes three components: \1~a string~$u∈Σ^*$,~where~$u$ is a suffix
    of~$w$, specifying the \emph{remainder of input} to read; \2~the
    \emph{current state}~$q∈Q$, and, \3~$𝛄∈𝚪$, the current contents of the
    auxiliary storage.
    \nomenclature[c]{$q$}{a state of a finite control automaton}
  \end{definition}
  \nomenclature[c]{$u$}{the word denoting the remainder of input to a language recognizer}
  \nomenclature[H]{$ι$}{an instantaneous description of a finite control automaton, see~\cref{ID}}
  \nomenclature[H]{$𝛄$}{entire contents of auxiliary storage}

  The auxiliary storage is initialized by a designated value~$𝛄₀∈𝚪$. Any
  \emph{run} of~$A$ on input~$w∈Σ^*$ begins with ID~$ι₀=⟨w,q₀,𝛄₀⟩$, and then
  proceeds as dictated by the transitions functions.
  \nomenclature[H]{$ι₀$}{initial instantaneous description of a finite control automaton}
  \nomenclature[H]{$𝛄₀$}{initial contents of auxiliary storage}

    \begin{definition}\label{no-store}
      $A$ is \textsl{no-store} if~$|Γ|=0$, in which case~$𝚪$ is
      degenerate,~$𝚪=❴𝛄₀❵$.
    \end{definition}

    \item \emph{Consuming transition function.}
      Denoted by~$δ$, this \emph{partial}, possibly \emph{multi-valued} function,
      defines how~$A$ proceeds from a certain ID to a subsequent ID in response
      to a consumption of a single input letter.
    \begin{itemize}
        \item Function~$δ$ depends on \1~$σ∈Σ$, the \emph{current input
          symbol}, being the first letter of~$u$, i.e.,~$u=σu'$,~$u'∈Σ^*$
          \2~$q∈Q$ the current state, and, \3~$𝛄∈𝚪$,
          the current contents of the auxiliary storage.
          \nomenclature[H]{$σ$}{a letter in alphabet~$Σ$}
          \nomenclature[H]{$δ$}{the consuming transition function of a finite control automaton}
        \item Given these,~$δ$ returns a new internal state~$q'∈Q$ and the new
          storage contents~$𝛄'$ for the subsequent ID\@. The ‟remaining input”
          component of the subsequent ID is set to~$u'$.
    \end{itemize}

   \item \emph{$𝜀$-transition function.}
    A partial, multi-valued function~$ξ$ specifies how~$A$ moves from a certain
    ID to a subsequent ID, \emph{without} consuming any input. Function~$ξ$
    depends on the current state~$q∈Q$ and~$𝛄$, the storage's contents,
    \emph{but not} on the current input symbol. Just like~$δ$, function~$ξ$
    returns a new internal state~$q'∈Q$ and storage contents~$𝛄'$ for the
    subsequent ID\@. However, the remaining input component of IDs is unchanged
    by~$ξ$.
    \nomenclature[H]{$ξ$}{the~$𝜀$-transition function of a finite control automaton}

 \end{enumerate}

Automaton~$A$ \emph{accepts}~$w$ if there exists a run~$ι₀,ι₁,…,ιₘ$, that
begins with the initial ID~$ι₀=⟨w,q₀,𝛄₀⟩$ and ends with an ID~$ιₘ=⟨𝜀,q,α⟩$ in 
which all the input was consumed, the internal state~$q$ is accepting, i.e.,~$q∈F$,
and no further~$𝜀$-transitions are possible, i.e.,~$ξ(ιₘ)$ is not defined.

On each input letter, automaton~$A$ carries one transition defined by~$δ$,
followed by any number of~$𝜀$-transitions defined by~$ξ$, including none at
all. A real-time automaton is one which carries precisely one transition for
each input symbol.

\begin{definition}\label{real-time}
$A$ is \textsl{real-time} if there is no id~$ι$ for which~$ξ(ι)$ is defined.
\end{definition}

Real-time and non-real-time automata are, by the above definitions,
non-deterministic. Since both~$δ$ and~$ξ$ are multi-valued, an ID does not
uniquely determine the subsequent ID\@.

\begin{definition}\label{deterministic}
$A$ is \textsl{deterministic} if \1~partial functions~$ξ$ and~$δ$ are single
valued, and, \2~there is no ID~$ι$ for which both~$ξ(ι)$ and~$δ(ι)$ are
defined.
\end{definition}

Both deterministic and non-deterministic automata may \emph{hang}, i.e., they
might reach an ID~$ι$ for which neither~$ξ(ι)$ nor~$δ(ι)$ are defined. If
\emph{all} runs of a non-deterministic automaton~$A$ on a given input~$w$
either hang or reach a non-accepting state,~$A$ rejects~$w$. Alternatively, if
the \emph{the only run} of a deterministic automaton~$A$ on~$w$ hangs,
automaton~$A$ rejects~$w$. Hanging is the only way a \textsl{stateless}
automaton can reject. A stateful automaton rejects also in the case it reaches
a non-accepting state~$q∈Q∖F$ after consuming all input.

\subsection{Rewrites of Auxiliary Storage}
\label{section:auxiliary:storage}

Since functions~$δ$ and~$ξ$ are finitely described, they are specified as two
finite sets,~$Δ$ and~$Ξ$ of input-to-output items, e.g., the requirement in
\cref{real-time} can be written as~$|Ξ|=0$. Since the transformation of
auxiliary storage~$𝛄$ to~$𝛄'$ by these functions must be finitely described,
only a bounded portion of~$𝛄$ can be examined by~$A$.
The transformation~$𝛄$ to~$𝛄'$, what we call \emph{rewrite of auxiliary
storage}, must be finitely described in terms of this portion.
\nomenclature[G]{$Δ$}{set of input-output items of consuming transition function~$δ$}
\nomenclature[G]{$Ξ$}{set of input-output items of~$𝜀$-transition function~$ξ$}

\paragraph{Tape initialization, rewrite, and head overflow.}
The literature often defines tape automata with no consuming transitions,
by making the assumption that they receive their input on the tape store which
allows bidirectional movements. Our lattice~$𝔄$ specifies that the input
word~$w$ is consumed one letter at a time. No generality is lost, since with
the following definitions tape automaton~$A∈𝔘$ may begin its run by consuming
the input while copying it to the tape, and only then process it with as
many~$𝜀$-transitions are necessary.

The contents~$𝛄$ of tape auxiliary storage is a pair~$(h,γ₀γ₁⋯γ_{m-1})$, where
integer~$h≥0$ is the head's position and~$γ₀γ₁⋯γ_{m-1}∈Γ^*$ is the tape's
content. Let~$𝛄₀=(𝜀,0)$, i.e., the tape is initially empty and the head is at
location~$0$. Rewrites of tape are the standard read and replace of 
symbol-under-head, along with the move-left and move-right instructions to the
head: Tape rewrite~$γ→γ'₊$ (respectively, tape rewrite~$γ→γ₋'$) means that
if~$γₕ=γ$ then replace it with, not necessarily distinct, symbol~$γ'∈Γ$ and
increment (respectively, decrement)~$h$. A third kind of rewrite is~$⊥→γ$, which
means that if the current cell is undefined, i.e.,~$h∉❴0,…,m-1❵$, replace it
with~$γ∈Γ$.
\nomenclature[c]{$h$}{position of the read/write head on tape auxiliary storage}

The automaton hangs if~$h$ becomes negative, or if~$h$ exceeds~$n$, the input's
length, in the case of a linear bounded automaton.
\ifjournal
†{We omit the very elaborate
definition of head overflow in the case that~$A$ is
\textsl{primitive-recursive}.}
\fi
\nomenclature[c]{$n$}{length of word input to finite control automaton}

\paragraph{Rewrites of a pushdown.}
Rewrites of a pushdown auxiliary storage are the usual push and pop
operations; we will see that these can be regarded as tree rewrites.

\paragraph{Trees.}
A finite alphabet~$Γ$ is a \emph{signature} if each~$γ∈Γ$ is associated with a
integer~$r=r(γ)≥1$ (also called
\emph{arity})%
\ifjournal
†{To make the material more accessible, the definitions here are
  repeated, expanded and
exemplified in \cref{definitions:trees}.}. 
\fi.
A \emph{tree} over~$Γ$ is either a
\emph{leaf}, denoted by the symbol~$𝛜$, or a (finite) structure in the
form~$γ(𝐭)$, where~$γ∈Γ$ of arity~$r$ is the \emph{root} and~$𝐭=t₁,…,tᵣ$ is a
\emph{multi-tree}, i.e., a sequence of~$r$ (inductively constructed) trees
over~$Γ$. Let~$Γ^▵$ denote the set of trees over~$Γ$.
\nomenclature[G]{$Γ^▵$}{set of all trees over signature~$Γ$}
\nomenclature[H]{$𝛜$}{degenerate tree, also denoting a leaf in any
tree in~$Γ^▵$}
\nomenclature[c]{$r$}{rank/number of children in a node of a tree in~$Γ^▵$}
\nomenclature[c]{$r(γ)$}{rank of symbol~$γ$ drawn from a signature}

\ifjournal
For tree~$t∈Γ^▵$ let
\begin{equation}\label{tree:depth}
  \Depth(t)=\begin{cases}
    t=𝛜& 0⏎
    t=γ(t₁,…,tᵣ) & \displaystyle 1+\max_{1≤i≤r}❨\Depth(tᵢ)❩.
  \end{cases}
\end{equation}
\else
Let~$\Depth(t)$ be the depth of tree~$t$ (for leaves let,~$\Depth(𝛜)=0$).
\fi
We shall use freely a \emph{monadic tree abbreviation} by which
tree~$γ₂(γ₁(𝛜),γ₁(γ₁(𝛜)))$ is written as~$γ₂(γ₁,γ₁γ₁)$, and
tree~$γ₁(γ₂(⋯(γₙ(𝛜))))$ is written as~$γ₁γ₂⋯γₙ$. If the rank of all~$γ∈Γ$ is 1,
then~$Γ^▵$ is essentially the set~$Γ^*$, and every tree~$t∈Γ^▵$ can be viewed
as a stack whose top is the root of~$t$ and depth is~$\Depth(t)$.
\nomenclature[c]{$t$}{a tree in~$Γ^▵$}
\nomenclature[O]{$\Depth(t)$}{depth of tree~$t∈Γ^▵$,~$\Depth(𝛜)=0$}

In this perspective, a pushdown automaton is a tree automaton in which the
auxiliary tree is monadic. We set~$𝛄₀$ in tree automata to the leaf~$𝛜∈Γ^▵$,
i.e., the special case pushdown automaton starts with an empty stack.

\paragraph{Terms.}
Let~$X=❴x₁,x₂,…❵$ be an unbounded set of \emph{variables} disjoint to all
alphabets. Then, a \emph{pattern} (also called \emph{term})
over~$Γ$ is either some variable~$x∈X$ or a structure in the form~$γ(τ₁,…,τᵣ)$,
where the arity of~$γ∈Γ$ is~$r$ and each of~$τ₁,…,τᵣ$ is, recursively, a
term over~$Γ$. Let~$Γ⧋$ denote the set of terms over~$Γ$. Thus,~$Γ^▵⊊Γ⧋$,
i.e., all trees are terms. Trees are also called \emph{grounded terms};
\emph{ungrounded terms} are members of~$Γ⧋∖Γ^▵$. A term is \emph{linear} if
no~$x∈X$ occurs in it more than once, e.g., $γ(x,x)$ is not linear
while~$γ(x₁,γ(x₂,x₃))$ is linear,

\nomenclature[B]{$X$}{unbounded set of variables disjoint to all alphabets}
\nomenclature[c]{$x$}{variable used in a term}
\nomenclature[H]{$τ$}{term in set~$Γ⧋$}
\nomenclature[G]{$Γ⧋$}{set of all terms over signature~$Γ$}

\enlargethispage*{2\baselineskip}
\paragraph{Terms match trees.}
Atomic term~$x$ \emph{matches} all trees in~$Γ^▵$; a compound linear
term~$τ=γ(τ₁,…,τᵣ)$ matches tree~$t=γ(t₁,…,tᵣ)$ if for all~$i=1,…,r$,~$τᵢ$
recursively matches~$tᵢ$, e.g.,~$γ(x₁,γ(x₂,x₃))$ matches~$γ(𝛜,γ(𝛜,γ(𝛜,𝛜)))$. To
define matching of non-linear terms define \emph{tree substitution}~$s$
(substitution for short) as a mapping of variables to
terms,~$s=❴x₁→τ₁,…,xᵣ→τᵣ❵$. Substitution~$s$ is \emph{grounded} if all
terms~$τ₁,…,τᵣ$ are grounded. An \emph{application} of substitution~$s$ to
term~$τ$, denoted~$τ/s$, replaces each variable~$xᵢ$ with term~$τᵢ$ if and only
if~$xᵢ→τᵢ∈s$. The notation~$τ'⊑τ$ is to say that term~$τ$ \emph{matches} a
term~$τ'$, which happens if there exists a substitution~$s$ such that~$τ'=τ/s$.
\nomenclature[c]{$s$}{tree substitution~$❴x₁→τ₁,…,xᵣ→τᵣ❵$}

\paragraph{Tree rewrites.}
A \emph{tree rewrite rule}~$ρ$ (\emph{rewrite} for short) is a pair of two
terms written as~$ρ=τ₁→τ₂$. Rewrite~$ρ$ is \emph{applicable} to (typically
grounded) term~$τ₁'$ if~$τ₁'=τ₁/s$ for some substitution~$s$. If rewrite~$ρ$
matches term~$τ₁'$ then~$τ₁'/ρ$, the \emph{application} of~$ρ$ to~$τ₁'$ (also
written~$τ₁'/τ₁→τ₂$) yields the term~$τ₂'=τ₂/s$.
\nomenclature[H]{$ρ$}{tree rewrite rule}

The definition of rewrites does not exclude a rewrite~$γ₁(x₁)→γ₂(x₁,γ₁(x₂))$,
whose right-hand-side term introduces variables that do not occur in the
left-hand-side term. Applying such a rewrite to a tree will always convert it
to a term. Since the primary intention of rewrites is the manipulation of
trees, we tacitly assume here and henceforth that it is never the case; a
rewrite~$τ₁→τ₂$ is valid only if~$\Vars(τ₂)⊆\Vars(τ₁)$.
\nomenclature[O]{$\Vars(τ)$}{set of variables in term~$τ$}
\nomenclature[O]{$\Vars(ρ)$}{set of variables in rewrite~$ρ$}

Manipulation of tree and pushdown auxiliary storage is defined with rewrites.
For example, the rewrite~$γ₁(γ₂(x))→γ₃(x)$, or in abbreviated
form~$γ₁γ₂x→γ₃x$, is, in terms of stack operations: \emph{if the top of the
  stack is symbol~$γ₁$ followed by symbol~$γ₂$, then pop these two symbols and
  then push symbol~$γ₃$ onto the stack.}

With these definitions:
\begin{itemize}
  \item Each member of set~$Δ$ is in the form~$⟨σ,q,ρ,q'⟩$ meaning:
    if the current input symbol is~$σ$, the current state is~$q$ and
    auxiliary storage~$t$ matches~$ρ$, then, consume~$σ$, move to state~$q'$
    and set the storage to~$t/ρ$.

  \item Each member of set~$Ξ$ is in the form~$⟨q,ρ,q'⟩$ meaning: if the
    current state is~$q$ and auxiliary storage~$t$ matches~$ρ$, then, move to
    state~$q'$ and set the storage to~$t/ρ$.
\end{itemize}

A tree rewrite~$ρ=τ₁→τ₂$ is linear if~$τ₁$ is linear, e.g.,
rewrites~$γ(x)→γ'(x,x,x)$ and~$γ(x₁,x₂)→γ(γ(x₂,x₁),𝛜)$ are linear, but~$γ(x,x)→𝛜$
is not. Notice that rewrites of tape and pushdown auxiliary storage are linear:
the transition functions of these do never depend on the equality of two
tape or pushdown symbols.

\begin{definition}\label{linear-rewrite}
$A$ is \textsl{linear-rewrite} if all rewrites in~$Ξ$ and~$Δ$ are linear.
\end{definition}

\ifjournal
Let the depth of rewrite~$ρ=τ₁→τ₂$ be the depth of its left hand
side,~$\Depth(ρ)=\Depth(τ₁)$, where the depth of terms is defined by
generalizing \cref{tree:depth}.
\begin{equation}\label{term:depth}
  \Depth(t)=\begin{cases}
    τ=𝛜& 0⏎
    τ=x& 0⏎
    τ=γ(τ₁,…,τᵣ) & \displaystyle 1+\max_{1≤i≤r}❨\Depth(tᵢ)❩.
  \end{cases}
\end{equation}
\else
Let~$\Depth(ρ)$,~$ρ=τ₁→τ₂$, be~$\Depth(τ₁)$, and where the depth of terms is
defined like tree depth, a variable~$x∈X$ considered a leaf.
\fi
A term (rewrite) is \emph{shallow} if its depth is at most one,
e.g.,~$x$,~$γ(x)$, and~$γ(x,x)$ are shallow, while~$γ(γ(x))$ is not. Rewrite of
tape storage are shallow by definition, since only the symbol under the head is
inspected.
\nomenclature[O]{$\Depth(τ)$}{depth of term~$t∈Γ^▵$,~$\Depth(x)=0$}
\nomenclature[O]{$\Depth(ρ)$}{depth of pattern~$ρ∈Γ⧋$}
\begin{definition}\label{shallow-rewrite}
  $A$ is \textsl{shallow-rewrite} if all rewrites in~$Ξ$ and~$Δ$ are shallow.
\end{definition}

\subsection{Finite Control Automata for Forest Recognition}
\label{section:forest}

In the case that the set of input symbols~$Σ$ is a signature rather than a
plain alphabet, the input to a finite control automata is then a tree~$t∈Σ^▵$
rather than a plain word. We use the term \emph{forest} for what some call
\emph{tree language}, i.e., a (typically infinite) set of trees. Generalizing
\cref{language:recognizer} we define:

\begin{definition}\label{forest:recognizer}
A \emph{recognizer of forest}~$\text{\normalfont\pounds}⊆Σ^▵$ is a device that takes as
input a tree~$t∈Σ^▵$ and determines whether~$t∈\text{\normalfont\pounds}$.
\end{definition}

\nomenclature[B]{$\text{\normalfont\pounds}$}{forest, or language of trees,~$\text{\normalfont\pounds}⊆Σ^▵$}

As explained in \cref{section:language:recognizer} a
\textsl{language-recognizer} automaton scans the input left-to-right.
However, this order is not mandatory, and there is no essential difference
between left-to-right and right-to-left automata. This symmetry does not
necessarily apply to a \textsl{forest-recognizer} automaton---there is much
research work on comparing and differentiating bottom-up and top-down traversal
strategies of finite control automata (e.g., Coquidé et
al\@.~\citeyear{coquide1994bottom} focus on bottom-up automata, Guessarian
\citeyear{Guessarian:83} on top-down, while Comon et
al\@.~\citeyear{Comon:07} presents several cases in which the two traversal
strategies are equivalent.)

Our interest in parametrically polymorphic type systems sets the focus here on
the bottom-up traversal strategy only. Most of the description of
\textsl{language-recognizer} automata above in
\cref{section:language:recognizer} remains unchanged. The state and storage
specification are the same in the two kinds of recognizers, just as the
definitions of deterministic and real-time automata. Even the specification
of~$ξ$, the~$𝜀$-transition function is the same, since the automaton does not
change its position on the input tree during an~$𝜀$-transition.

However, input consumption in forest recognizers is different than in language
recognizers, and can be thought of as visitation. A bottom-up forest-recognizer
consumes an input tree node labeled~$σ$ of rank~$r$ by visiting it
after its~$r$ children were visited. Let~$q₁,q₂,…,qᵣ$ be the states of the automaton
in the visit to these children, and let~$𝐪$ be the \emph{multi-state} of
the~$r$ children, i.e.,~$𝐪=q₁,q₂,…,qᵣ$. Then, the definition of~$δ$ is modified
by letting it depend on multi-state~$𝐪∈Qᵏ$ rather than on a single
state~$q∈Q$. More precisely, each input-to-output item in~$Δ$ takes the
form~$⟨σ,𝐪,ρ,q'⟩$, meaning, \emph{if \1~the automaton is in a node
  labeled~$σ$, and \2~it has reached states~$q₁,q₂,…,qᵣ$ in the~$r$ children of
  this node, and if storage rewrite rule~$ρ$ is applicable, then select
  state~$q'$ for the current node and apply rewrite~$ρ$}.
  \nomenclature[c]{$𝐪$}{multi-state~$q₁,q₂,…,qᵣ$,~$r$ determined by context}

Consider~$ρ$, the rewrite component of an input-output item. As it turns out,
only tree auxiliary storage makes sense for bottom up forest
recognizers†{In top-down forest recognizers pushdown
auxiliary storage is also admissible.}. Let~$t₁,…,tᵣ$ be the trees representing
the contents of auxiliary storage in~$r$ children of the current node.
Rewrite rule~$ρ$ should produce a new tree~$t$ of unbounded size from
a finite inspection of the~$r$ trees, whose size is also unbounded.

\enlargethispage*{\baselineskip}
We say that~$ρ$ is a \emph{many-input tree rewrite rule} (for short,
\emph{rewrite} when the context is clear) if it is in the form~$ρ=τ₁,…,τᵣ→τ'$.
Rule~$ρ=τ₁,…,τᵣ→τ'$ is applied to all children, with the straightforward
generalization of the notions
of matching and applicability of a single-input-rewrite:
\begingroup
\addtolength\leftmargini{-.25ex}
\begin{quote}
  A \emph{multi-term}~$𝛕$ is a sequence of terms~$𝛕=τ₁,…,τᵣ$, and a
  \emph{multi-tree}~$𝐭$ is a sequence of trees,~$𝐭=t₁,…,tᵣ$. Then,
  rule~$ρ=𝛕→τ'$ \emph{applies to} (also, \emph{matches})~$𝐭$ if there is a
  single substitution~$s$ such that~$τᵢ/s=tᵢ$ for all~$i=1,…,r$. The
  \emph{application} of~$ρ$ to~$𝐭$ is~$τ/s$.
\end{quote}
\endgroup
\nomenclature[c]{$𝐭$}{multi-tree~$t₁,…,tᵣ$,~$r$ determined by context}
\nomenclature[H]{$𝛕$}{multi-term,~$τ₁,…,τᵣ$,~$r$ determined by context}

\section{Parametrically Polymorphic Type Systems}
\label{section:lattice}
This section offers a unifying framework for parametrically polymorphic type
systems. Definitions reuse notations and symbols introduced in
\cref{section:automata} in the definition of automata, but with different
meaning. For example, the Greek letter~$σ$ above denoted an input letter, but
will be used here to denote the name of a function defined in a certain type
system. This, and all other cases of overloading of notation are
\emph{intentional}, with the purpose of highlighting the correspondence between
the two unifying frameworks.

\subsection{The Lattice of Parametrically Polymorphic Type Systems}

\begin{wraptable}{r}{52ex}%
  \vspace{-1em}
  \scriptsize
  \rowcolors{2}{gray!25}{yellow!25}%
  \begin{tabular}{>{\raggedright}p{30ex}>{\slshape}p{37ex}}
    \toprule
    \bfseries Characteristic & \bfseries Values in increasing order ⏎
    \midrule
    $C₁$ Number of type arguments \newline
      (\cref{section:PP}) &
    \begin{celllist}
      \item nyladic-parametric-polymorphism
      \item monadic-parametric-polymorphism
      \item dyadic-parametric-polymorphism
      \item polyadic-parametric-polymorphism
    \end{celllist} ⏎
    $C₂$ Type pattern depth \newline
      (\cref{section:pattern:depth}) &
      \begin{celllist}
        \item shallow-type-pattern
        \item almost-shallow-type-pattern
        \item deep-type-pattern
      \end{celllist} ⏎
    $C₃$ Type pattern multiplicity \newline
      (\cref{section:pattern:multiplicity}) &
        \begin{celllist}
          \item linear-patterns
          \item non-linear
        \end{celllist} ⏎
  $C₄$ Arity of functions \newline
      (\cref{section:function:arity}) &
        \begin{celllist}
          \item unary-functions
          \item n-ary-functions
        \end{celllist} ⏎
 $C₅$ Type capturing \newline
      (\cref{section:type:capturing}) &
    \begin{celllist}
      \item no-typeof
        \item rudimentary-typeof
        \item full-typeof
    \end{celllist} ⏎
   $C₆$ Overloading \newline (\cref{section:function:overloading})
       & \begin{celllist}
            \item one-type
            \item eventually-one-type
            \item multiple-types
         \end{celllist} ⏎
    \bottomrule
  \end{tabular}
  \caption{Six characteristics and 17 properties spanning lattice~$𝔗$ of
     parametrically polymorphic type systems.}
  \label{table:lattice}
%  \vspace{-4em}
\end{wraptable}

\nomenclature[B]{$Cᵢ$}{a characteristic of lattice~$𝔗$, see~\cref{table:lattice}}
\nomenclature[B]{$C₁$}{number of type arguments (characteristic of lattice~$𝔗$), see~\cref{table:lattice}}
\nomenclature[B]{$C₂$}{type pattern depth (characteristic of lattice~$𝔗$), see~\cref{table:lattice}}
\nomenclature[B]{$C₃$}{type pattern multiplicity (characteristic of lattice~$𝔗$), see~\cref{table:lattice}}
\nomenclature[B]{$C₄$}{arity of functions (characteristic of lattice~$𝔗$), see~\cref{table:lattice}}
\nomenclature[B]{$C₅$}{type capturing (characteristic of lattice~$𝔗$), see~\cref{table:lattice}}
\nomenclature[B]{$C₆$}{overloading (characteristic of lattice~$𝔗$), see~\cref{table:lattice}}

Examine \Cref{table:lattice} describing~$𝔗$, the lattice (Boolean algebra) of
parametrically polymorphic type systems. This table is the equivalent of
\cref{table:automata} depicting~$𝔄$, the lattice of finitely controlled
automata. We use the terms \emph{potence}, \emph{characteristics}, and
\emph{properties} as before, just as the conventions of writing lattice points
and use of abbreviations.

\Cref{table:lattice} give means for economic specification of different
variations of parametrically polymorphic types systems. For example,
inspecting Yamazaki et al.'s work we see that the type system of the \Fluent
intermediate language is
\begin{equation}
  \label{eq:fluent}
  \Fluent=\A{monadic-parametric-polymorphism,deep-type-pattern,rudimentary},
\end{equation}
\enlargethispage*{2\baselineskip}
i.e., \1~it allows only one parameter generics, e.g.,%
†{For concreteness we exemplify abstract syntax with the concrete syntax of Java or {C++}.}
\begin{JAVA}
interface¢¢ γ1<x> {} interface¢¢ γ2<x> {} interface¢¢ γ3<x> {}
\end{JAVA}

\noindent\2 it allows generic functions to be defined for deeply nested
generic parameter type, such as
\begin{JAVA}
static <x>¢¢ γ1<γ2<γ3<x>>> f(γ3<γ2<x>> e) {return null;}
\end{JAVA}
and, \3~it allows in the definition of function return type, a \kk{typeof}
clause, but restricted to use only one function invocation, e.g.,%
†{Code rendered in distinctive color as in abuses the host language syntax
for the purpose of illustration.}
\begin{JAVA}[morekeywords=typeof,style=pseudo]
static <x>¢¢ typeof(f(e)) g(γ3<x> e) {return null;}
\end{JAVA}

In contrast, the type system used by, e.g., \GR, is simply
\begin{equation}
  \label{eq:PP}
  \PP=\A{polyadic-parametric-polymorphism}.
\end{equation}

The remainder of this section describes in detail the characteristics in
\cref{table:lattice}.\vspace{-1ex}

\subsubsection{Object Based Type System}
Type system \A{}, the bottom of~$𝔗$, also denoted~$𝔗_⊥$ models \emph{object
based} programming paradigm, i.e., a paradigm with objects and classes, but
without inheritance nor parametric polymorphism. A good approximation of the
paradigm is found in the Go programming language~\cite{Donovan:15}. The essence
of~$𝔗_⊥$ is demonstrated in this (pseudo Java syntax) example:
\begin{JAVA}[style=pseudo]
interface A {¢¢ B a(); void b(); ¢¢} interface B {¢¢ B b(); A a(); ¢¢} new A().a().b().b().a().b();
\end{JAVA}
The example shows \1~definitions of two \emph{classes}†{ignore the somewhat
idiosyncratic distinction between classes and interfaces},~\cc{A} and~\cc{B},
\2~\emph{methods} in different classes have the same name,  but different return type,
\3~an expression whose type correctness depends on these definitions.

\Cref{figure:bottom} presents the abstract syntax, notational conventions and
typing rules of~$𝔗_⊥$. The subsequent
description of type systems in~$𝔗$ is by additions and modifications to the
figure.

\begin{figure}[htb]
  \caption{The bottom of lattice~$𝔗$: the type system \A{} modeling the object-based paradigm}
  \label{figure:bottom}
  \scriptsize
  \begin{adjustbox}{minipage=\linewidth,bgcolor={RoyalBlue!20}}
  \begin{tabularx}\textwidth{>{\hsize=.45\hsize}X>{\hsize=.45\hsize}X>{\hsize=.75\hsize}X}
    $\qquad\quad\begin{aligned}
      P &::=Δ~e ⏎
      Δ &::=δ^* ⏎
      δ &::=σ:γ→γ' ⏎
        &::=σ:𝛜→γ' ⏎
      e &::=𝜀~|~e.σ |~σ(e)⏎
    \end{aligned}$
&
    $\begin{array}{c}
      \typing{Function ⏎ Application}
      {\infer{e.σ:t'}{\begin{aligned}&e:t ⏎[-3pt] &σ:t→t'\end{aligned}}} ⏎⏎
      \typing{One Type ⏎ Only}
      {\infer{e:⊥}{\begin{aligned}&e:t₁ ⏎[-3pt] &e:t₂ ⏎[-3pt] &t₁≠t₂\end{aligned}}}
    \end{array}$
    &
    \begin{tabular}{rm{30ex}} ⏎
      $P$ & Program ⏎
      \nomenclature[B]{$P$}{program (abstract syntax start symbol), see~\cref{figure:bottom}}
      $e$ & Expression ⏎
      \nomenclature[c]{$e$}{expression (abstract syntax category), see~\cref{figure:bottom}}
      $Δ$ & Set of function definitions ⏎
      \nomenclature[G]{$Δ$}{set of primary function definitions (abstract syntax category), see~\cref{figure:bottom}}
      $δ$ & A function definition ⏎
      \nomenclature[H]{$δ$}{a definition of primary function (abstract syntax category), see~\cref{figure:bottom}}
      $σ$ & Function name, drawn from alphabet~$Σ$⏎
      \nomenclature[H]{$σ$}{name of primary function (abstract syntax category), see~\cref{figure:bottom}}
      $γ$ & Class names, drawn from alphabet~$Γ$ disjoint to~$Σ$ ⏎
      \nomenclature[H]{$σ$}{class name (abstract syntax category), see~\cref{figure:bottom}}
      $t,t',t₁,t₂$ & Grounded (non-generic) types ⏎
      \nomenclature[c]{$t$}{grounded type (abstract syntax category), see~\cref{figure:bottom}}
      $𝛜$ & The unit type ⏎
      $𝜀$ & The single value of the unit type⏎
      \nomenclature[H]{$𝛜$}{the unit type (terminal of abstract syntax), see~\cref{figure:bottom}}
      \nomenclature[H]{$𝜀$}{the single value of the unit type (terminal of abstract syntax), see~\cref{figure:bottom}}
      $⊥$ & The error type ⏎
      \nomenclature[O]{$⊥$}{the error type (terminal of abstract syntax), see~\cref{figure:bottom}}
    \end{tabular}
 ⏎⏎[-1.5ex]
    \multicolumn1c{\bfseries(a) Abstract syntax} &
    \multicolumn1c{\bfseries(b) Typing rules} &
    \multicolumn1c{\bfseries(c) Variables and notations}
  \end{tabularx}
  \end{adjustbox}
\end{figure}

A type in~$𝔗_⊥$ is either drawn from~$Γ$, or is the designated bottom type~$𝛜$.
The atomic expression, bootstrapping expression~$e$, is denoted by~$𝜀$, and
its type is~$𝛜$.

The figure defines program~$P$ in \A{} as a set~$Δ$ of function definitions~$δ$
followed by an expression~$e$ to type check. For~$σ$ drawn from set~$Σ$ of
function names, and types names~$γ₁,γ₂$ drawn from set~$Γ$ of class names, we
can think of a function definition of the form~$σ:γ₁→γ₂$ as \emph{either}
\begin{itemize}
  \item a \emph{method} named~$σ$ in class~$γ₁$ taking no parameters and
    returning class~$γ₂$, \emph{or},
  \item an external \emph{function} taking a single parameter of type~$γ₁$,
    and returning a value of type~$γ₂$.
\end{itemize}
With the first perspective, the recursive description of expressions is the
Polish convention,~$e::=e.σ$, best suited for making APIs fluent. With the
latter perspective, this recursive definition should be made in prefix
notation, i.e.,~$e::=σ(e)$. \Cref{figure:bottom} uses both variants, and we
will use these interchangeably. Indeed, the distinction between methods and
functions is in our perspective only a syntactical matter.

The special case of a function taking the unit type as argument,~$σ:𝛜→γ$, can be
thought of as an instantiation of the return type, \cc{\kk{new}~$γ$}. The
function name,~$σ$, is not essential in this case, but is kept for consistency.
Also in the figure is the standard \textsc{Function Application} typing
rule. Overloading on the parameter type is intentionally allowed, i.e., methods
defined in different classes may use the same name. The \textsc{One Type Only}
rule excludes overloading based on the return type.
\vspace{-1ex}

\subsubsection{Plain Parametric polymorphism}\label{section:PP}

Let \PP be short for lattice point
\textsl{$⟨$polyadic-parametric-polymorphism$⟩$}, as demonstrated in
\cref{example:pp} above.
\PP is the type system behind LINQ%
†{\url{https://docs.microsoft.com/en-us/dotnet/api/system.linq}},
the first theoretical treatise of fluent API~\cite{Gil:Levy:2016},
\Fling and other fluent API generators, e.g., of~\cite{Xu:2010} and~\cite{Nakamaru:17}.

The definition of \PP relies on the definitions of trees, terms and rewrites in
\cref{section:auxiliary:storage}. Notice that in~$𝔗_⊥$, types were drawn from
set~$Γ$. In allowing generic types the type repertoire is extended to~$Γ^▵$,
the set of trees over signature~$Γ$. A type~$γ∈Γ$ of rank~$r≥1$ is a generic
with~$r$ type parameters; the only leaf, of rank~$0$, is the unit type~$𝛜$.
\PP also admits ‟terms”, i.e., trees including formal variables drawn from the
set~$Γ⧋$. We refer to terms of \PP as ‟\emph{ungrounded types}”; an ungrounded
type is also viewed in \PP as a \emph{type pattern} that typically match
‟grounded types” (trees in~$Γ^▵$), but can also be used for matching over
ungrounded types.

\Cref{figure:PP} summarizes the changes in \PP's definitions with respect to
those of~$𝔗_⊥$ in \cref{figure:bottom}.

\begin{figure}[ht]
  \caption{The type system \protect{\PP}}
  \label{figure:PP}
  \scriptsize
  \begin{adjustbox}{minipage=\linewidth,bgcolor={RoyalBlue!20}}
  \begin{tabularx}\textwidth{>{\hsize=.55\hsize}X>{\hsize=.55\hsize}X>{\hsize=.75\hsize}X}
    \multicolumn1c{\emph{(same as \cref{figure:bottom} (a) and…)}} &
    \multicolumn1c{\emph{(same as \cref{figure:bottom} (b) and…)}} &
    \multicolumn1c{\emph{(same as \cref{figure:bottom} (c) and…)}} ⏎⏎[-1.5ex]
    $\qquad\qquad
  \begin{aligned}
    δ &::=σ:γ(𝐱)→τ \gcomment{\text{term~$γ(𝐱)$ is linear}} ⏎
      &::=σ:𝛜→t ⏎
    𝐱 &::=x₁,…,xᵣ ⏎
    τ &::=γ(𝛕)~|~x~|~t ⏎
    𝛕 &::=τ₁,…,τᵣ ⏎
    t &::=γ(𝐭)~|~𝛜⏎
    𝐭 &::=t₁,…,tᵣ
  \end{aligned}
  $
&
  $\typing{Generic ⏎ Function ⏎ Application}
  {\infer{e.σ:τ'/s}{%
      \begin{aligned}
              &e:t ⏎[-3pt]
              &σ:τ→τ' ⏎[-3pt]
              &t=τ/s
      \end{aligned}}}$
  &
      \begin{tabular}{rm{30ex}}
        $τ,τ'$ &Type patterns, drawn from~$Γ⧋$ ⏎
        \nomenclature[H]{$τ$}{type pattern, i.e., ungrounded type (abstract syntax category), see~\cref{figure:PP}}
  $𝛕$ & Multi-pattern, i.e., a sequence of type patterns~$τ$ ⏎
        \nomenclature[H]{$𝛕$}{multi-type pattern, i.e., multi ungrounded type (abstract syntax category), see~\cref{figure:PP}}
    $x$ & Type variables, drawn from set~$X$ disjoint to all alphabets ⏎
        \nomenclature[c]{$x$}{type variable (abstract syntax category), see~\cref{figure:PP}}
  $𝐱$ & Multi-variable, i.e., a sequence of type variables ⏎
        \nomenclature[c]{$𝐱$}{multi-variable (abstract syntax category), see~\cref{figure:PP}}
 $s$ & Tree substitution ⏎
  \end{tabular}
 ⏎⏎[-1.5ex]
    \multicolumn1c{\bfseries(a) Abstract syntax} &
    \multicolumn1c{\bfseries(b) Typing rules} &
    \multicolumn1c{\bfseries(c) Variables and notations}
  \end{tabularx}
  \end{adjustbox}
\end{figure}

The main addition of \PP to~$𝔗_⊥$ is allowing function definition~$δ$ to take
also the form~$σ:γ(𝐱)→τ$, where~$𝐱=x₁,…,xᵣ$ here is a sequence of~$r$
distinct type variables:
\begin{itemize}
  \item The single parameter to functions is a multi-variable, yet shallow and
    linear, type pattern~$γ(𝐱)$. This requirement models the definition of
    methods in \cref{example:pp}, i.e., in generic classes with~$r$ independent
    type
    variables. The structure of this pattern implicitly models the Java/C＃
    decree (which is absent from {C++}) against specialization of generics for
    specific values of the parameters.

  \item Also, as demonstrated by \cref{example:pp},~$τ$, the return type of a
    function in this form, is a type pattern of any depth constructed from the
    variables that occur in~$𝐱$ but also from any other types in~$Γ$.
\end{itemize}

The figure also shows how the \textsc{Function Application} typing rule is
generalized by employing the notions of matching and tree substitution
from \cref{section:auxiliary:storage}.

The definition of a \textsl{dyadic-parametric-polymorphism}
type system adds to \cref{figure:PP} the requirement that~$r(γ)≤2$.
In \textsl{monadic-parametric-polymorphism}, used for fluent API generation by
Nakamaru et al\@.~\citeyear{Nakamaru:17} and Yamazaki et
al\@.~\citeyear{Yamazaki:2019}, the requirement becomes~$r(γ)=1$
which means abstract syntax rule~$t::=γ(t)$ instead of~$t::=γ(𝐭)$,~$τ::=γ(τ)$
instead of~$τ::=γ(𝛕)$, and~$δ::=σ:γ(x)→τ$ instead of~$δ::=σ:γ(𝐱)→τ$.\vspace{-1ex}

\subsubsection{Type Pattern depth}
\label{section:pattern:depth}
Java, C＃, {C++} and other languages allow definitions of generic functions
which are not methods. For example, static Java function~\cc{f} defined by
\begin{JAVA}
static <x1,x2,x3>¢ ¢γ1<γ2<x3,x2>,x1> f(γ1<x1,γ2<x2,x3>> e) {return null;}
\end{JAVA}
is applicable only if the type of its single argument matches the
deep type pattern~$γ₁(x₁,γ₂(x₂,x₃))$. The corresponding lattice
property is obtained by adding derivation rule
\begin{equation}\label{eq:deep}
   δ::=σ:τ→τ' \gcomment{\text{term~$τ$ is linear}}.
\end{equation}
along with the requirement that~$τ$ is linear to \cref{figure:PP}.

As we shall see, the \textsl{deep-type-pattern} property increases the
expressive power of \PP\@. However, the syntax of invoking generic, non-method
functions in contemporary languages breaks the elegance of fluent
API: Using functions instead of methods,~\eq{first} takes the more
awkward form
\begin{equation} \label{eq:function:style}
  \cc{end(a(a(b(a(new Begin())))))}.
\end{equation}
The syntactic overhead of the above ‟reverse fluent API” can be lessened with a 
change to the host language; the case for making the change can be made by
sorting out the expressive power added by the \textsl{deep} property.\vspace{-1ex}

\subsubsection{Type Pattern Multiplicity}
\label{section:pattern:multiplicity}
Recall the abstract syntax rule of~$δ$ in type system
\PP (\cref{figure:PP}),
\begin{equation}\label{eq:method:style}
  δ::=σ:γ(𝐱)→τ \gcomment{\text{term~$γ(𝐱)$ is linear}} ⏎
\end{equation}
The \textsl{deep-type-pattern} property generalized this abstract syntax rule by
allowing functions whose argument type is not restricted to the flat
form~$γ(𝐱)$. Another orthogonal dimension in which \eq{method:style} can be
generalized is by removing the constraint that ‟term~$γ(𝐱)$ is linear”, i.e.,
allowing non-linear type patterns. Such patterns make it possible to define
function~$σ:γ(x,~x)→x$ that type checks with expression parameter~$e:γ(t₁,t₂)$
if and only if~$t₁=t₂$. Noticing that~$t₁$ and~$t₂$ are trees whose size is
unbounded, and may even be exponential in the size of the program, we
understand why the term non-linear was coined. Non-linear type patterns may
coerce the type-checker into doing non-linear amount of work, e.g., the little
Java program in \cref{listing:s2} brings the Eclipse IDE and its command line
compiler \texttt{ecj} to their knees.

\begin{code}[
    %xrightmargin=8ex,
    caption={Java proram in type system~$S₂=\A{n-ary,deep,non-linear}$ requiring over five minutes of compilation time by \texttt{ecj}
      executing on contemporary hardware},
    label=listing:s2,
  ]
class S2 {¢¢
  interface ¢$ε$¢{}
  interface C<x1, x2>{C<C<x2, x1>, C<x1, x2>> f();}
  C<¢$ε$¢,¢$ε$¢> f() {return null;}
  <x> void ¢$γ$¢(x e1, x e2) {}
  {¢¢ ¢$γ$¢(f().f().f().f().f().f().f().f().f().f().f().f().f().f().f().f().f().f().f().f().f().f().f().f().f().f()
    .f().f().f().f().f().f(), f().f().f().f().f().f().f().f().f().f().f().f().f().f().f().f().f().f().f().f()
    .f().f().f().f().f().f().f().f().f().f().f().f());¢¢ }
}
\end{code}

\ifjournal
\begin{code}[caption={Non-linear typing example in Java},label={listing:non-linear}]
interface Cons<Car, Cdr> {¢¢ // Cons type with two type parameters
  Cons<Cons<Car, Cdr>, Cons<Car, Cdr>> f(); } // Multiply current type size
static Cons<__,__> f() {¢¢ return null; } // Initialize Cons type
static <X> void equals(X x1, X x2) {} // Ensure parameter types are equals
equals(f().f().¢…¢.f(), f().f().¢…¢.f()); // Exponential compilation time
\end{code}
\fi

\enlargethispage*{2\baselineskip}
Type system \A[\PP]{non-linear} is defined by replacing \eq{method:style} by
its relaxed version,
\begin{equation}\label{eq:non-linear:style}
  δ::=σ:γ(𝐱)→τ \gcomment{\text{term~$γ(𝐱)$ may be non-linear}}.
\end{equation}
Likewise, type system \A[\PP]{deep,non-linear} is obtained by replacing
\eq{deep} by the relaxed version,
\begin{equation}\label{eq:deep:non-linear}
   δ::=σ:τ→τ' \gcomment{\text{term~$τ$ may be non-linear}}.
   \vspace{-1ex}
\end{equation}

\subsubsection{Arity of functions}
\label{section:function:arity}
Yet a third orthogonal dimension of generalizing \eq{method:style} is the
number of arguments; so far,~$σ$ was thought of as unary function, i.e., either
as a nullary method that takes no explicit parameters, or a generic unary,
non-method function. The \textsl{n-ary-functions} property of polymorphic type
systems allows~binary, ternary, and in general~$n$-ary functions,~$n≥1$. The
details are in \cref{figure:function:arity}.

\begin{figure}[ht]
  \caption{The type system \A{n-ary-functions,deep}}
  \label{figure:function:arity}
  \scriptsize
  \begin{adjustbox}{minipage=\linewidth,bgcolor={RoyalBlue!20}}
    \hspace{-3ex}
    \begin{tabularx}\textwidth{>{\hsize=.5\hsize}X>{\hsize=.9\hsize}X>{\hsize=.5\hsize}X}
    \multicolumn1c{\emph{(same as \cref{figure:PP} (a) and…)}} &
    \multicolumn1c{\emph{(same as \cref{figure:PP} (b) and…)}} &
    \multicolumn1c{\emph{(same as \cref{figure:PP} (c) and…)}} ⏎⏎[-1.5ex]
        $\qquad\quad\begin{aligned}
          δ &::=σ:𝛕→τ \quad \gcomment{\Vars(τ)⊆\Vars(𝛕)} ⏎
          𝛕 &::=τ₁⨉τ₂⨉⋯⨉τᵣ ⏎
          e &::=𝜀~|~𝐞.σ~|~σ(𝐞) ⏎
          𝐞 &::=e₁,e₂,…,eᵣ ⏎
        \end{aligned}$
        &
    $\typing{Multiple ⏎ Arguments}{\infer{e₁,e₂,…,eᵣ.σ:τ/s}{%
        \begin{aligned}
          &e₁:t₁,e₂:t₂,…,eᵣ:tᵣ ⏎[-3pt]
          &σ:τ₁⨉τ₂⨉⋯⨉τᵣ→τ ⏎[-3pt]
          &t₁=τ₁/s \quad t₂=τ₂/s \quad…\quad tᵣ=τᵣ/s
        \end{aligned}}}$
        &\hspace{-3.5ex}
    \begin{tabular}{rm{20ex}}
      $e₁,e₂,…,eᵣ$ & Expressions ⏎
      $t₁,t₂,…,tᵣ$ & Grounded types ⏎
      $τ₁,τ₂…,τᵣ$ & Generic types ⏎
      $𝐞$ & Multi-expression, i.e., a sequence of expressions~$e₁,…,eᵣ$ ⏎
      \nomenclature[c]{$𝐞$}{multi-expression (abstract syntax category), see~\cref{figure:function:arity}}
    \end{tabular}
 ⏎⏎[-1.5ex]
    \multicolumn1c{\bfseries(a) Abstract syntax} &
    \multicolumn1c{\bfseries(b) Typing rules} &
    \multicolumn1c{\bfseries(c) Variables and notations}
  \end{tabularx}
  \end{adjustbox}
\end{figure}

Comparing the figure to \cref{figure:PP} above we notice the introducing of
notation~$𝐞$ for a sequence of expressions. With this notation, a call to
an~$n$-ary function can be written~$𝐞.σ$ (Polish, fluent API like, convention)
or as~$σ(𝐞)$ (traditional convention). As might be expected, the figure also
extends the function application typing rule to non-unary functions.

Note that languages embedded in \A[\PP]{n-ary} are no longer
languages of words, but rather \emph{forests}---languages of trees. Indeed, an
expression in \A[\PP]{n-ary} is a tree of method calls, and the
set~$Δ$ in an \A[\PP]{n-ary} program defines the set of tree-like
expressions that type-check against it.

\subsubsection{Type capturing}
\label{section:type:capturing}
A primary motivation for introducing keyword \kk{decltype} to {C++}, was
streamlining the definition of wrapper functions---functions whose return
type is the same as the wrapped function, e.g.,
\begin{JAVA}[language=c++,morekeywords={decltype}]
template<typename x>auto wrap(x e)->decltype(wrapee(e)) {/*$\color{comment}⋯$*/auto \$=wrapee(e);/*$\color{comment}⋯$*/return \$;}
\end{JAVA}
As it turns out, keyword \kk{decltype} dramatically changes the type system, by
bringing about the undesired effect that type checking is undecidable. The
predicament is due to the idiom of using the type of one function to declare
the return type of another. Alternative, seemingly weaker techniques for
piecemeal definition of the return type, e.g., allowing \kk{typedef}s
in classes do not alleviate the problem. Likewise, the idiom is possible even
with the seemingly weaker feature, of allowing functions whose return type 
is \kk{auto}, as in
\begin{JAVA}[language=c++,morekeywords={decltype}]
template<typename x> auto wrap(x e){return wrappee(e);}
\end{JAVA}
Note that neither Java nor C＃ permit \kk{auto} functions; it appears
that the designers of the languages made a specific effort to block
loopholes that permit piecemeal definition of functions' return type.

\Cref{figure:type:capturing} presents abstract modeling of
C++'s \kk{decltype}; for readability we use the more familiar
\kk{typeof} keyword. The figure describes
\textsl{n-ary-functions}; for \textsl{unary-functions}
let~$n=1$.

\begin{figure}[ht]
  \caption{Type system \A{full-typeof,deep,n-ary-functions}}
  \label{figure:type:capturing}
  \scriptsize
  \begin{adjustbox}{minipage=\linewidth,bgcolor={RoyalBlue!20}}
  \begin{tabularx}\textwidth{>{\hsize=.4\hsize}X>{\hsize=.75\hsize}X>{\hsize=.5\hsize}X}
    \multicolumn1c{\emph{(same as \cref{figure:function:arity} (a) and…)}} &
    \multicolumn1c{\emph{(same as \cref{figure:function:arity} (b) and…)}} &
    \multicolumn1c{\emph{(same as \cref{figure:function:arity} (c) and…)}} ⏎⏎[-1.5ex]
    $\quad\begin{aligned}
      P &::=Δ~Ξ~e ⏎
      Ξ &::=ξ^* ⏎
      ξ &::=φ:𝛕→\kk{\scriptsize typeof}~ϑ ⏎[-0.3\baselineskip]
        &::=φ:𝛕→τ \gcomment{\Vars(ϑ)⊆\Vars(𝛕)} ⏎
      δ &::=σ:𝛕→\kk{\scriptsize typeof}~ϑ \gcomment{\Vars(ϑ)⊆\Vars(𝛕)} ⏎
      ϑ &::=𝛝.φ~|~𝛝.δ~|~τ ⏎
      𝛝 &::=ϑ₁,…,ϑᵣ ⏎
    \end{aligned}$
    &
   \begin{tabular}{l}
      $\typing{Typeof ⏎ Expression}
      {\infer{e₁,…,eᵣ.f:t}{%
      \begin{aligned}\renewcommand\baselinestretch{0.3}
        & f=σ \textit{\quad or\quad} f=φ⏎[-3pt]
        & f:τ₁⨉⋯⨉τᵣ→\text{\kk{\scriptsize typeof}}~ϑ ⏎[-3pt]
        & e₁:t₁\quad⋯\quad eᵣ:tᵣ ⏎[-3pt]
        & t₁=τ₁/s \quad⋯\quad tᵣ=τᵣ/s ⏎[-3pt]
        & ϑ/s:t ⏎
      \end{aligned}}}$ ⏎[10ex]
      \begin{minipage}{50ex}
        The \textsc{Multiple Arguments} typing rule of
      \cref{figure:function:arity} is also generalized for auxiliary functions ($φ$).
    \end{minipage}
    \end{tabular}
    $\begin{aligned}
    \end{aligned}$
    &
    \begin{tabular}{rm{30ex}}
      $Ξ$ & Set of auxiliary function definitions, used only in \kk{\scriptsize typeof} clause ⏎
      \nomenclature[G]{$Ξ$}{set of auxiliary function definitions (abstract syntax category), see~\cref{figure:type:capturing}}
      $ξ$ & An auxiliary function definition ⏎
      \nomenclature[H]{$ξ$}{auxiliary function definition (abstract syntax category), see~\cref{figure:type:capturing}}
      $φ$ & Auxiliary function names, drawn from alphabet~$Φ$ disjoint to~$Σ$ ⏎
      \nomenclature[G]{$Φ$}{set of auxiliary function names, disjoint to~$Σ$, see~\cref{figure:type:capturing}}
      \nomenclature[H]{$ϕ$}{auxiliary function name, drawn from set~$Φ$ (abstract syntax category), see~\cref{figure:type:capturing}}
      $ϑ$ & Pseudo expression, an expression whose type is not grounded ⏎
      \nomenclature[H]{$ϑ$}{pseudo expression, an expression whose type is ungrounded
        (abstract syntax category), see~\cref{figure:type:capturing}}
      $𝛝$ & Sequence of pseudo-expressions ⏎
      \nomenclature[H]{$𝛝$}{multi-pseudo expression (abstract syntax category), see~\cref{figure:type:capturing}}
    \end{tabular}
 ⏎⏎[-1.5ex]
    \multicolumn1c{\bfseries(a) Abstract syntax} &
    \multicolumn1c{\bfseries(b) Typing rules} &
    \multicolumn1c{\bfseries(c) Variables and notations}
  \end{tabularx}
  \end{adjustbox}
\end{figure}

The figure uses two syntactical categories for defining functions:~$δ∈Δ$,
which as before, defines a function named~$σ∈Σ$ that may occur in
expression~$e$ (more generally~$𝐞$); the similarly structured~$ξ∈Ξ$ uses 
distinct namespace~$φ∈Φ$ is for functions that may occur in a \kk{typeof}
clause.
\vspace{-1ex}

\enlargethispage*{2\baselineskip}
\paragraph{Pseudo-expressions.} Compare~$𝛕→\kk{typeof}~ϑ$ (the format of a
definition of function named~$σ$ in the figure) with~$𝛕→τ$ (the format of this
definition in \textsl{n-ary-function} type system
(\cref{figure:function:arity})). Without type capturing,~$σ$'s return type is
determined by a tree rewrite of the argument type(s). With type capturing, the
return type is determined by subjecting type~$τ$ to other function(s). To see
this, expand the recursive abstract syntax definition of~$ϑ$, assuming for
simplicity that~$n=1$,
\vspace{-1ex}
\begin{equation}\label{eq:typeof}
  δ::=σ:𝛕→\kk{typeof}~τ.φ₁.⋯.φᵣ,
  \vspace{-1ex}
\end{equation}
i.e., the pseudo-expression~$ϑ$ in this case is~$ϑ=τ.φ₁.⋯.φᵣ$. If~$n>1$ the
return type of a function defined with \kk{typeof} is specified by hierarchical
structure~$ϑ$, for which the figure coins the term \emph{pseudo-expression}.
Notice that a plain expression is a tree whose leaves (type instantiations)
are drawn from~$Γ$ and internal nodes (function calls) are drawn from~$Σ$.
Pseudo expressions are more general in allowing type variables in their leaves.
As emphasized in the figure, these variables must be drawn from~$𝛕$, the
multi-pattern defining the types of arguments to~$σ$.

A \textsl{full-typeof} type system allows any number of function calls in
pseudo-expression~$ϑ$, as in \eq{typeof}. In contrast, a
\textsl{rudimentary-typeof} type system allows at most one function symbol in
pseudo-expressions. This restriction is obtained by replacing the abstract
syntax rule for~$ϑ$ in \cref{figure:type:capturing} with a simpler,
non-recursive variant,~$ϑ::=𝛕.σ~|~τ$.

To describe the semantics of \kk{typeof}, we need to extend the notion of tree
substitution to pseudo-expressions as well. The application of function~$σ$ of
\eq{typeof} to a multi-expression~$𝐞$ with multi-type~$𝐭$ requires first
that~$𝐭⊑𝛕$, where the matching uses a grounded substitution~$s$. Then,~$ϑ/s$,
the application of~$s$ to pseudo-expression~$ϑ$ is the plain-expression
obtained by replacing the type variables in~$ϑ$ with the ground types defined
by~$s$.

\textsc{Typeof Expression} typing rule employs this notion as follows: typing
expression~$𝐞.σ$ with function~$σ:𝛕→\kk{typeof}~ϑ$ and arguments~$𝐞:𝐭$, we
\1~match the argument types with the parameter types,~$𝐭=𝛕/s$, deducing
substitution~$s$, \2~type~$ϑ/s:t$ (using an appropriate typing rule), and
finally \3~type~$𝐞.σ:t$. As an application of the \textsc{Type
of Expression} rule requires an additional typing, of~$ϑ$, its definition is
recursive.
\vspace{-0.5\baselineskip}

\subsubsection{Overloading}
\label{section:function:overloading}
The \emph{one-type} property means that expressions must have exactly one type
(as defined in \cref{figure:bottom}). With the more potent, \emph{multi-type}
property, expressions are allowed multiple types, by disposing the \textsc{One 
Type Only} type inference rule of \cref{figure:bottom}. With
\emph{multi-type-overloading}, expressions are allowed multiple types. With
\emph{eventually-one-type}, the semantics of
the Ada programming language~\cite{Persch:1980} apply:
Sub-expressions are allowed to have multiple types. However, upper level
expressions are still required to be singly typed. For example, while the upper
level expression~$e=\cc{$σ₃$($σ₂$($σ₁$()))}$ can be assigned at most one type,
both \cc{$σ₁$()} and \cc{$σ₂$($σ₁$())} may have many types.
\enlargethispage*{1.5\baselineskip}
\vspace{-0.5\baselineskip}

\subsection{Bisimulation of Automata and Type Systems}
The notation used in this section highlight ties between tree automata and type
systems, e.g., a tree~$t=γ₁(γ₂(γ₃),γ₄)$ can be understood as an instantiated
generic type, \lstinline/γ1<γ2<γ3>,γ4>/, to use Java syntax. Likewise the tree
rewrite~$ρ=γ₁(γ₂(x₁), x₂)→γ₂(x₂)$ can be interpreted as a Java function
\lstinline/static<x1,x2>γ2<x2>foo(γ1<γ2<x1>,x2>e){}/. Applying~$ρ$ to~$t$
yields the tree~$γ₂(γ₄)$, while the return type of the invocation
\lstinline/foo(new γ1<γ2<γ3>,γ4>())/ is \lstinline/γ2<γ4>/.

In fact, with the above definitions of type systems and finite control
automata, we can now easily pair certain automata with type systems.

\vspace{1ex}
\begin{mdframed}[align=center,userdefinedwidth=.87\textwidth,innertopmargin=0ex,
innerbottommargin=0ex,innerleftmargin=1ex,backgroundcolor=green!20,linecolor=lighter-grey,nobreak=true]
\begin{observation}\label{observation:correspondence}
  \leavevmode\vspace{1ex}
  \color{black}
  \begin{enumerate}
    \begin{minipage}{16ex}
    \item~$𝔗_⊥=\|FSA|$
    \item~$\PP=\|TA|$
    \end{minipage}
    \begin{minipage}{36ex}
    \item~$\A[\PP]{deep}=\A[TA]{deep}$
    \item~$\A[\PP]{non-linear}=\A[TA]{non-linear}$
    \end{minipage}
    \begin{minipage}{26ex}
    \item~$\A{monadic}=\|SRDPDA|$ \newline
    \end{minipage}
  \end{enumerate}
  %with the equality being understood in terms of computational classes.
\end{observation}
\end{mdframed}

To be convinced, notice the natural bisimulation of automata and type system,
obtained by a one-to-one correspondence between, e.g.,
\begin{itemize}
  \item a \emph{run} of an automaton and the \emph{type checking process} as
    dictated by the type checking rules,
  \item the \emph{hanging} of an automaton, and \emph{failure of type checking},
  \item the \emph{input} word or tree, and the type-checked \emph{expression},
  \item \emph{input-output items} in~$Δ_A$ and \emph{function definitions} in~$Δ_T$, 
  \item the \emph{contents} of auxiliary storage, and the \emph{type} of checked
        expression.
\end{itemize}

Observe however that states of an automaton do not easily find their parallel
in the typing world (except for~$𝔗_⊥=\|FSA|$, in which classes correspond to
states). Luckily, the expressive power of many of the automata we deal with
does not depend on the presence of states, e.g., it is easy to see
that~$\A[TA]{deep}=\A[TA]{deep,stateful}$.
%\input{zz.tex}
%\end{document}

\section{Parametric Polymorphism and Real-time Automata}
\label{section:polymorphism}
%\subsection{Polyadic Parametric Polymorphism}
The following result employs the type-automata correspondence to characterize
the complexity class of type system \PP.

\begin{theorem}\label{theorem:pp:dcfl}
$\PP=\|TA|=\|DCFL|$
\end{theorem}

Recalling the equivalence~$\PP=\|TA|$ (\cref{observation:correspondence}), the
gist of the theorem is the claim~$\|TA|=\|DCFL|$.
Towards the proof we draw attention to \GR's ‟\emph{tree encoding}”, which is
essentially a reduction by which every DPDA is converted to an equivalent tree automaton.
Their work then proceeds to show how this tree automaton is emulated in the \PP type system
they use (and that the emulation does not incur exponential space (and time)
overhead).
Hence, by \GR~\citeyear{Gil:2019},
\vspace{-0.3\baselineskip}
\begin{equation}\label{ECOOP}
 \|DCFL|=\|DPDA|⊆\|TA|=\PP.
\vspace{-0.3\baselineskip}
\end{equation}
A similar result is described by Guessarian~\citeyear{Guessarian:83}.
In fact, we note that Guessarian's contribution is more general, specifically she
achieves the result that augmenting tree automata with~$𝜀$-transitions and multiple
states does not increase their computational class.

\begin{fact}[\protect{\cite[Corollary 1.\1]{Guessarian:83}}]\label{Guessarian:trees}
  $\A[TA]{$𝜀$-transitions, stateful}=\|TA|$
\end{fact}
\Cref{Guessarian:trees} generalizes \cref{ECOOP}, since DPDAs are instances of
\A[TA]{$𝜀$-transitions, stateful}, where the tree store is linear.  The proof
of \cref{theorem:pp:dcfl} is completed by showing the inverse of \cref{ECOOP}.

\begin{lemma}\label{lower:bound}
$\|TA|⊆\|DPDA|$.
\end{lemma}
\vspace{-1.5ex}
\begin{proof}
The proof is constructed by employing Theorem~3 of
Guessarian~\citeyear{Guessarian:83}. (Notice that she uses the
term ‟\emph{pushdown tree automaton}” (PDTA) for \emph{top-down} tree-automata.
However, for the purpose of the reduction, we concentrate on input trees that
are in the form of a string, i.e., the tree traversal order is
immaterial.) 
\nomenclature[A]{PDTA}{pushdown tree automaton}
%An independent, more accessible proof is given in \cref{section:lower:bound}
\end{proof}
\vspace{-0.5\baselineskip}

Observe that \cref{lower:bound} means that \GR's result is the best
possible in the following sense: It is impossible to extend \Fling
to support any wider family of fluent API languages within the limits of the
fragment of the Java type system that \Fling uses. Moreover, as shown by
Grigore~\citeyear{Grigore:2017}, allowing the fluent API generator
a larger type system fragment, makes type-checking undecidable if the larger
fragment includes the common Java idiom of employing \kk{super} in signatures,
as in e.g., method
\lstinline/boolean removeIf(Predicate<? super E> filter)/
found in the standard \lstinline/java.util.Collection/ class. %TODO is Grigore's result relevant here?

%\subsection{Deep Patterns}
Combining \cref{observation:correspondence}~(5),
known results (\cref{table:common:automata}) and \cref{theorem:pp:dcfl}, we have
\begin{equation}\label{eq:monadic:polyadic}
  \A{monadic}=\|SRDPDA|⊊\|DCFL|=\PP=\A{polyadic},
\end{equation}
i.e., had \PP been weakened to allow only monadic generics, its expressive
power would have been reduced. Conversely, we would like to check the changes
to complexity when \PP is made more potent. Consider now allowing \emph{generic
functions} (on top of methods of generic classes) by adding the
\textsl{deep-type-pattern} feature to \PP.

\begin{theorem}\label{theorem:PP:deep}
$\|DCFL|⊊\A[TA]{deep}=\A[\PP]{deep}$
\end{theorem}

Again, recall the equivalence~$\A[\PP]{deep}=\A[\|TA|]{deep}$ from
\cref{observation:correspondence}. The set
containment,~$\|DCFL|⊆\A[TA]{deep}$ follows from \cref{ECOOP}. It
remains to show that this containment is proper.

\nomenclature[C]{$c$}{an example letter in alphabet}
\ifjournal
The proof of \Cref{theorem:PP:deep} in \cref{section:PP:deep} is
by encoding the context sensitive language~$aⁿbⁿcⁿ⊆❴a,b,c❵^*$ in type system
\A[\PP]{deep}.
\nomenclature[C]{$c$}{an example letter in alphabet}

The following definition is pertinent to the proof.
\begin{definition}
  For an integer~$k≥0$, let~$Uₖ$, let the \emph{unary type encoding of~$k$} be
  a grounded type in the type system \PP defined by
  \begin{equation*}
    Uₖ=\begin{cases}
      k=0 & \cc{Zero} ⏎
      k>0 & \cc{Succ<$u_{k-1}$>}.
    \end{cases}
  \end{equation*}
  where types \lstinline/Zero/ and \lstinline/Succ/, are, in Java syntax,
  \lstinline/interface Zero{}/ and \lstinline/interface Succ<T>{}/.
\end{definition}
Thus,~$U₀=\cc{Zero}$,~$U₁=\cc{Succ<Zero>}$,~$U₂=\cc{Succ<Succ<Zero>{}>}$, etc.
The sequel tacitly assumes these type definitions.
\else
\nomenclature[C]{$c$}{an example letter in alphabet}
The proof of \Cref{theorem:PP:deep} in \cref{section:PP:deep} is
by encoding the context sensitive language~$aⁿbⁿcⁿ⊆❴a,b,c❵^*$ in type system
\A[\PP]{deep}, and relying on the following definition: For an integer~$k≥0$,
let~$Uₖ$, the \emph{unary type encoding of~$k$}, be a grounded type in
\PP,~$U₀=\cc{Zero}$, and~$Uₖ=\cc{Succ<}u_{k-1}\cc>$,
with types (in Java syntax) \lstinline/interface Zero{}/ and \lstinline/interface Succ<T>{}/
(assumed implicitly henceforth).
Thus,~$U₀=\cc{Zero}$,~$U₁=\cc{Succ<Zero>}$,~$U₂=\cc{Succ<Succ<Zero>{}>}$, etc.
\fi
\nomenclature[B]{$Uₖ$}{unary type encoding of integer~$k∈ℕ$, defined recursively}
\nomenclature[B]{$U₀$}{unary type encoding of 0, base of the~$Uₖ$ recursion}

Note that in type system \PP it is possible to \emph{increment} and \emph{decrement} integers,
\vspace{-4ex}
\begin{JAVAc}
static Zero zero() {¢¢ return null; }
interface Zero {¢¢
  Succ<Zero> inc();
}¢\columnbreak\leavevmode\newline\vspace{-3pt}¢
interface Succ<T> {¢¢
  Succ<Succ<T>> inc();
  T dec();
}
\end{JAVAc}\vspace{-3ex}
We have, e.g., that the type of expression \lstinline/zero().inc()/\lstinline/.inc().inc().inc().dec().inc()/
is~$U₄$.

%\subsection{Non-Linear Patterns}
We now show that just like deep patterns, non-linear patterns, i.e., patterns
in which the same type variable occurs more than once, increase the
computational power of \PP.

This increase is attributed here to the ability of non-linear patterns to
\emph{compare} nested types, in particular types that are unary encoding of the
integers. For example, the Java generic function
\lstinline/static<x>void equal(x e1,x e2){}/
(in type system \A[\PP]{non-linear-patterns,n-ary-functions}) type-checks if
the types of its arguments are (say)~$U₉$ and~$U₉$,
and does not type check if these are (say)~$U₈$ and~$U₇$. More importantly,
type comparison is also possible if all functions are unary.

Consider, e.g., the two argument generic type~$γ$,
\lstinline/interface/ \lstinline/γ<x1,x2>{}/, and the generic unary function
\cc{equal}, \lstinline/static<x>void/ \lstinline/equal(γ<x,x>e){}/,
in type system \A[\PP]{non-linear}. Then, function \cc{equal}
\emph{type-checks} if the type of its single argument is
(say)~$γ\cc<U₉\cc,U₉\cc>$, and \emph{does not type-check} if this type is
(say)~$γ\cc<U₇\cc,U₈\cc>$.  With this observation, we can state.

\begin{theorem}\label{theorem:PP:nonlinear}
$\|DCFL|⊊\A[TA]{non-linear}=\A[\PP]{non-linear}
$\end{theorem}

The proof of \cref{theorem:PP:nonlinear} in \cref{section:PP:nonlinear} is
again by encoding the context sensitive language~$aⁿbⁿcⁿ⊆❴a,b,c❵^*$, but this
time in type system \A[\PP]{non-linear}. The ability of this type system to
compare integers encoded as types is the gist of the proof.

%\subsection{Dyadic Parametric Polymorphism}

Recall that a type system is \emph{dyadic} if no generic takes more than two
type parameters. Considering the \textsl{shallow} case, we claim no more than
placing \textsl{dyadic} between \textsl{monadic} and \textsl{polyadic} in
\eq{monadic:polyadic},
\vspace{-0.5ex}
\begin{equation}\label{eq:dyadic}
  \A{monadic}⊆\A{dyadic}⊆\A{polyadic},
\vspace{-0.5ex}
\end{equation}
although we conjecture~$\A{monadic}⊊\A{dyadic}$ can be shown relatively easily.
In contrast, in \textsl{deep} type system, the expressive power does not
increase by allowing more than two generic parameters.

\begin{theorem}\label{theorem:dyadic}
  \A{deep,polyadic}=\A{deep,dyadic}
\end{theorem}
\vspace{-1.5ex}
\begin{proof} (sketch)
Relying on the automata-type correspondence, we construct for every
\A[TA]{deep} automaton~$A$, an equivalent binary \A[TA]{deep}~$A'$. Let~$γ$ be %TODO binary TA not defined
a tree node in~$A$ of rank~$k>2$: Replace~$γ$ with nodes~$γ₁,γ₂,…,γ_{k-1}$ of
rank two, and~$γₖ$ of rank one. Tree nodes appear in both sides of tree rewrite
rules, and in the initial auxiliary storage tree: Replace every occurrence
of~$γ$ in~$A$,~$γ(τ₁,τ₂,…,τₖ)$, with~$γ₁(τ₁,γ₂(τ₂,…γ_{k-1}(τ_{k-1},γₖ(τₖ))…))$.
\end{proof}

\section{Type Capturing and~$𝜀$-Transitions}
\label{section:capture}
In the previous section we showed that the addition of
\textsl{deep-type-pattern} property, as found in generic, non-method
functions of (say) Java, to the \PP type system, increases its computational
complexity, but does not render it undecidable. We now prove that the addition
of even rudimentary \kk{typeof} to \PP makes it undecidable.

\begin{theorem}\label{theorem:TA:capturing:TM}
$\A[\PP]{deep,rudimentary-typeof}=\text{RE}$.
\end{theorem}

The following reduction is pertinent to the proof of~\cref{theorem:TA:capturing:TM}.

\begin{lemma}\label{lemma:deep}
A Turing machine~$M$ can be simulated by a \textsl{deep-rewrite},
\textsl{stateful} tree automaton~$A$ which is allowed
\textsl{$𝜀$-transitions}.
\end{lemma}
\nomenclature[B]{$M$}{a Turing machine}

%\vspace{-0.5ex}
\begin{proof}
As explained in \cref{section:automata}, we can assume that~$M$ accepts its
input on the tape with the head on the first letter, and then engages
in~$𝜀$-transitions only. Also, w.l.o.g.,~$M$'s tape is extended infinitely
in both directions by an infinite sequences of a designated blank symbol~$♭$.
\nomenclature[C]{$♭$}{designated blank symbol occupying uninitialized cells of tape auxiliary storage}
\newlength{\savedparindent}
\setlength{\savedparindent}{\parindent}
⏎[1.5pt]\hspace{-3ex} % this, the minipage and \parindent enable wrapfigure
\begin{minipage}\textwidth
  \begin{wrapfigure}{r}{32ex}
      \caption{Turing machine accepting the language~$aⁿbⁿ$}
      \label{figure:Turing}
      \centering
      \scriptsize
      \begin{adjustbox}{max width=\linewidth,bgcolor={RoyalBlue!20}}
          \hspace{1ex}
          \begin{tikzpicture}[node distance=8ex,auto,shorten >=1pt]
\tikzstyle{s}=[circle,fill=SpringGreen!60,inner sep=2pt,draw]
\tikzstyle{e}=[thick]

\node[s] (q1) {$q₁$};
\node[s] (q2) [below right=of q1] {$q₂$};
\node[s] (q3) [below left=of q1] {$q₃$};
\node[s,initial] (q0) [above left=of q3] {$q₀$};
\node[s,accepting] (q4) [below left=of q0] {$q₄$};

\path[->]
  (q0) edge[color=purple,e] node[above]{$a→♭₊$} (q1)
  (q1) edge[e] node[]{$♭→♭₋$} (q2)
  (q2) edge[e] node[]{$b→♭₋$} (q3)
  (q3) edge[e] node[above right=-.8em and 0.5em]{$♭→♭₊$} (q0)
  (q0) edge[e] node[below right=0em and -0.5em]{$♭→♭₊$} (q4)
(q1) edge[e,loop above] node{$\begin{aligned}a→a₊ ⏎[-.6em] b→b₊\end{aligned}$} ()
(q3) edge[e,loop below] node{$\begin{aligned}a→a₋ ⏎[-.6em] b→b₋\end{aligned}$} ()
;
\end{tikzpicture}
        \end{adjustbox}
    \end{wrapfigure}
  \setlength{\parindent}{\savedparindent}
  \indent
\Cref{figure:Turing} is an example of such a machine~$M$ with internal
states~$q₀$ through~$q₄$, single accepting state~$q₄$, and, tape
alphabet~$Γ=❴a,b,♭❵$. The machine terminates in an accepting state if and only
if the tape is initialized with a word~$aⁿbⁿ$,~$n≥0$: To see this, notice that
the machine repeatedly replaces~$a$ from the beginning of the word and its
counterpart letter~$b$ from the word's end by~$♭$, until no more~$a$'s or~$b$'s
are left. The convention of depicting transitions over edges in the graph of
states is standard, e.g., the arrow and label rendered in purple (going from
state~$q₁$ to state~$q₂$) is the~$𝜀$-transition item
\vspace{-0.0\baselineskip}\par
\begin{minipage}{0.96\linewidth}
\begin{equation}\label{eq:Turing}
  ⟨q₀, a→♭₊,q₁⟩,
\end{equation}
\end{minipage}\vspace{0.2\baselineskip}
\par
\hspace{-\parindent}%
which states that if the Turing machine is in internal state~$q₀$, and, the
symbol under head is~$a$, then \1 replace~$a$ by~$♭$, \2 increment~$h$, and
and, \3 change internal state to~$q₁$.
\end{minipage}
\vspace{1ex}

The encoding of~$M$ in~$A$ includes the following components:
\begin{enumerate}
  \item Adopting the set of states~$Q$, set of accepting states~$F$, and
    initial state~$q₀$ of~$M$.

  \item A rank-1 tree symbol for each of the tape symbols, including~$♭$.

  \item Employing the designated leaf symbol~$𝛜∉Γ$ to encode the infinite
    sequences of~$♭$ at the ends of the tape.

  \item Introducing a rank-3 tree symbol~$∘$ for encoding the tape itself. The
    center child of a node labeled~$∘$ encodes of a~$∘$ node encodes the cell
    under the head; its left (resp.~right) child encodes the tape to the left
    (resp.~to the right) of the head. For example, the tape
    contents~$⋯♭♭♭b\underline{a}abb♭♭♭⋯$ is encoded by a certain
    tree~$t=∘(b(𝛜)),a(𝛜),a(b(b(𝛜)))$.

    For the sake of readability we write~$∘$ nodes in infix notation,
    e.g.,~$t=b(𝛜))/a(𝛜)/a(b(b(𝛜)))$, or even more concisely~$t=b/a/abb$.

  \item Setting~$𝛄₀=𝛜/σ₁/σ₂⋯σₙ$, i.e., letting the initial state of auxiliary
    storage encode the input word~$σ₁σ₂⋯σₙ$.

  \item Introducing~$|Σ|+1$ transitions in~$A$ for each of~$M$'s transitions: A
    single transition for dealing with the~$𝛜$ leaf denoting an infinite
    sequence of blanks, and a transition for each tape symbol. In
    demonstration, transition~$⟨q₀, a→♭₊,q₁⟩$~\eq{Turing} is encoded in
    four~$𝜀$-transitions of~$A$ which differ only in their tree rewrite rule.
    \begin{equation}\label{eq:Turing:ta}
      \begin{array}{*2l}
        ⟨q₀,~x₁/a/ax₂→♭x₁/a/x₂,~q₁⟩ &\qquad ⟨q₀,~x₁/a/bx₂→♭x₁/b/x₂,~q₁⟩ ⏎
        ⟨q₀,~x₁/a/♭x₂→♭x₁/♭/x₂,~q₁⟩ &\qquad ⟨q₀,~x₁/a/𝛜→♭x₁/♭/𝛜,~q₁⟩.
      \end{array}
    \end{equation}
  The rules above distinguish between the values the right child of node~$∘$,
  i.e., the symbol to the right of the head: For example, the first
  rule,~$x₁/a/ax₂→♭x₁/a/x₂$, deals with the case this child is~$a$ followed by
  some tape suffix captured in variable~$x₂$. The rule rewrites the node,
  making~$a$ the center child.
\end{enumerate}
Notice that with the encoding, the input to~$A$ is encoded in its transitions rules.
  \qedhere
\end{proof}

Relying on \cref{lemma:deep}, the proof of \cref{theorem:TA:capturing:TM} is
completed by encoding the automaton~$A$ of the lemma in the appropriate
type system.

\begin{proof}[Proof of \cref{theorem:TA:capturing:TM}]
\ifjournal
An evidence of type correctness of a program
of type system~$\A[\PP]{rudi\-mentary-typeof, deep}$ is a sequence of
application of typing rules. Since the evidence can be checked
non-deterministically by a Turing
machine, we have that~$\A[\PP]{rudimentary-typeof, deep}⊆\text{RE}$.

Conversely, encode the automaton~$A$ of \cref{lemma:deep}
\else
We encode automaton~$A=A(M)$ 
\fi
as a program~$P=ΔΞe$ in type system \A[\PP]{deep,rudi\-mentary}.
In this encoding, set~$Δ$ is empty, and there is a function~$ξ∈Ξ$ for
every~$𝜀$-transition item in set~$Ξ$ of~$A$. Expression $e$ type
checks against $Ξ$, if, and only if, machine~$M$ (automaton~$A$) halts. 

In the encoding, the tree vocabulary of~$A$ incarnates as generic types: A
three parameter generic type~$∘$, and generic one-parameter type~$γ$ for each
tape symbol, including~$♭$. Also the argument to every function~$ξ∈Ξ$
function is a deep pattern over possible instantiations of~$∘$.

Also, introduce a function symbol~$φ_q$ for every~$q∈Q$, and let every
transition~$⟨q,τ→τ',q'⟩$ of~$A$ add an overloaded
definition~$φ_q:τ→\kk{typeof}~τ'.φ_q'$ to this symbol. Thus, function~$φ_q$
emulates~$A$ in state~$q$ with tape~$τ$: It applies the rewrite~$τ→τ'$ to the
type, and employs the resolution of \kk{typeof} to continue the computation in
function~$φ_q'$ which corresponds to the destination state~$q'$.

For example, the Turing machine transition shown in \eq{Turing}, encoded by the
tree automaton transitions of \eq{Turing:ta}, is embedded in
{C++} using \kk{decltype}, as depicted in \cref{listing:Turing:example}.

\begin{code}[language=c++,style={cc14},
    caption=Definitions in type system \protect{\PP[rudimentary-typeof, deep]}
      (using {C++} syntax) encoding the tree automata transitions of~\protect\eq{Turing:ta},
    label={listing:Turing:example},morekeywords={typeof}]
#define typeof decltype
template<typename xL, typename xR> typeof(q2(O<xL, B<E>, B<xR>>())) q1(O<B<xL>, B<E>, xR>) {}
template<typename xR> typeof(q2(O<E, B<E>, B<xR>>())) q1(O<E, B<E>, xR>) {}
template<typename xL, typename xR> typeof(q2(O<xL, a<E>, B<xR>>())) q1(O<a<xL>, B<E>, xR>) {}
template<typename xL, typename xR> typeof(q2(O<xL, b<E>, B<xR>>())) q1(O<b<xL>, B<E>, xR>) {}
\end{code}

Further, to encode the input word, set~$e=∘(𝛜,σ₁,σ₂(⋯σₙ(𝛜)⋯)).φ_{q₀}$, or, in
monadic abbreviation form,~$e=∘(𝛜,σ₁,σ₂⋯σₙ).φ_{q₀}$.

To terminate the typing process, further overload~$φ_q$ with
definition~$φ_q:∘(x₁,γ(𝛜),x₂)→𝛜$ for every accepting state~$q∈F$ and cell
symbol~$γ∈Γ$, for which a Turing machine transition is not defined. These
definitions correspond to the situation of~$A$ reaching an accepting
state---type checking succeeds if and only if \kk{typeof} resolution reaches
such a definition.

The full {C++} encoding of the Turing machine of \cref{figure:Turing} is
shown in \cref{figure:Turing:cc} in the appendices.
\end{proof}

Having examined the contribution of \textsl{deep} by itself, and the
combination of \textsl{deep} and \textsl{rudimentary} to the computational
complexity of \PP, it is time to consider the contribution of
\textsl{rudimentary} \emph{by itself} to complexity. The following shows that
there is no such contribution.

\begin{theorem}\label{theorem:TA:capturing:DCFL}
$\PP=\A[\PP]{rudimentary-typeof}$
\end{theorem}
\vspace{-1.5ex}
\begin{proof}
  The first direction~$\PP⊆\A[\PP]{rudimentary-typeof}$ is immediate, as every
  \PP program is also a \A[\PP]{rudimentary-typeof} program by definition. We
  prove~$\A[\PP]{rudimentary-typeof}⊆\PP$.

  Given a program~$P=ΔΞ e$ in \A[\PP]{rudimentary} we need to convert it
  into equivalent program~$P'$ in type system \PP. By
  \cref{theorem:pp:dcfl} it is sufficient to convert~$P$ into a
  vanilla tree automaton, i.e., one with neither states nor~$𝜀$-transitions.
  Instead, we convert~$P$ into a more potent tree automaton~$A$ which is
  allowed both~$𝜀$-transitions and states, and then employ Guessarian's
  observation~$\A[TA]{$𝜀$-transitions,stateful}=\|TA|$ (see
  \cref{Guessarian:trees} above) to complete the proof.

  The set of internal states of~$A$ includes an initial and accepting
  state~$q₀$ and a state~$qᵩ$ for every auxiliary function name~$φ$
  used in~$Ξ$.

  Consider a definition in~$P=ΔΞ e$ of a (primary or auxiliary) function that
  employs a \kk{typeof} clause~$τ→\kk{typeof}~ϑ$. With rudimentary \kk{typeof},
  pseudo-expression~$ϑ$ is either~$τ'$ or~$τ'.φ$. Therefore, every function
  definition is either in the direct form~$τ→τ'$ or in the forwarding
  form~$τ→\kk{typeof} τ'.φ$. There are four cases to consider:

  \begin{enumerate}
    \item \emph{Primary function definitions}, found in~$Δ$, are encoded as
      consuming transitions of~$A$:
      \begin{enumerate}
        \item \emph{Direct definition}~$σ:τ→τ'$ is encoded as
          transition~$⟨σ,q₀,τ→τ',q₀⟩$.
        \item \emph{Forwarding definition}~$σ:τ→\kk{typeof}~τ'.φ$ is encoded as
          transition~$⟨σ,q₀,τ→τ',qᵩ⟩$.
      \end{enumerate}
    \item \emph{Auxiliary function definitions}, found in~$Ξ$, are encoded
      as~$𝜀$ transitions of~$A$:
      \begin{enumerate}
        \item \emph{Direct defintion}~$φ:τ→τ'$ is encoded as transition~$⟨qᵩ,τ→τ',q₀⟩$.
        \item \emph{Forwarding definition}~$φ:τ→\kk{typeof}~τ'.φ'$ is converted
          to~$𝜀$-transition~$⟨qᵩ,τ→τ',q_{φ'}⟩$.
      \end{enumerate} % Auxiliary function definition
    \end{enumerate}
  In all four cases, the change from input type to output type by a function is
  encoded as a rewrite of the tree auxiliary storage of~$A$. Direct definitions
  are encoded by~$A$ moving into state~$q₀$ Forwarding to function~$ψ$ is
  encoded by~$A$ moving into state~$qᵩ$.

  Notice that state~$q₀$, the only accepting state, is the only state with
  outgoing consuming transitions, and it is also the only one without
  outgoing~$𝜀$-transitions. Therefore, the automaton consumes a letter in
  state~$q₀$, and finishes conducting~$𝜀$-transitions back in~$q₀$, or
  otherwise it rejects the input.

  With the above construction, expression~$e=𝜀.σ₁.⋯.σₙ$ type-checks against~$Δ$
  and~$Ξ$ if and only if~$A$ accepts word~$w=σ₁⋯σₙ$.
\end{proof}

\vspace{-1.5ex}
\begin{theorem}\label{theorem:fluent}$\Fluent=\text{DCFL}$\end{theorem}
\vspace{-1.5ex}
\begin{proof}
  Yamazaki et al\@.~\citeyear{Yamazaki:2019} showed that~$\text{DCFL}⊆\Fluent$,
  i.e., that any LR language, alternatively, any DCFL, can be encoded in a
  \Fluent program. It remains to show the converse,~$\Fluent⊆\text{DCFL}$.
  We prove~$\Fluent⊆\A[DPDA]{deep}$, noting the folk-lore equality
  \begin{equation}\label{eq:deep:DPDA}
    \A[DPDA]{deep}=\|DPDA|.
  \end{equation}

  The encoding of a \Fluent program in a \A[DPDA]{deep} is reminiscent
  of the encoding of a program in \A[\PP]{rudimentary} type system in a vanilla
  tree automaton in the proof of \cref{theorem:TA:capturing:DCFL} just above.
  The full proof of the current theorem is in \cref{section:theorem:fluent}.
\end{proof}

Having seen that \Fluent is not more expressive than it was intended to be,
it is interesting to check whether its expressive power would increase if it
allowed unrestricted \kk{typeof} clauses.

\begin{theorem}\label{theorem:Fluent:full} 
$\A[\Fluent]{full-typeof}⊋\|DCFL|
$\end{theorem}

The proof is by showing that type system \A[\Fluent]{full-typeof} is expressive
enough to encode the language~$w＃w$, known to be context sensitive.
The full proof is in \cref{section:Fluent:full}.

\section{Overloading Resolution and Deterministic Computation}
\label{section:overloading}
Most previous work concentrated in recognition of deterministic
languages~\cite{Gil:Levy:2016,Grigore:2017,Nakamaru:17,Gil:2019}. We show here
that type system with Ada-like overloading can encode non-deterministic context
free languages as well. Its proof relies on creating a direct correspondence of
the type system and \emph{\textbf context \textbf free \textbf grammars}
(CFGs).

\nomenclature[A]{CFG}{context free grammar}

\begin{theorem}\label{theorem:ada}$
\|UCFL|⊆\A{monadic, eventually-one-type}$
\end{theorem}
\nomenclature[A]{UCFL}{unambiguous context free language}

\begin{proof}
  Given an unambiguous context free grammar~$G$, we encode it as~$Δ$, a set of
  function definitions in \A{monadic, eventually-one-type} such that~$G$ derives
  word~$σ₁⋯σₙ$ if, and only if, expression~$𝜀.σ₁.⋯.σₙ.\$$ ($\$$ being a
  dedicated function symbol) type checks against~$Δ$.

  We redefine CFGs using a notation more consistent with this
  manuscript: Context free grammar~$G$ is a specification of a formal language
  over alphabet~$Σ$ in the form of a quadruple~$⟨Σ,Γ,𝛜,R⟩$ where~$Σ$ is the set
  of~$G$'s terminals,~$Γ$ is the set of grammar variables,~$𝛜∉Γ$ is the start
  symbol, and~$R$ is a set of derivation rules. Each derivation rule~$ρ∈R$ is
  either in the form~$𝛜→ω$, or in the form~$γ→ω$, where~$γ∈Γ$ and where~$ω$ is
  a possibly empty sequence of terminals and grammar variables,
  i.e.,~$ω∈❨Σ∪Γ❩^*$.

  \nomenclature[B]{$R$}{set of derivation rules of CFG}
  \nomenclature[G]{$Γ$}{set of variables of CFG}
  \nomenclature[G]{$Σ$}{set of terminals of CFG}
  \nomenclature[H]{$γ$}{variable (non-terminal) of CFG}
  \nomenclature[H]{$𝛜$}{start symbol of CFG}
  \nomenclature[H]{$ρ$}{derivation rule of CFG}
  \nomenclature[H]{$σ$}{terminal of CFG}
  \nomenclature[H]{$ω$}{sentential form, i.e., a sequence of terminals and variables of a CFG,~$ω∈❨Σ∪Γ❩^*$}

  Recall that a grammar is in \emph{\textbf Greibach \textbf Normal \textbf
  Form} (GNF) if every rule~$ρ∈R$ is in one of three forms \1~the \emph{usual}
  form,~$ρ=γ→σ𝛄$, where~$σ∈Σ$ is a terminal and~$𝛄∈Γ^*$ is a sequence of
  variables, \2~the \emph{initialization} form,~$ρ=𝛜→σ𝛄$, or,
  \3~the~\emph{$𝜀$-form},~$ρ=𝛜→𝜀$, present only if the grammar derives the
  empty word~$𝜀∈Σ^*$.
  \nomenclature[H]{$𝛄$}{sequence of CFG variables,~$𝛄∈Γ^*$}
  \nomenclature[A]{GNF}{Greibach normal form (of CFG)}

% For ease of exposition, we modify the definition of GNF to derive words that
% end with a designated letter~$\$∈Σ$ and that~$\$$ does not occur anywhere
% else in the word. In this modification, we have, e.g., that the {$𝜀$-form}
% rule is in the form~$𝛜→\$$. No generality is lost.

  For the encoding, first convert unambiguous grammar~$G$ into an equivalent
  unambiguous grammar in GNF. This is done using the algorithm of Nijholt
  \citeyear{Nijholt:79} (also presented in more accessible form by Salomaa and
  Soittola~\citeyear{Salomaa:78}).

  The type encoding of GNF grammar~$G$ uses a monadic generic type~$γ$ for
  every symbol~$γ∈Γ$, an additional monadic generic type~$\boldsymbol{\$}$,
  and, one non-generic type~$𝛜$, also known as the unit type.
% Let~$Γ⁺$ be the set of possible grounded types.
% A type~$𝛄∈Γ⁺$ is an instantiated generic in the form~$γ₁⋯γₖ𝛜$,
% where~$k≥0$ and for~$i=1,…,k$,~$γᵢ∈Γ$.
% \nomenclature[G]{$Γ⁺$}{set of concrete types used to encode a GNF grammar}

  For each derivation rule~$ρ∈R$ introduces a function~$δ∈Δ$ that uses these
  types:
  \begin{itemize}
    \item Suppose~$R$ includes the~$𝜀$-form rule~$𝛜→𝜀\$$, introduce (one
      overloaded) definition of function~$\$:𝛜→𝛜$. Then,~$𝜀.\$$, the expression
      corresponding to the empty word, type-checks to type~$𝛜$. (Recall
      that~$𝜀$ is the single type of the unit type~$𝛜$.)
    \item If~$ρ$ is in the initialization form~$𝛜→σ𝛄$
      then~$δ=σ:𝛜→𝛄\boldsymbol{\$}$. For such a rule
      introduce also function~$\$:→\boldsymbol{\$}𝛜→𝛜$.
    \item If~$ρ$ is in the usual form~$γ→σ𝛄$, then~$δ=σ:γx→𝛄x$.
 \end{itemize}

We show by induction on~$i=1,…,n$ the following claim on the partial
expression~$eᵢ=𝜀.σ₁.⋯.σᵢ$: The set of types assigned by the type checker
to~$eᵢ$ includes a type~$𝛄\boldsymbol{\$}$,~$𝛄∈Γ⁺$, if and only if, there
exists a \emph{\textbf left \textbf most \textbf derivation} (LMD) that yields
the sentential form~$σ₁⋯σᵢ𝛄$.

%\nomenclature[H]{$Γ⁺$}{set of strings including~$γ₁⋯γₖλ𝛜$,~$𝛜∉Γ$,~$k≥0$ and~$γᵢ∈Γ$ for~$i=1,…,k$}
\nomenclature[A]{LMD}{left-most derivation}

For the inductive base observe that~$e₀=𝜀$ and that the set of types of~$𝜀$
includes only the unit type~$𝛜$; indeed there is a (trivial) LMD of the
degenerate sentential form~$𝜀𝛜=𝛜$.

Consider an LMD of~$σ₁⋯σᵢσ_{i+1}𝛄'\boldsymbol{\$}$, where~$i<n$,~$𝛄'∈Γ⁺$
and~$σ_{i+1}$ is the terminal~$ς∈Σ$,~$ς≠\$$. We show that~$𝛄'$ is a type
of~$e_{i+1}=ς(eᵢ)$. The said LMD can only be obtained by applying a
rule~$ρ=γ→ς𝛄"$ to the sentential form~$σ₁⋯σᵢ𝛄\boldsymbol{\$}$, where~$γ$ is the
first symbol of~$𝛄$.

\nomenclature[H]{$ς$}{A terminal of a CFG, or the special symbol \$}

% \nomenclature[H]{$β$}{terminal \emph{or} start symbol of CFG, i.e., a member of~$Γ∪❴𝛜❵$}
%\begin{itemize}
% \item Suppose first that~$β∈Γ$. Then, \item If~$β∉Γ$, then~$β=𝛜$, and~$𝛄$
% degenerates to~$𝛜$. Again, the inductive hypothesis,~$eᵢ$ has a type~$𝛜$.
% Therefore, there is an LMD of~$σ₁⋯σᵢ𝛄$, and a rule application.
%\end{itemize}

By examining the kind of functions in~$Δ$, one can similarly show that every
type~$𝛄'$ of~$e_{i+1}$ is an evidence of an LMD
of a sentential form~$σ₁⋯σᵢσ_{i+1}𝛄'$.

The proof is completed by manually checking that a full expression, ending with
the~$.\$$ invocation can only type check to a single type,~$𝛜$, and
this can happen only if the type of~$𝜀.σ₁.⋯.σₙ$ is~$𝛄$, where~$𝛄$ occurs
in an initialization rule~$𝛜→σₙ𝛄$. \qedhere
\end{proof}

\Cref{section:palindrome} demonstrates the proof by presenting
a fluent API of the non-deterministic context free language of even
length palindromes.

If final expressions are also allowed to be multi-typed, then we can construct
fluent API for all context free languages.

\begin{theorem}\label{theorem:multi}$
\A{monadic, multiple-type}=\text{CFL}$
\end{theorem}

\begin{proof}
  The construction in the proof of \cref{theorem:ada} works here as well.
  Note that here the transition from a plain CFG to GNF does not have to
  preserve unambiguity.
\end{proof}

\section{Conclusions}
\label{section:zz}
\paragraph{Perspective.}
Revisiting \cref{table:lattice}, we see that in total it
has~{\small$|C₁|·|C₂|·|C₃|·|C₄|·|C₅|·|C₆|=4·3·2·2·3·3=432$} lattice points.
Accounting for the fact that in a \textsl{nyladic} type system, the values
of~{\small$C₂$} (\textsl{type pattern depth}), and~{\small$C₃$} (\textsl{type
pattern multiplicity}) are meaningless, we see that lattice~$𝔗$
spans~{\small$|C₄|·|C₅|·|C₆|=2·3·3=18$} monomorphic type systems ($𝔗_⊥$ among
them), and~{\small$(|C₁|-1)·|C₂|·|C₃|·|C₄|·|C₅|·|C₆|=3·3·2·2·3·3=324$}
potential polymorphic type systems (\PP and~\Fluent among them). To make the
count more exact, account for~{\small$C₃$} being irrelevant in a
\textsl{monadic} type system,
obtaining~{\small$|C₂|·|C₄|·|C₅|·|C₆|=3·2·3·3=36$} \textsl{monadic}, yet
polymorphic type systems,
and~{\small$(|C₁|-2)·|C₂|·|C₃|·|C₄|·|C₅|·|C₆|=2·3·2·2·3·3=216$}
non-\textsl{monadic} polymorphic type systems.

Beyond the implicit mention that the type-automata correspondence applies to
\emph{monomorphic type systems}, these were not considered here. Our study also
invariably assumed \textsl{unary-function}, ignoring in characteristic~$C₄$
\emph{\textsl{n-ary-functions} type systems}%
†{the ignored \textsl{n-ary-functions} correspond to the \textsl{forest-recognizer} brand
  of automata; however \textsl{forest-recognizer} automata were used in the
construction, e.g., in \cref{lemma:deep}.}
which comprise half of the type systems of~$𝔗$.

Even though most of this work was in characterizing the complexity classes of
type systems, it could not have covered even the~$(36+216)/2=126$ type systems
remaining in scope. The study rather focused on these systems which we thought
are more interesting: We gave an exact characterization of the complexity
classes of two central type systems,~\PP (\cref{theorem:pp:dcfl})
and~\Fluent (\cref{theorem:fluent}), and investigated how this complexity
changes if the type systems are made more or less potent along~$𝔗$'s
characteristics (with the exception of~$C₄$, the function arity characteristic).
Comparing \eq{fluent} with \cref{table:lattice} we see that \Fluent can be made
more potent along~$C₁$,~$C₅$, or~$C₆$, and, as follows from our results, its
complexity class increases in all three cases:
\begin{enumerate}
  \item In~$C₁$,~$\Fluent⊊\A[\Fluent]{dyadic}=\|RE|$, by combining \cref{theorem:dyadic}
    and \cref{theorem:TA:capturing:TM}.
  \item In~$C₅$,~$\Fluent⊊\A[\Fluent]{eventually-one-type}$ (\cref{theorem:ada}).
  \item In~$C₆$,~$\Fluent⊊\A[\Fluent]{full-typeof}$ (\cref{theorem:Fluent:full}).
\end{enumerate}
Conversely, \Fluent can be made less potent along characteristics~$C₁$,~$C₂$
and~$C₅$:
\begin{enumerate}
  \item In~$C₁$ complexity decreases,~$\Fluent-\textsl{monadic}=FSA⊊\Fluent$
    (\cref{observation:correspondence}).
  \item In~$C₂$, \eq{deep:DPDA} makes us 
    believe that complexity does not change,
    $\Fluent-\textsl{deep}+\textsl{shallow}=\Fluent$.
  \item In~$C₅$, then, by \cref{observation:correspondence,eq:deep:DPDA}),~$\Fluent-\textsl{rudimentary}=\A[RDPDA]{deep}$.
    We believe complexity decreases but are unsure.
\end{enumerate}
Type system \PP can be made more potent along characteristics~$C₂$,~$C₃$,~$C₅$ and~$C₆$:
\begin{enumerate}
  \item In~$C₂$ complexity increases,~$\PP⊊\A[\PP]{deep}$ (\cref{theorem:PP:deep}).
  \item In~$C₃$ complexity increases,~$\PP⊊\A[\PP]{non-linear}$ (\cref{theorem:PP:nonlinear}).
  \item In~$C₅$ complexity does not change,~$\PP=\A[\PP]{rudimentary-typeof}$
    (\cref{theorem:TA:capturing:DCFL}).
  \item In~$C₆$ complexity increases,~$\PP⊊\A[\PP]{eventually-one-type}$ (\cref{theorem:ada}).
\end{enumerate}
Type system \PP can be made less potent only along characteristic~$C₁$.
From \cref{observation:correspondence} and \cref{theorem:pp:dcfl},
\begin{equation}
\begin{aligned}
  FSA=\A{nylaldic}⊊SRDPDA=\A{monadic} &⊆\A{dyadic}⊆\A{polyadic} ⏎
                                      &⊊\A{polyadic}=\|DCFL|,
\end{aligned}
\end{equation}
i.e., it is not known whether decreasing~$\PP$ along~$C₁$ to \textsl{dyadic}
reduces its complexity, but decreasing it further to \textsl{monadic} certainely
does.

This work should also be viewed as a study of the type-automata correspondence:
\1~The results in \cref{section:polymorphism} revolve around the correspondence
between \textsl{tree-store} automata employing tree rewrites, and
type system in which the signature of functions employs type pattern to match
its argument. \2~\Cref{section:capture} explored the correspondence between
\kk{typeof} clause in the signature of functions, and~$𝜀$-transitions of automata.
\3~The correspondence between non-deterministic runs and allowing
multiple types of expressions, or at least as a partial step during
resolution of overloading was the subject of \cref{section:overloading}. 
Overall, our study confirmed that
the type-automata correspondence is a significant aid in the characterization
of complexity classes, either by a direct bisimulation between the two, or by
employing and adapting (sometimes ancient) contributions in the decades old
research of automata.

\ifjournal
In \Cref{theorem:fluent} we showed that~$\Fluent=\text{DCFL}$, i.e., the
computational complexity of the intermediate language is in exact match with
the purpose it was designed to serve---parsing LR languages. We are in
position now to examine the change to the \Fluent's computational complexity by
making more potent along any of the dimensions of \cref{table:lattice}.
\fi

\paragraph{Open Problems.}
Technically, we leave open the problem of characterizing the complexity class
of each of the 126 type systems that were not considered at all, or, considered,
but not fully characterized. However, many of these can be trivially solved, e.g.,
since~$T₁=⟨deep,rudimentary,polyadic⟩=\text{RE}$,
(\cref{theorem:TA:capturing:TM}),~$T₂=\|RE|$ for all~$T₂∈𝔗$,~$T₂>T₁$.
We draw attention to four type systems for which we are able to set a lower and
an upper bound, but still miss precise characterization, e.g., in terms of
familiar computational complexity classes.
\begin{enumerate}
  \item \PP[deep], for which we have~$\|DCFL|⊊\PP[deep]⊆\|CSL|$ by
    \cref{theorem:PP:deep}.
  \item \PP[non-linear], for which we also have~$\|DCFL|⊊\PP[non-linear]⊆\|CSL|$ by
    \cref{theorem:PP:nonlinear}.
  \item \PP[deep,non-linear], for which we have again~$\|DCFL|⊊\PP[deep, non-linear]⊆\|CSL|$ by
      \cref{theorem:PP:deep,theorem:PP:nonlinear}.
  \item \A[\Fluent]{full-typeof}, for which we have~$\|DCFL|⊊\A[\Fluent]{full-typeof}⊆\|RE|$ by
      \cref{theorem:Fluent:full}.
\end{enumerate}
Also, we do not know yet how these relate to each other in terms of
computational complexity, beyond what can be trivially inferred by~$𝔗$'s
partial order. \Cref{section:deep} may offer some insights.

\paragraph{Expression Trees vs.\ Expression Words}
Language recognizers, i.e., automata which take trees as inputs 
were defined and used in the proofs. Still, this study does not offer much on
the study of \textsl{n-ary-functions}---the type counterpart of language
recognizers.  There is potential in exploring the theory of polymorphic types
of tree shaped expressions.
In particular, it is interesting to study type
systems~$S₁=⟨\textsl{n-ary},\textsl{deep}⟩$
and~$S₂=⟨\textsl{n-ary},\textsl{deep},\textsl{non-linear}⟩$, both modeling
\kk{static} generic multi-argument functions of C＃ and Java, except that~$S₂$
adds the power, and predicament (see \cref{listing:s2}), of non-linear type
patterns. In the type-automata perspective~$S₁$ and~$S₂$ correspond
to \textsf{forest-recognizer} \textsf{real-time} \textsf{tree-store} brand
of automata, which received little attention in the literature.
We see two number of potential applications of type theory, for which (say) \PP
is insufficient, and could serve as motivation for resolving the open problems
above and for the study of~$S₁$ and~$S₂$.

\nomenclature[B]{$S₁$}{$⟨\textsl{n-ary},\textsl{deep}⟩$ (type system in~$𝔗$)}
\nomenclature[B]{$S₂$}{$⟨\textsl{n-ary},\textsl{deep},\textsl{non-linear}⟩$ (type system in~$𝔗$) \cref{listing:s2}}

\begin{description}
  \nomenclature[B]{$A$}{a two-dimensional matrix}
  \nomenclature[B]{$B$}{a two-dimensional matrix}
  \nomenclature[C]{$m$}{a dimension of a matrix}
  \item[Types for linear algebra]
    The matrix product~$A⨉B$ is defined if~matrix~$A$ is~$m₁⨉m₂$ and
    matrix~$B$ is~$m₂×m₃$, in which case the result is an~$m₁⨉m₃$ matrix.
    The matrix addition~$A+B$ is defined only if both~$A$ and~$B$ are~$m₁⨉m₂$,
    in which case the result is also~$m₁⨉m₂$. The unary encoding of integers
    and their comparison in one step in the proof of
    \cref{theorem:PP:nonlinear} seem to be sufficient for developing a
    \emph{decidable} type system that enforces such constraints.

    However, unlike type systems for checking fluent API, types for linear
    algebra implemented this way are impractical: matrices whose dimensions are
    in the range of thousands are common, e.g., in image processing. But,
    programmers cannot be expected to encode integers this large in unary, not
    mentioning the fact that such types tend to challenge compilers' stability.
    The problem is ameliorated in~$S₂$ in which a decimal (say) representation
    of integers is feasible. A more precise design is left for future research.

    A more difficult challenge is the type system support and checking of
    operations which involve integer arithmetic. A prime example is
    \textsc{numpy}†{\url{https://numpy.org/}}'s \texttt{reshape} operation
    which converts, e.g., an~$m₁⨉m₂$ matrix to an~$m₃⨉m₄$~matrix,
    where correctness is contingent on the equality if~$m₁·m₂=m₃·m₄$. Indeed,
    we are not aware of any decidable type system that can do integer multiplication.

  \item[Dimensional types] A similar challenge is supporting of
    \textit{physical dimensions}, i.e., a design of a type system allowing,
    e.g., the division of distance quantity by time quantity obtaining speed
    quantity, and addition and comparison distance quantities, but 
    forbidding, e.g., addition and comparison of time and
    distance quantities. To do so, the type system should probably
    encode~$∏_{i=1}ʳxᵢ^{mᵢ}$,~$mᵢ∈ℤ$, the general form of a physical dimension
    (in say MKS), as a tuple of~$r$ of signed integers.

    To enforce the rules of addition and comparison of physical dimensions,
    the type system should be able compare (typically very small) integers,
    as done in \cref{theorem:PP:nonlinear}, although the implementation should
    be tweaked to support negative integers.
    For multiplying and dividing physical quantities, the type system
    should be able to add (small) integers.
    We do not know whether this is possible in~$S₁$ or~$S₂$.

    \nomenclature[A]{MKS}{meter-kilogram-second (system of physical units}
    \nomenclature[C]{$xᵢ$}{a physical unit such as centimeter, second, gram, and coulomb}
    \nomenclature[C]{$r$}{number of physical units in a system of physical units such as MKS}
    \nomenclature[C]{$m$}{exponent of certain physical unit in a physical dimension such as kilogram/meter-squared}
    \nomenclature[B]{$ℤ$}{set of signed integers,~$❴⋯,-2,-1,0,1,2,…❵$}
\end{description}

\paragraph{Modeling type erasure}
Finally, we draw attention to the fact that Java's type erasure is not
accurately modeled by our system. In particular Java forbids function
overloading if the type of the overloaded functions becomes identical
after type erasure. We propose this type inference rule for type erasure
\begin{equation}
\hspace{-40ex}
\typing{Type ⏎ Erasure}
{\infer{σ:⊥}{σ:γ(𝛕)→τ & σ:γ(𝛕')→τ'}}
\end{equation}
and leave the problem of studying type systems with type erasure to future
research.

\bibliographystyle{ACM-Reference-Format}
\bibliography{00,author-names,publishers,big}


\begin{thebibliography}{00}

%%% the Web address.

\ifx \showCODEN    \undefined \def \showCODEN     #1{\unskip}     \fi
\ifx \showDOI      \undefined \def \showDOI       #1{#1}\fi
\ifx \showISBNx    \undefined \def \showISBNx     #1{\unskip}     \fi
\ifx \showISBNxiii \undefined \def \showISBNxiii  #1{\unskip}     \fi
\ifx \showISSN     \undefined \def \showISSN      #1{\unskip}     \fi
\ifx \showLCCN     \undefined \def \showLCCN      #1{\unskip}     \fi
\ifx \shownote     \undefined \def \shownote      #1{#1}          \fi
\ifx \showarticletitle \undefined \def \showarticletitle #1{#1}   \fi
\ifx \showURL      \undefined \def \showURL       {\relax}        \fi
% invisible to TeX
\providecommand\bibfield[2]{#2}
\providecommand\bibinfo[2]{#2}
\providecommand\natexlab[1]{#1}
\providecommand\showeprint[2][]{arXiv:#2}

\bibitem[\protect\citeauthoryear{Amin and Tate}{Amin and Tate}{2016}]%
        {Amin:16}
\bibfield{author}{\bibinfo{person}{Nada Amin} {and} \bibinfo{person}{Ross
  Tate}.} \bibinfo{year}{2016}\natexlab{}.
\newblock \showarticletitle{{Java} and {Scala}’s type systems are unsound:
  {The} existential crisis of null pointers}. In \bibinfo{booktitle}{{\em
  Proceedings of the 2016 ACM SIGPLAN International Conference on
  Object-Oriented Programming, Systems, Languages, and Applications}} {\em
  (\bibinfo{series}{OOPSLA 2016})}. \bibinfo{publisher}{Association for
  Computing Machinery}, \bibinfo{address}{New York, NY, USA},
  \bibinfo{pages}{838--848}.
\newblock
\showISBNx{9781450344449}
\showDOI{%
\url{https://doi.org/10.1145/2983990.2984004}}


\bibitem[\protect\citeauthoryear{Autebert, Berstel, and Boasson}{Autebert
  et~al\mbox{.}}{1997}]%
        {Autebert:1997}
\bibfield{author}{\bibinfo{person}{Jean-Michel Autebert}, \bibinfo{person}{Jean
  Berstel}, {and} \bibinfo{person}{Luc Boasson}.}
  \bibinfo{year}{1997}\natexlab{}.
\newblock \bibinfo{booktitle}{{\em Context-free languages and pushdown
  automata}}.
\newblock \bibinfo{publisher}{Springer}, \bibinfo{address}{Berlin, Heidelberg}.
  111--174 pages.
\newblock
\showISBNx{978-3-642-59136-5}
\showDOI{%
\url{https://doi.org/10.1007/978-3-642-59136-5₃}}


\bibitem[\protect\citeauthoryear{Barendregt}{Barendregt}{1991}]%
        {Barendregt:91}
\bibfield{author}{\bibinfo{person}{Henk Barendregt}.}
  \bibinfo{year}{1991}\natexlab{}.
\newblock \showarticletitle{Introduction to generalized type systems}.
\newblock \bibinfo{journal}{{\em J. Functional Programming\/}}
  \bibinfo{volume}{1} (\bibinfo{year}{1991}), \bibinfo{pages}{125–154}.
\newblock
Issue 2.
\showISSN{0956-7968}
\showDOI{%
\url{https://doi.org/10.1017/s0956796800020025.}}


\bibitem[\protect\citeauthoryear{Bruda and Akl}{Bruda and Akl}{1999}]%
        {Bruda:Akl}
\bibfield{author}{\bibinfo{person}{Stefan~D. Bruda} {and}
  \bibinfo{person}{Selim~G. Akl}.} \bibinfo{year}{1999}\natexlab{}.
\newblock \bibinfo{booktitle}{{\em On the power of real-time {Turing}
  machines:~$k$ tapes are more powerful than~$k-1$ tapes}}.
\newblock \bibinfo{type}{{T}echnical {R}eport}. \bibinfo{institution}{Queen's
  University}.
\newblock


\bibitem[\protect\citeauthoryear{Comon, Dauchet, Gilleron, Löding, Jacquemard,
  Lugiez, Tison, and Tommasi}{Comon et~al\mbox{.}}{2007}]%
        {Comon:07}
\bibfield{author}{\bibinfo{person}{H. Comon}, \bibinfo{person}{M. Dauchet},
  \bibinfo{person}{R. Gilleron}, \bibinfo{person}{C. Löding},
  \bibinfo{person}{F. Jacquemard}, \bibinfo{person}{D. Lugiez},
  \bibinfo{person}{S. Tison}, {and} \bibinfo{person}{M. Tommasi}.}
  \bibinfo{year}{2007}\natexlab{}.
\newblock \bibinfo{title}{Tree Automata Techniques and {App.}}
\newblock \bibinfo{howpublished}{Available on:
  \url{http://www.grappa.univ-lille3.fr/tata}}.   (\bibinfo{year}{2007}).
\newblock
\newblock
\shownote{release Oct, 12th 2007.}


\bibitem[\protect\citeauthoryear{Coquidé, Dauchet, Gilleron, and
  Vágvölgyi}{Coquidé et~al\mbox{.}}{1994}]%
        {coquide1994bottom}
\bibfield{author}{\bibinfo{person}{Jean-Luc Coquidé}, \bibinfo{person}{Max
  Dauchet}, \bibinfo{person}{Rémi Gilleron}, {and} \bibinfo{person}{Sándor
  Vágvölgyi}.} \bibinfo{year}{1994}\natexlab{}.
\newblock \showarticletitle{Bottom-up tree pushdown automata: classification
  and connection with rewrite systems}.
\newblock \bibinfo{journal}{{\em Theoretical {Comp.} Science\/}}
  \bibinfo{volume}{127}, \bibinfo{number}{1} (\bibinfo{year}{1994}),
  \bibinfo{pages}{69--98}.
\newblock
\showISSN{0304-3975}
\showDOI{%
\url{https://doi.org/10.1016/0304-3975(94)90101-5}}


\bibitem[\protect\citeauthoryear{Damas and Milner}{Damas and Milner}{1982}]%
        {Damas:82}
\bibfield{author}{\bibinfo{person}{Luis Damas} {and} \bibinfo{person}{Robin
  Milner}.} \bibinfo{year}{1982}\natexlab{}.
\newblock \showarticletitle{Principal type-schemes for functional programs}. In
  \bibinfo{booktitle}{{\em Proceedings of the 9th ACM SIGPLAN-SIGACT Symposium
  on Principles of Programming Languages}} {\em (\bibinfo{series}{POPL
  ’82})}. \bibinfo{publisher}{Association for Computing Machinery},
  \bibinfo{address}{New York, NY, USA}, \bibinfo{pages}{207--212}.
\newblock
\showISBNx{0897910656}
\showDOI{%
\url{https://doi.org/10.1145/582153.582176}}


\bibitem[\protect\citeauthoryear{Donovan and Kernighan}{Donovan and
  Kernighan}{2015}]%
        {Donovan:15}
\bibfield{author}{\bibinfo{person}{Alan~A.A. Donovan} {and}
  \bibinfo{person}{Brian~W. Kernighan}.} \bibinfo{year}{2015}\natexlab{}.
\newblock \bibinfo{booktitle}{{\em The {Go} Programming Language}}.
\newblock \bibinfo{publisher}{Addison-Wesley Professional},
  \bibinfo{address}{Boston, MA, USA}.
\newblock


\bibitem[\protect\citeauthoryear{Gil and Levy}{Gil and Levy}{2016}]%
        {Gil:Levy:2016}
\bibfield{author}{\bibinfo{person}{Yossi Gil} {and} \bibinfo{person}{Tomer
  Levy}.} \bibinfo{year}{2016}\natexlab{}.
\newblock \showarticletitle{Formal language recognition with the {Java} type
  checker}. In \bibinfo{booktitle}{{\em 30th European {Conf.} on {OO} {Prog.}
  (ECOOP 2016)}} {\em (\bibinfo{series}{Leibniz International Proceedings in
  {Inf.} (LIPIcs)})}, \bibfield{editor}{\bibinfo{person}{Shriram Krishnamurthi}
  {and} \bibinfo{person}{Benjamin~S. Lerner}} (Eds.),
  Vol.~\bibinfo{volume}{56}. \bibinfo{publisher}{Schloss
  Dagstuhl--Leibniz-Zentrum fuer Informatik}, \bibinfo{address}{Dagstuhl,
  Germany}, \bibinfo{pages}{10:1--10:27}.
\newblock
\showISBNx{978-3-95977-014-9}
\showISSN{1868-8969}
\showDOI{%
\url{https://doi.org/10.4230/LIPIcs.ECOOP.2016.10}}


\bibitem[\protect\citeauthoryear{Gil and Roth}{Gil and Roth}{2019}]%
        {Gil:2019}
\bibfield{author}{\bibinfo{person}{Yossi Gil} {and} \bibinfo{person}{Ori
  Roth}.} \bibinfo{year}{2019}\natexlab{}.
\newblock \showarticletitle{{Fling}---a fluent {API} generator}. In
  \bibinfo{booktitle}{{\em 33rd European {Conf.} on {OO} {Prog.} (ECOOP 2019)}}
  {\em (\bibinfo{series}{Leibniz International Proceedings in {Inf.}
  (LIPIcs)})}, \bibfield{editor}{\bibinfo{person}{Alastair~F. Donaldson}}
  (Ed.), Vol.~\bibinfo{volume}{134}. \bibinfo{publisher}{Schloss
  Dagstuhl--Leibniz-Zentrum fuer Informatik}, \bibinfo{address}{Dagstuhl,
  Germany}, \bibinfo{pages}{13:1--13:25}.
\newblock
\showISBNx{978-3-95977-111-5}
\showISSN{1868-8969}
\showDOI{%
\url{https://doi.org/10.4230/LIPIcs.ECOOP.2019.13}}


\bibitem[\protect\citeauthoryear{Girard}{Girard}{1971}]%
        {Girard:71}
\bibfield{author}{\bibinfo{person}{Jean-Yves Girard}.}
  \bibinfo{year}{1971}\natexlab{}.
\newblock \showarticletitle{Une Extension De ĽInterpretation De Gödel a
  ĽAnalyse, Et Son Application a ĽElimination Des Coupures Dans ĽAnalyse Et
  La Theorie Des Types}.
\newblock In \bibinfo{booktitle}{{\em Proceedings of the Second Scandinavian
  Logic Symposium}}, \bibfield{editor}{\bibinfo{person}{J.E. Fenstad}} (Ed.).
  \bibinfo{series}{Studies in Logic and the Foundations of Mathematics},
  Vol.~\bibinfo{volume}{63}. \bibinfo{publisher}{Elsevier},
  \bibinfo{pages}{63--92}.
\newblock
\showISSN{0049-237X}
\showDOI{%
\url{https://doi.org/10.1016/S0049-237X(08)70843-7}}


\bibitem[\protect\citeauthoryear{Girard}{Girard}{1972}]%
        {Girard:72}
\bibfield{author}{\bibinfo{person}{Jean-Yves Girard}.}
  \bibinfo{year}{1972}\natexlab{}.
\newblock {\em \bibinfo{title}{Interpr{\'e}tation fonctionnelle et
  {\'e}limination des coupures de l'arithm{\'e}tique d'ordre sup{\'e}rieur}}.
\newblock \bibinfo{thesistype}{Ph.D. Dissertation}.
  \bibinfo{school}{Universit{\'e} Paris}.
\newblock


\bibitem[\protect\citeauthoryear{Grigore}{Grigore}{2017}]%
        {Grigore:2017}
\bibfield{author}{\bibinfo{person}{Radu Grigore}.}
  \bibinfo{year}{2017}\natexlab{}.
\newblock \showarticletitle{{Java} generics are {Turing} complete}.
\newblock \bibinfo{journal}{{\em SIGPLAN Not.\/}} \bibinfo{volume}{52},
  \bibinfo{number}{1} (\bibinfo{date}{Jan.} \bibinfo{year}{2017}),
  \bibinfo{pages}{73–85}.
\newblock
\showISSN{0362-1340}
\showDOI{%
\url{https://doi.org/10.1145/3093333.3009871}}


\bibitem[\protect\citeauthoryear{Guessarian}{Guessarian}{1983}]%
        {Guessarian:83}
\bibfield{author}{\bibinfo{person}{Ir{\`e}ne Guessarian}.}
  \bibinfo{year}{1983}\natexlab{}.
\newblock \showarticletitle{Pushdown tree automata}.
\newblock \bibinfo{journal}{{\em Math. {Syst.} Theory\/}} \bibinfo{volume}{16},
  \bibinfo{number}{1} (\bibinfo{year}{1983}), \bibinfo{pages}{237--263}.
\newblock


\bibitem[\protect\citeauthoryear{Hindley}{Hindley}{1969}]%
        {Hindley:69}
\bibfield{author}{\bibinfo{person}{R. Hindley}.}
  \bibinfo{year}{1969}\natexlab{}.
\newblock \showarticletitle{The principal type-scheme of an object in
  combinatory logic}.
\newblock \bibinfo{journal}{{\it Trans. Amer. Math. Soc.}}
  \bibinfo{volume}{146} (\bibinfo{date}{Dec.} \bibinfo{year}{1969}),
  \bibinfo{pages}{29--60}.
\newblock
\showISSN{00029947}
\showURL{%
\url{http://www.jstor.org/stable/1995158}}


\bibitem[\protect\citeauthoryear{Hopcroft, Motwani, and Ullman}{Hopcroft
  et~al\mbox{.}}{2007}]%
        {Hopcroft:Motwani:Ullman:07}
\bibfield{author}{\bibinfo{person}{John~E. Hopcroft}, \bibinfo{person}{Rajeev
  Motwani}, {and} \bibinfo{person}{Jeffrey~D. Ullman}.}
  \bibinfo{year}{2007}\natexlab{}.
\newblock \bibinfo{booktitle}{{\em Introduction to automata theory, languages,
  and computation\/} (\bibinfo{edition}{3rd} ed.)}.
\newblock \bibinfo{publisher}{Pearson Addison Wesley},
  \bibinfo{address}{Boston, MA}.
\newblock
\showISBNx{0321455371}


\bibitem[\protect\citeauthoryear{Karp}{Karp}{1972}]%
        {Karp:1972}
\bibfield{author}{\bibinfo{person}{Richard~M. Karp}.}
  \bibinfo{year}{1972}\natexlab{}.
\newblock \showarticletitle{Reducibility among combinatorial problems}. In
  \bibinfo{booktitle}{{\em Proc. Symp. Complex. {Comp.}}},
  \bibfield{editor}{\bibinfo{person}{Raymond~E. Miller},
  \bibinfo{person}{James~W. Thatcher}, {and} \bibinfo{person}{Jean~D.
  Bohlinger}} (Eds.). \bibinfo{publisher}{Springer}, \bibinfo{address}{Yorktown
  Heights, NY}, \bibinfo{pages}{85--103}.
\newblock
\showISBNx{978-1-4684-2001-2}
\showDOI{%
\url{https://doi.org/10.1007/978-1-4684-2001-2₉}}


\bibitem[\protect\citeauthoryear{Kennedy and Pierce}{Kennedy and
  Pierce}{2007}]%
        {Kennedy:Pierce:07}
\bibfield{author}{\bibinfo{person}{Andrew Kennedy} {and}
  \bibinfo{person}{Benjamin Pierce}.} \bibinfo{year}{2007}\natexlab{}.
\newblock \showarticletitle{On decidability of nominal subtyping with
  variance}. In \bibinfo{booktitle}{{\em Int. Workshop Found. \& Devel. {OO}
  Lang.}} {\em (\bibinfo{series}{FOOL/WOOD`07})}. \bibinfo{address}{Nice,
  France}.
\newblock
\showURL{%
\url{http://foolwood07.cs.uchicago.edu/program/kennedy-abstract.html}}


\bibitem[\protect\citeauthoryear{Kfoury, Tiuryn, and Urzyczyn}{Kfoury
  et~al\mbox{.}}{1990}]%
        {Kfoury:1990}
\bibfield{author}{\bibinfo{person}{A.~J. Kfoury}, \bibinfo{person}{J. Tiuryn},
  {and} \bibinfo{person}{P. Urzyczyn}.} \bibinfo{year}{1990}\natexlab{}.
\newblock \showarticletitle{{ML} typability is {DEXPTIME}-complete}. In
  \bibinfo{booktitle}{{\em CAAP '90}},
  \bibfield{editor}{\bibinfo{person}{A.~Arnold}} (Ed.).
  \bibinfo{publisher}{Springer}, \bibinfo{address}{New York},
  \bibinfo{pages}{206--220}.
\newblock
\showISBNx{978-3-540-47042-7}


\bibitem[\protect\citeauthoryear{Knuth}{Knuth}{1965}]%
        {Knuth:1965}
\bibfield{author}{\bibinfo{person}{Donald~E. Knuth}.}
  \bibinfo{year}{1965}\natexlab{}.
\newblock \showarticletitle{On the translation of languages from left to
  right}.
\newblock \bibinfo{journal}{{\em Info ＆ Comp.\/}} \bibinfo{volume}{8},
  \bibinfo{number}{6} (\bibinfo{year}{1965}), \bibinfo{pages}{607--639}.
\newblock
\showISSN{0019-9958}
\showDOI{%
\url{https://doi.org/10.1016/S0019-9958(65)90426-2}}


\bibitem[\protect\citeauthoryear{Milner}{Milner}{1978}]%
        {Milner:78}
\bibfield{author}{\bibinfo{person}{Robin Milner}.}
  \bibinfo{year}{1978}\natexlab{}.
\newblock \showarticletitle{A theory of type polymorphism in programming}.
\newblock \bibinfo{journal}{{\it J. Comput. System Sci.}} \bibinfo{volume}{17},
  \bibinfo{number}{3} (\bibinfo{year}{1978}), \bibinfo{pages}{348--375}.
\newblock
\showISSN{0022-0000}
\showDOI{%
\url{https://doi.org/10.1016/0022-0000(78)90014-4}}


\bibitem[\protect\citeauthoryear{Nakamaru and Chiba}{Nakamaru and
  Chiba}{2020}]%
        {Nakamaru:2020}
\bibfield{author}{\bibinfo{person}{Tomoki Nakamaru} {and}
  \bibinfo{person}{Shigeru Chiba}.} \bibinfo{year}{2020}\natexlab{}.
\newblock \showarticletitle{Generating a generic fluent {API} in {Java}}.
\newblock \bibinfo{journal}{{\em The Art, Science, and {Eng.} of {Prog.}\/}}
  \bibinfo{volume}{4}, \bibinfo{number}{3} (\bibinfo{date}{Feb.}
  \bibinfo{year}{2020}).
\newblock
\showISSN{2473-7321}
\showDOI{%
\url{https://doi.org/10.22152/programming-journal.org/2020/4/9}}


\bibitem[\protect\citeauthoryear{Nakamaru, Ichikawa, Yamazaki, and
  Chiba}{Nakamaru et~al\mbox{.}}{2017}]%
        {Nakamaru:17}
\bibfield{author}{\bibinfo{person}{Tomoki Nakamaru}, \bibinfo{person}{Kazuhiro
  Ichikawa}, \bibinfo{person}{Tetsuro Yamazaki}, {and} \bibinfo{person}{Shigeru
  Chiba}.} \bibinfo{year}{2017}\natexlab{}.
\newblock \showarticletitle{{Silverchain}: a fluent {API} generator}. In
  \bibinfo{booktitle}{{\em Proc. 16th ACM SIGPLAN{Int.} {Conf.} Generative
  Prog.}} {\em (\bibinfo{series}{GPCE'17})}. \bibinfo{publisher}{ACM},
  \bibinfo{address}{Vancouver, BC, Canada}, \bibinfo{pages}{199--211}.
\newblock
\showISBNx{978-1-4503-5524-7}


\bibitem[\protect\citeauthoryear{Nijholt}{Nijholt}{1979}]%
        {Nijholt:79}
\bibfield{author}{\bibinfo{person}{Anton Nijholt}.}
  \bibinfo{year}{1979}\natexlab{}.
\newblock \showarticletitle{Grammar functors and covers: From
  non-left-recursive to {Greibach} normal form grammars}.
\newblock \bibinfo{journal}{{\em BIT Numerical Mathematics\/}}
  \bibinfo{volume}{19}, \bibinfo{number}{1} (\bibinfo{date}{01 March}
  \bibinfo{year}{1979}), \bibinfo{pages}{73--78}.
\newblock
\showISSN{1572-9125}
\showDOI{%
\url{https://doi.org/10.1007/BF01931223}}


\bibitem[\protect\citeauthoryear{Persch, Winterstein, Dausmann, and
  Drossopoulou}{Persch et~al\mbox{.}}{1980}]%
        {Persch:1980}
\bibfield{author}{\bibinfo{person}{Guido Persch}, \bibinfo{person}{Georg
  Winterstein}, \bibinfo{person}{Manfred Dausmann}, {and}
  \bibinfo{person}{Sophia Drossopoulou}.} \bibinfo{year}{1980}\natexlab{}.
\newblock \showarticletitle{Overloading in Preliminary {Ada}}.
\newblock \bibinfo{journal}{{\em SIGPLAN Not.\/}} \bibinfo{volume}{15},
  \bibinfo{number}{11} (\bibinfo{date}{Nov.} \bibinfo{year}{1980}),
  \bibinfo{pages}{47–56}.
\newblock
\showISSN{0362-1340}
\showDOI{%
\url{https://doi.org/10.1145/947783.948640}}


\bibitem[\protect\citeauthoryear{Rabin}{Rabin}{1963}]%
        {Rabin:63}
\bibfield{author}{\bibinfo{person}{Michael~O. Rabin}.}
  \bibinfo{year}{1963}\natexlab{}.
\newblock \showarticletitle{Real time computation}.
\newblock \bibinfo{journal}{{\em Israel J. Math.\/}}  \bibinfo{volume}{1}
  (\bibinfo{year}{1963}), \bibinfo{pages}{203--211}.
\newblock


\bibitem[\protect\citeauthoryear{Reynolds}{Reynolds}{1974}]%
        {Reynolds:74}
\bibfield{author}{\bibinfo{person}{John~C. Reynolds}.}
  \bibinfo{year}{1974}\natexlab{}.
\newblock \showarticletitle{Towards a theory of type structure}. In
  \bibinfo{booktitle}{{\em Programming Symposium}},
  \bibfield{editor}{\bibinfo{person}{B.~Robinet}} (Ed.).
  \bibinfo{publisher}{Springer Berlin Heidelberg}, \bibinfo{address}{Berlin,
  Heidelberg}, \bibinfo{pages}{408--425}.
\newblock
\showISBNx{978-3-540-37819-8}


\bibitem[\protect\citeauthoryear{Salomaa and Soittola}{Salomaa and
  Soittola}{1978}]%
        {Salomaa:78}
\bibfield{author}{\bibinfo{person}{A. Salomaa} {and} \bibinfo{person}{M.
  Soittola}.} \bibinfo{year}{1978}\natexlab{}.
\newblock \bibinfo{booktitle}{{\em Automata-Theoretic Aspects of Formal Power
  Series}}.
\newblock \bibinfo{publisher}{Springer-Verlag}, \bibinfo{address}{NY}.
\newblock


\bibitem[\protect\citeauthoryear{Wells}{Wells}{1999}]%
        {Wells:99}
\bibfield{author}{\bibinfo{person}{J.B. Wells}.}
  \bibinfo{year}{1999}\natexlab{}.
\newblock \showarticletitle{Typability and type checking in {System F} are
  equivalent and undecidable}.
\newblock \bibinfo{journal}{{\em Annals of Pure and Applied Logic\/}}
  \bibinfo{volume}{98}, \bibinfo{number}{1} (\bibinfo{year}{1999}),
  \bibinfo{pages}{111 -- 156}.
\newblock
\showISSN{0168-0072}
\showDOI{%
\url{https://doi.org/10.1016/S0168-0072(98)00047-5}}


\bibitem[\protect\citeauthoryear{Xu}{Xu}{2010}]%
        {Xu:2010}
\bibfield{author}{\bibinfo{person}{Hao Xu}.} \bibinfo{year}{2010}\natexlab{}.
\newblock \showarticletitle{{EriLex}: an embedded domain specific language
  generator}. In \bibinfo{booktitle}{{\em Objects, Models, Components,
  Patterns}}, \bibfield{editor}{\bibinfo{person}{Jan Vitek}} (Ed.).
  \bibinfo{publisher}{Springer}, \bibinfo{address}{Berlin, Heidelberg},
  \bibinfo{pages}{192--212}.
\newblock
\showISBNx{978-3-642-13953-6}


\bibitem[\protect\citeauthoryear{Yamazaki, Nakamaru, Ichikawa, and
  Chiba}{Yamazaki et~al\mbox{.}}{2019}]%
        {Yamazaki:2019}
\bibfield{author}{\bibinfo{person}{Tetsuro Yamazaki}, \bibinfo{person}{Tomoki
  Nakamaru}, \bibinfo{person}{Kazuhiro Ichikawa}, {and}
  \bibinfo{person}{Shigeru Chiba}.} \bibinfo{year}{2019}\natexlab{}.
\newblock \showarticletitle{Generating a fluent {API} with syntax checking from
  an {LR} grammar}.
\newblock \bibinfo{journal}{{\em Proc. ACM Program. Lang.\/}}
  \bibinfo{volume}{3}, Article \bibinfo{articleno}{Article 134}
  (\bibinfo{date}{Oct.} \bibinfo{year}{2019}), \bibinfo{numpages}{24}~pages.
\newblock
\showDOI{%
\url{https://doi.org/10.1145/3360560}}


\end{thebibliography}

\appendix

\eject
\section{Abbreviations, Acronyms, and Notation}
\label{section:symbols}
\printnomenclature

\eject
\section{Fluent API: From Practice to Theory}
\label{section:fluent}
An \emph{application programming interface} (API) provides the means to
interact with an application via a computer program. For example, using
a file system API we can open, read, and close files from within {C}
code:
\begin{JAVA}[language=c]
open(); // Open file
read(); // Read line
read(); // Read another line
close(); // Close file
\end{JAVA}
Accompanied to an API is a \emph{protocol of use}, defining rules for
good API practice. A protocol is usually brought in internal and
external documentation, delegating its imposition to the programmer.
For instance, a typical file system API protocol disallows \cc{read()}
to be called before \cc{open()}, and \cc{close()} to be called
twice in a row. Although breaking the protocol may result in malicious
run time behaviors, it nonetheless yields coherent, runnable programs.

With \emph{object oriented programming} (OOP)†{Strictly speaking, we
need only ‟object based” programming, which admits classes and objects, but
no class inheritance.}, functions (methods)
are defined within classes. To invoke a method, it must be sent as
a message to an object of the defining class.
Methods of an OO \emph{fluent API} yield objects that
accept other API methods:
\begin{codeC}[caption={\mbox{Fluent file system API implemented in {Java}}},label={listing:fluent:example}]
class ClosedFile {¢¢
  OpenedFile open() {¢…¢}
}
¢\lstbreak¢
class OpenedFile {¢¢
  OpenedFile read() {¢…¢}
  ClosedFile close() {¢…¢}
}
\end{codeC}

In this OO file system API there are two classes, \cc{ClosedFile} and \cc{OpenedFile}.
Every API call returns either an object of class \cc{ClosedFile} or an object of
class \cc{OpenedFile}, and thus may immediately be followed by a successive API
call:
\begin{code}[caption={Chain of fluent API method calls},label={listing:chain:example}]
closedFile.open().read().read().close();
\end{code}

This expression conducts multiple API calls: Invoking \cc{open}
on a \cc{ClosedFile} object yields an \cc{OpenedFile} object.
Calling \cc{read} on the \cc{OpenedFile} yields itself,
but a \cc{close} invocation returns a \cc{ClosedFile}.

The main advantage of fluent APIs is their ability to enforce a protocol \emph{at
compile time}: The object returned from API call \cc{$σᵢ$()} is missing method
\cc{$f$()}, if calling~$f$ at that location ($σ_{i+1}← f$) breaks the protocol.
Consider, for instance, finishing the methods chain of \cref{listing:chain:example}
with a second \cc{close()} call, therefore breaking the file system protocol
which forbids double closing: This call fails at compile time, raising a
compilation error, as the first \cc{close} call returns a \cc{ClosedFile}
object, defined in \cref{listing:fluent:example}, which lacks a \cc{close} method.

Fluent APIs grew in fame due to their
application for \emph{domain specific languages} (DSLs).
In contrast to general purpose programming languages, as Java and C++, DSLs
employ syntax and semantics designed for a specific component.
Standard query language (SQL), for example, is a DSL for
writing database queries.
To make use of an application in a general software library,
its DSL has to be substituted for an API.
Making the API fluent is then ideal: it makes it possible
to \emph{embed} DSL programs in code as chains of method calls,
that preserve and enforce the original syntax of the DSL.
Additional details on DSLs and fluent APIs may be found in~\cite{Gil:2019}.

A protocol or a DSL may be described by a formal language~$ℓ$: Then, the \emph{fluent
API problem} is to compile~$ℓ$ into a fluent API that enforces the protocol.
The fluent API problem is parameterized by the complexity of the input
language, and the capabilities of the host type system. The file system
protocol, for instance, is described by a regular expression,
\begin{equation*}
❨\text{\cc{open}}·\text{\cc{read}}^*·\text{\cc{close}}❩^*,
\end{equation*}
and therefore defines a regular language.
Given a class of formal languages~$𝕃$, we seek a minimal
set of type system features required to embed~$𝕃$ languages.

As many programming languages and DSLs are not regular,
practical interest lies with stronger language classes.
A popular approach is to use \emph{parametric polymorphism},
yet another common OOP feature†{Java generics, C++templates, etc.}. A fixed number of polymorphic classes
define an infinite number of types (\cc{A}, \cc{A<A>},
\cc{A<A<A>{}>},…): Intuitively, these types can be used to simulate an
unbounded storage, required to accept non-regular languages.

Consider, for example, the following Java definitions:
\begin{codeC}[caption={\mbox{Fluent stack API implemented in {Java} using (monadic) polymorphism}},label={listing:fluent:stack}]
class Empty {¢¢
  Stack<Empty> push() {¢…¢}
  Empty empty() {¢…¢}
}¢\lstbreak[-6.25pt]¢
class Stack<T> {¢¢
  Stack<Stack<T>> push() {¢…¢}
  T pop() {¢…¢}
}
\end{codeC}
With these definitions, an expression of the form
\begin{equation}
  \label{eq:stack}
  e=\cc{\kk{new} Empty().}σ₁\cc{().}σ₂\cc{().}…\cc{.}σₙ\cc{().empty()},
\end{equation}
where~$
  σᵢ∈❴\cc{push}, \cc{pop} ❵
$ type checks if and only if,~$σ₁σ₂⋯σₙ$ belongs in the Dyck language of
balanced parentheses with the homomorphism
\begin{align*}
  h(σ)=\begin{cases}
    \cc{push} & σ=\text{`\cc{(}'} ⏎
    \cc{pop} & σ=\text{`\cc{)}'}
  \end{cases}
\end{align*}
A pop from empty stack (conversely, unbalanced parenthesis) is signaled
by a type error generated at compile time, e.g., in
\begin{JAVA}
new Empty().push().pop().pop().empty();
\end{JAVA}
the second call to \cc{pop()} triggers a compile time error, to say that type
\cc{Empty} does not feature this method.

With the fluent API problem trivial for regular languages†{A finite state
machine can be encoded using simple OO classes. A Java fluent API generator for
regular languages is available at \url{https://github.com/verhas/fluflu}.}, recent
studies~\cite{Xu:2010,Nakamaru:17,Gil:Levy:2016,Gil:2019,Yamazaki:2019,Nakamaru:2020}
introduced various methods for composing fluent APIs of more complex languages.
Two promising results are those of \GR and Yamazaki
et al.~\citeyear{Yamazaki:2019}: Released roughly at the same time, both papers
showed any deterministic context free languages (including the Dyck language) can
be composed into a fluent API.

%\section{Trees, Terms, and Rewrites}
%\label{section:trees}
%\input{definitions}

\eject
\section{Proofs}
\label{section:proofs}
\subsection{Proof of \cref{theorem:PP:deep}} \label{section:PP:deep}

Recall that~$aⁿbⁿcⁿ∈\|CSL|$, and that~$\|DCFL|⊂\|CSL|$. We show
that~$aⁿbⁿcⁿ∈\A[\PP]{deep}$. The details are in \cref{listing:sd:anbncn}, that
employs Java syntax to show a set of definitions that recognizes the
language~$aⁿbⁿcⁿ$.

\begin{code}[
    caption={Definitions in type system \protect{\PP[deep]} (using Java syntax) for
    the language~$aⁿbⁿcⁿ$},
  label={listing:sd:anbncn}
]
interface γ1<x1, x2> {} // Type after reading \b$aᵏ$ is \b$\cc{$γ$1<$uₖ$,$uₖ$>}$
interface γ2<x1, x2> {} // Type after reading \b$aⁿbᵏ$ is \b$\cc{$γ$2<$u_{n-k}$,$uₙ$>}$
interface γ3<x2> {} // Type after reading \b$aⁿbⁿcᵏ$ is \b$\cc{$γ$3<$u_{n-k}$>}$
static γ1<Zero,Zero> begin() {¢¢ return null; } // chain start
static <x1, x2> γ1<Succ<x1>, Succ<x2>> a(γ1<x1, x2> e) {¢¢ return null; } // Increment both arguments
static <x1, x2> γ2<x1, x2> b(γ1<Succ<x1>, x2> e) {¢¢ return null; } // \b$b$ after \b$aⁿ$; decrement first argument
static <x1, x2> γ2<x1, x2> b(γ2<Succ<x1>, x2> e) {¢¢ return null; } // \b$b$ after \b$aⁿbᵏ$, \b$k>0$; decrement first argument
static <x> γ3<x> c(γ2<Zero, Succ<x>> e) {¢¢ return null; } // \b$c$ after \b$aⁿbⁿ$; decrement second argument
static <x> γ3<x> c(γ3<Succ<x>> e) {¢¢ return null; } // \b$c$ after \b$aⁿbⁿcᵏ$, \b$k>0$; decrement argument
static void end(γ3<Zero> e) {} // Accept after \b$aⁿbⁿcᵏ$, \b$k=n$
static {¢¢ // Test definitions in static initializer
  end(c(c(c(b(b(b(a(a(a(begin())))))))))); // Expression \b$e=e(a³b³c³)$ type-checks
  end(c(c(c(b(b(a(a(a(begin()))))))))); // Expression \b$e=e(a³b²c³)$ does not type-check
}
\end{code}

The three generic types~$γ1$,~$γ2$ and~$γ3$ rely on the unary encoding and
increment and decrement operations for maintaining counts of letters~$a$,~$b$,
and~$c$ (calls to functions \cc{a()}, \cc{b()} and \cc{c()}) in the input string:
\begin{enumerate}
  \item The type of expression
\begin{JAVA}
a($⋯$a(begin())$⋯$)
\end{JAVA} ($k$ occurrences of~\cc{a}) is {$γ$1<$uₖ$,$uₖ$>},
      where~$uₖ$ is the type encoding of~$k$;
  \item the type of
\begin{JAVA}
b($⋯$b(a($⋯$a(begin())$⋯$))$⋯$)
\end{JAVA} ($n$ occurrences of~\cc{a},~$k$ of~\cc{b}) is
    \cc{$γ$2<$u_{n-k}$,$uₙ$>}; \emph{and}
  \item the type of expression
\begin{JAVA}
c($⋯$c(b($⋯$b(a($⋯$a(begin())$⋯$))$⋯$))$⋯$)
\end{JAVA} ($n$ occurrences of~\cc{a} and~\cc{b};~$k$ occurrences of
  of~\cc{c}) is \cc{$γ$3<$u_{n-k}$>}.
\end{enumerate}

For example, observe the (overloaded) definition of function \cc{b($·$)} in the listing,
\begin{JAVA}
static <x1, x2> γ2<x1, x2> b(γ1<Succ<x1>, x2> e) {¢¢ return null; }
\end{JAVA}
This version of \cc{b($·$)}, intended for expressions of the form
\begin{JAVA}
b(a($⋯$a(begin())$⋯$))
\end{JAVA}
converts \cc{$γ1$<$uₙ$,~$uₙ$>}, the type of its argument
to \cc{$γ2$<$u_{n-1}$,~$uₙ$>}.

Consider the general case expression
\begin{JAVA}
end(c($⋯$c(b($⋯$b(a($⋯$a(begin())$⋯$))$⋯$))$⋯$))
\end{JAVA}
and, starting at the inner most invocation, \cc{begin()}, whose type is
\cc{$γ1$<$u₀$,$u₀$>}, and tracing, bottom up, types of the successive nested
expressions, we see that:
\begin{itemize}
  \item First, a count of the~$a$'s is recorded in both arguments of
    generic~$γ₁$. This count is incremented with each call to \cc{a()}.
  \item Once the first~$b$ is seen, these arguments are passed to generic~$γ2$.
    The first argument of~$γ2$ is decremented with each~$b$ encountered. The
    second argument remains however unchanged during these encounters.
  \item This second argument is then passed to generic~$γ3$ when the first~$c$
    is encountered. It is then decremented for each~$c$ encountered.
  \item Function \cc{end} type-checks only if this argument is~$u₀$.
\end{itemize}

\subsection{Proof of \cref{theorem:PP:nonlinear}} \label{section:PP:nonlinear}

The Java definitions in \cref{listing:non-linear} realize the
language~$aⁿbⁿcⁿ∈\|CSL|$.

\begin{code}[
  label={listing:non-linear},
  caption={Definitions in type system \protect{\PP[non-linear]} (using Java syntax) for
  the language~$aⁿbⁿcⁿ$}]
interface γ1<x1, x2, x3> {¢¢ // Type after reading \b$aᵏ$ is \b$γ1\cc<uₖ,u₀,u₀\cc>$
  γ1<Succ<x1>, x2, x3> a(); // No phase change: increment the first type argument
  γ2<x1, Succ<x2>, x3> b(); // First \b$b$ seen: change phase, and increment second argument
}
interface γ2<x1, x2, x3> {¢¢ // Type after reading \b$aⁿbᵏ$ is \b$γ2\cc<uₙ,uₖ,u₀\cc>$
  γ2<x1, Succ<x2>, x3> b(); // No phase change: increment the second type argument
  γ3<x1, x2, Succ<x3>> c(); // First \b$c$ seen: change phase, and increment third argument
}
interface γ3<x1, x2, x3> {¢¢ // Type after reading \b$aⁿbᵐcᵏ$ is \b$γ3\cc<uₙ,uₘ,uₖ\cc>$
  γ3<x1, x2, Succ<x3>> c(); // No phase change: increment the third type argument
}
static γ1<Zero,Zero,Zero> begin() {¢¢ return null; } // Start with type \b$γ1\cc<u₀,u₀,u₀\cc>$
static <x> void end(γ3<x, x, x> e) {} // Accept only on type \b$γ3\cc<uₙ,uₙ,uₙ\cc>$ for some \b$n≥0$
static {// Test definitions in static initializer
  end(begin().a().a().a().b().b().b().c().c().c()); // Expression \b$e=e(a³b³c³)$ type-checks
  end(begin().a().a().a().b().b().c().c().c()); } // Expression \b$e=e(a³b²c³)$ does not type-check
\end{code}

The fluent API records the number of~$a$'s,~$b$'s and~$c$'s in three unary integer
encodings. The recording is in generic types~$γ1$,~$γ2$ and~$γ3$ (each taking three
type parameters). As before, type~$γ1$ is dedicated to the first phase in
which the~$a$'s are encountered, type~$γ2$ is to the second phase in which the~$b$'s
occur, and type~$γ3$ to the final phase in which the~$c$'s show.

When the entire input is read, the three counters are compared by function
\cc{end()}. This function relies on non-linearity, to check that they are
indeed equal.

\subsection{Proof of \cref{theorem:fluent}} \label{section:theorem:fluent}

Given is a fluent program~$P=ΔΞe$. We construct from the definitions~$Δ$
and~$Ξ$ \A[DPDA]{deep} automaton~$A$. Let~$e=𝜀.σ₁.⋯.σₙ$. Then,~$A$
accepts~$w=σ₁⋯σₙ$ if and only if~$P$ is type-correct.

The construction maintains the \emph{invariant} that after~$A$ consumes~$σᵢ$
and conducting all (if any) subsequent~$𝜀$-transitions, its stack contents
encodes~$tᵢ$, the type of the partial expression~$e=𝜀.σ₁.⋯.σᵢ$. Concretely,
since \Fluent is a \textsl{monadic} type system,~$tᵢ$ must be in the (full)
form~$γ₁(γ₂(⋯γₖ(𝛜)⋯))$. The stack encoding of~$tᵢ$ is~$γ₁γ₂⋯γₖ𝛜$, i.e., the
monadic abbreviation of the full form augmented with a designated symbol~$𝛜$ for
denoting the stack's bottom. For this reason, the set of stack symbols of~$A$
includes a symbol~$γ$ for every type name used in~$Δ∪Ξ$, and the extra
symbol~$𝛜$.
\nomenclature[H]{$𝛜$}{designated stack symbol denoting the bottom of the stack}

The set of internal states of~$A$ includes an initial and \emph{accepting}
state~$q₀$. The automaton will be in state~$q₀$ initially, and then whenever
it exhausted all possible~$𝜀$-transitions after consuming a letter, and is ready
to consume the next input symbol. Also,~$A$ has an internal (not-accepting)
state~$qᵩ$ for every auxiliary function name~$φ$ used in~$Ξ$. These states are
used while executing~$𝜀$-transitions, which emulate the resolution of the
rudimentary \kk{typeof} clauses allowed in \Fluent.

As in the proof of~\cref{theorem:TA:capturing:DCFL}, the \textsl{rudimentary-typeof} property of the type
systems makes it possible to classify any function definition in~$Δ∪Ξ$ as
either \emph{direct}, if its type signature is~$τ→τ'$, or as \emph{forwarding},
in case it is~$τ→\kk{typeof}~τ'.φ$.

Every \Fluent function is encoded in one (consuming- or~$𝜀$-) transition item
of~$A$. In this encoding, the function type signature uniquely determines the
stack rewrite rule~$ρ$, but unlike in the proof of
\cref{theorem:TA:capturing:DCFL},~$ρ$ is not identical to the type signature.

To see why, recall first that since \Fluent is \textsl{monadic}, we can write
any term~$τ$ as~$𝛄x$ where~$𝛄∈Γ^*$ (in the case~$τ$ is a proper term) or as~$𝛄$
(in the case it is a grounded). If a function's type is~$𝛄x→𝛄'$, then to
maintain the invariant,~$A$ needs to push the string~$𝛄'𝛜$ to stack after
emptying it, by popping first the~$𝛄$ fixed portion, and then the~$x$
variable portion which may of \emph{unbounded length}. Alas, this~$x$ portion
cannot be cleared with the single stack rewrite allowed in the single
transition encoding a \Fluent function.
\nomenclature[H]{$𝛄$}{a string of symbols drawn from alphabet~$Γ$}

For this reason, we use instead a stack rewrite~$ρ=𝛄x→𝛄'𝛜x$ in this case, i.e.,
emulating stack emptying by pushing another copy of~$𝛜$, the bottom of the
stack symbol. Automaton~$A$ is oblivious to the trick, since none of the
rewrites in its transitions of removes a~$𝛜$ symbol off the stack.

With the definition of~$ρ(τ→τ')$ by
\begin{equation}
  ρ(τ→τ')=
  \begin{cases}
    𝛄x→𝛄'x & \text{if~$τ=𝛄x$ and~$τ'=𝛄'x$} ⏎
    𝛄𝛜→𝛄'𝛜& \text{if~$τ=𝛄$ and~$τ'=𝛄'$}⏎
    𝛄x→𝛄'𝛜x & \text{if~$τ=𝛄x$ and~$τ'=𝛄'$}
  \end{cases}
\end{equation}
we can describe the transition encoding of each of the four kinds of functions
that may occur in~$P$.
\begin{enumerate}
 \item \emph{Primary function definitions}, found in~$Δ$, are encoded as
    consuming transitions of~$A$:
    \begin{enumerate}
      \item \emph{Direct definition}~$σ:τ→τ'$ as~$⟨σ,q₀,ρ(τ→τ'),q₀⟩$,
      \item \emph{Forwarding definition}~$σ:τ→\kk{typeof}~τ'.φ$ as~$⟨σ,q₀,ρ(τ→τ'),qᵩ⟩$.
    \end{enumerate}
   \item \emph{Auxiliary function definitions}, found in~$Ξ$, are encoded
    as~$𝜀$ transitions of~$A$:
    \begin{enumerate}
      \item \emph{Direct defintion}~$φ:τ→τ'$ as~$⟨qᵩ,ρ(τ→τ'),q₀⟩$.
      \item \emph{Forwarding definition}~$φ:τ→\kk{typeof}~τ'.φ'$ as~$⟨qᵩ,ρ(τ→τ'),q_{φ'}⟩$.
    \end{enumerate}% Auxiliary function definitions
\end{enumerate}

We can now verify that automaton~$A$ iteratively computes the type of the
word-encoded input expression: Consuming transitions correspond to type
checking of primary function invocation, while~$𝜀$-transitions make the detour
required to compute the type of functions defined by a \kk{typeof} clause. If
the input expression fails type checking, then automaton~$A$ hangs (whereby
rejecting the input), due to failure to find an appropriate transition for the
current stack contents, internal state (and the current input symbol, when
appropriate).

\subsection{Proof of \protect\cref{theorem:Fluent:full}}
\label{section:Fluent:full}
  We present a set of \A[\Fluent]{full-typeof} definitions that encodes the
  language~$w＃w∈\|CSL|$.

\begin{code}[
     language=c++,
     caption={C++, \protect{\A[\Fluent]{full-typeof}} program recognizing the CSL~$w＃w$},
     label={listing:ww},morekeywords={typeof}
]
struct E {}; // Bottom type
template<typename T> struct A {}; // Generic type, stands for~$a$
template<typename T> struct B {}; // Generic type, stands for~$b$
template<typename T> struct S {}; // Generic type, stands for~$＃$
A<E> a() {} // Begin expression with~$a$
B<E> b() {} // Begin expression with~$b$
template<typename T> A<T> a(T) {} // Accumulate~$a$ to the expression
template<typename T> B<T> b(T) {} // Accumulate~$b$ to the expression
template<typename T> S<T> s(T) {} // Accumulate~$＃$ to the expression
template<typename T> auto \$(A<T>) {¢¢ return match_¢¢a(\$(T())); } // Expression has ended; eventually match~$a$
template<typename T> auto \$(B<T>) {¢¢ return match_¢¢b(\$(T())); } // Eventually match~$b$
template<typename T> auto \$(S<T>) {¢¢ return reverse(T()); } //~$＃$ encountered, reverse the second~$w$
template<typename T> auto reverse(A<T>) {¢¢ return append2end_¢¢a(reverse(T())); } // Append~$a$ to the end, reverse the rest
template<typename T> auto reverse(B<T>) {¢¢ return append2end_¢¢b(reverse(T())); } // Append~$b$ to the end, reverse the rest
E reverse(E) {} // Done reversing
template<typename T> auto append2end_¢¢a(A<T>) {¢¢ return append2start_¢¢a(append2end_¢¢a(T())); } // Reattach~$a$
template<typename T> auto append2end_¢¢a(B<T>) {¢¢ return append2start_¢¢b(append2end_¢¢a(T())); } // Reattach~$b$
A<E> append2end_¢¢a(E) {} // Append~$a$ to the end
template<typename T> auto append2end_¢¢b(A<T>) {¢¢ return append2start_¢¢a(append2end_¢¢b(T())); } // Reattach~$a$
template<typename T> auto append2end_¢¢b(B<T>) {¢¢ return append2start_¢¢b(append2end_¢¢b(T())); } // Reattach~$b$
B<E> append2end_¢¢b (E) {} // Append~$b$ to the end
template<typename T> A<T> append2start_¢¢a(T) {} // Append~$a$ to the word
template<typename T> B<T> append2start_¢¢b(T) {} // Append~$b$ to the word
template<typename T> T match_¢¢a(A<T>) {} // Match~$a$
template<typename T> T match_¢¢b(B<T>) {} // Match~$b$
int main() {¢¢
  E w1=\$(a(b(a(a(s(a(b(a(a()))))))))); // Expression encoding~$w=abaa＃abaa$ type-checks
  E w2=\$(a(b(a(a(s(a(b(b(a()))))))))); // Expression encoding~$w=abaa＃abba$ does not type-check
  E w3=\$(b(a(a(s(a(b(a(a())))))))); } // Expression encoding~$w=baa＃abaa$ does not type-check
\end{code}

The definitions first accumulate the input to a monadic type, e.g.,
where expression \cc{a(b(s(a(b()))))} is typed as \cc{A<B<S<A<B<E>{}>{}>{}>{}>},
type~\cc{E} is the bottom type
({C++} indentifiers~$s$ and~$S$ stand for~$＃$).
Actual computation is done only when reaching method \cc{\$}, which terminate all expressions.

Function \cc{\$} first traverses the first~$w$ of~$w＃w$, while replacing types~\cc{A}
and~\cc{B} with calls to \cc{match\_a} and \cc{match\_b} respectively.
Upon reaching type~\cc{S}, encoding~$＃$, function \cc{\$} encodes the second~$w$
as a type, and reverses it; then functions \cc{match\_a} and \cc{match\_b} proceed to match the
words in the correct order.
For example, expression \cc{\$(a(b(s(a(b()$⋯$)} changes first into \cc{match\_a(match\_b(\$(s(a(b()$⋯$)},
and then into \cc{match\_a(match\_b(B<A<E>{}>))}; next the \cc{match} functions
match~$b$ and then~$a$, and return the bottom type~\cc{E}, successfully terminating the
typing process. If the word before the~$＃$ differs from the word after it, this
matching is ought to fail (if typing has not yet failed).

Rest of the code deals with implementing function \cc{reverse}.
Function \cc{reverse} appends the current type~\cc{A} (resp.~\cc{B}) to
the end of the type, recursively, using function \cc{append2end\_a} (\cc{append2end\_b}).
Function \cc{append2end\_a} examines its argument \cc{A<T>} (\cc{B<T>}), replaces it with a
call to \cc{append2start\_a} (\cc{append2start\_b}) and continue recursively into~\cc{T};
type~\cc{A} (\cc{B}) is reattached after the process has ended.
Function \cc{append2end\_b} is implemented in a similar way.

\eject
\section{Supplementary Material}
\subsection{Full Encoding of a Turing Machine in~$\A[\PP]{deep,rudimentary}$}

\Cref{theorem:TA:capturing:TM} above showed that any Turing machine can be
encoded by a program in type system \[
  T=\A[\PP]{deep,rudimentary}
\] The proof of theorem used \cref{figure:Turing} depicting an example of
such a machine. For the sake of completeness, \cref{figure:Turing:cc}
here presents the full encoding in~$T$ of the Turing machine of \cref{figure:Turing}.

\begin{code}[language=c++,morekeywords={decltype},
caption={C++program encoding the Turing machine of \cref{figure:Turing}},
label={figure:Turing:cc},morekeywords={typeof}]
#define typeof decltype
template<typename x> struct B {};
struct E {};
template<typename x> struct a {};
template<typename x> struct b {};
template<typename xL, typename x, typename xR> struct O {};
template<typename xL, typename xR> E q4(O<xL, B<E>, xR>) {}
template<typename xL, typename xR> E q4(O<xL, a<E>, xR>) {}
template<typename xL, typename xR> E q4(O<xL, b<E>, xR>) {}
template<typename xL, typename xR> typeof(q4(O<xL, B<E>, B<xR>>())) q0(O<B<xL>, B<E>, xR>) {}
template<typename xR> typeof(q4(O<E, B<E>, B<xR>>())) q0(O<E, B<E>, xR>) {}
template<typename xL, typename xR> typeof(q4(O<xL, a<E>, B<xR>>())) q0(O<a<xL>, B<E>, xR>) {}
template<typename xL, typename xR> typeof(q4(O<xL, b<E>, B<xR>>())) q0(O<b<xL>, B<E>, xR>) {}
template<typename xL, typename xR> typeof(q1(O<B<xL>, B<E>, xR>())) q0(O<xL, a<E>, B<xR>>) {}
template<typename xL> typeof(q1(O<B<xL>, B<E>, E>())) q0(O<xL, a<E>, E>) {}
template<typename xL, typename xR> typeof(q1(O<B<xL>, a<E>, xR>())) q0(O<xL, a<E>, a<xR>>) {}
template<typename xL, typename xR> typeof(q1(O<B<xL>, b<E>, xR>())) q0(O<xL, a<E>, b<xR>>) {}
template<typename xL, typename xR> typeof(q2(O<xL, B<E>, B<xR>>())) q1(O<B<xL>, B<E>, xR>) {}
template<typename xR> typeof(q2(O<E, B<E>, B<xR>>())) q1(O<E, B<E>, xR>) {}
template<typename xL, typename xR> typeof(q2(O<xL, a<E>, B<xR>>())) q1(O<a<xL>, B<E>, xR>) {}
template<typename xL, typename xR> typeof(q2(O<xL, b<E>, B<xR>>())) q1(O<b<xL>, B<E>, xR>) {}
template<typename xL, typename xR> typeof(q1(O<a<xL>, B<E>, xR>())) q1(O<xL, a<E>, B<xR>>) {}
template<typename xL> typeof(q1(O<a<xL>, B<E>, E>())) q1(O<xL, a<E>, E>) {}
template<typename xL, typename xR> typeof(q1(O<a<xL>, a<E>, xR>())) q1(O<xL, a<E>, a<xR>>) {}
template<typename xL, typename xR> typeof(q1(O<a<xL>, b<E>, xR>())) q1(O<xL, a<E>, b<xR>>) {}
template<typename xL, typename xR> typeof(q1(O<b<xL>, B<E>, xR>())) q1(O<xL, b<E>, B<xR>>) {}
template<typename xL> typeof(q1(O<b<xL>, B<E>, E>())) q1(O<xL, b<E>, E>) {}
template<typename xL, typename xR> typeof(q1(O<b<xL>, a<E>, xR>())) q1(O<xL, b<E>, a<xR>>) {}
template<typename xL, typename xR> typeof(q1(O<b<xL>, b<E>, xR>())) q1(O<xL, b<E>, b<xR>>) {}
template<typename xL, typename xR> typeof(q3(O<xL, B<E>, B<xR>>())) q2(O<B<xL>, b<E>, xR>) {}
template<typename xR> typeof(q3(O<E, B<E>, B<xR>>())) q2(O<E, b<E>, xR>) {}
template<typename xL, typename xR> typeof(q3(O<xL, a<E>, B<xR>>())) q2(O<a<xL>, b<E>, xR>) {}
template<typename xL, typename xR> typeof(q3(O<xL, b<E>, B<xR>>())) q2(O<b<xL>, b<E>, xR>) {}
template<typename xL, typename xR> typeof(q0(O<B<xL>, B<E>, xR>())) q3(O<xL, B<E>, B<xR>>) {}
template<typename xL> typeof(q0(O<B<xL>, B<E>, E>())) q3(O<xL, B<E>, E>) {}
template<typename xL, typename xR> typeof(q0(O<B<xL>, a<E>, xR>())) q3(O<xL, B<E>, a<xR>>) {}
template<typename xL, typename xR> typeof(q0(O<B<xL>, b<E>, xR>())) q3(O<xL, B<E>, b<xR>>) {}
template<typename xL, typename xR> typeof(q3(O<xL, B<E>, a<xR>>())) q3(O<B<xL>, a<E>, xR>) {}
template<typename xR> typeof(q3(O<E, B<E>, a<xR>>())) q3(O<E, a<E>, xR>) {}
template<typename xL, typename xR> typeof(q3(O<xL, a<E>, a<xR>>())) q3(O<a<xL>, a<E>, xR>) {}
template<typename xL, typename xR> typeof(q3(O<xL, b<E>, a<xR>>())) q3(O<b<xL>, a<E>, xR>) {}
template<typename xL, typename xR> typeof(q3(O<xL, B<E>, b<xR>>())) q3(O<B<xL>, b<E>, xR>) {}
template<typename xR> typeof(q3(O<E, B<E>, b<xR>>())) q3(O<E, b<E>, xR>) {}
template<typename xL, typename xR> typeof(q3(O<xL, a<E>, b<xR>>())) q3(O<a<xL>, b<E>, xR>) {}
template<typename xL, typename xR> typeof(q3(O<xL, b<E>, b<xR>>())) q3(O<b<xL>, b<E>, xR>) {}
int main() {¢¢
  E w1=q0(O<E, a<E>, a<a<a<b<b<b<b<E>>>>>>>>()); // compiles, \b$w₁=a⁴b⁴∈aⁿbⁿ$
  E w2=q0(O<E, a<E>, a<a<a<b<a<b<b<E>>>>>>>>()); // does not compile, \b$w₂=a⁴bab²∉aⁿbⁿ$
  E w3=q0(O<E, a<E>, a<a<b<b<b<b<E>>>>>>>()); // does not compile, \b$w₃=a³b⁴∉aⁿbⁿ$
}
\end{code}

\subsection{Fluent API for the Language of Palindromes}
\label{section:palindrome}

Here we demonstrate \cref{theorem:ada} and its proof, by constructing a fluent
API library for palindromes in an Ada like type system, i.e., a type system
with \textsl{eventually-one-type} style of overloading resolutions.

Consider the formal language of even length palindromes over
alphabet~$❴a,b❵$, as defined by the following context free grammar
\begin{equation}\label{eq:palindrome}
  \begin{aligned}
      𝛜&→a𝛜a ⏎
      &→b𝛜b ⏎
      &→𝜀.
  \end{aligned}
\end{equation}
It is well known that the language \eq{palindrome} is not-deterministic yet
unambiguous. Rewriting its grammar in Greibach normal form gives
\begin{equation}\label{eq:Greibach}
\begin{aligned}
  𝛜 & →aγ₁ ⏎
   & →bγ₂ ⏎
  γ₁ & →aγ₁γ₃ ⏎
   & →bγ₂γ₃ ⏎
   & →a ⏎
  γ₂&→aγ₁γ₄⏎
   & →bγ₂γ₄⏎
   & →b ⏎
  γ₃&→a ⏎
  γ₄&→b.
\end{aligned}
\end{equation}

Applying the construction in the proof of \cref{theorem:ada} to the
grammar~\eq{Greibach} gives the program in \cref{listing:palindrome}, that
realizes a fluent API for \eq{palindrome}.

\captionsetup[lstlisting]{format=pruled}
\begin{code}[style=java11,style=pseudo,
caption={Definitions in type system \protect{\A{monadic, eventually-one-type}} (using Java-like syntax)
encoding the language of even lengthed palindromes},label={listing:palindrome}
%,multicols=3
]
interface¢¢ ε {¢¢
  γ1<\$> a();
  γ2<\$> b();
}
interface¢¢ γ1<T> {¢¢
  γ1<γ3<T>> a();
  γ2<γ3<T>> b();
  T a();// Java error, overloaded functions cannot differ only by return type
}
interface¢¢ γ2<T> {¢¢
  γ1<γ4<T>> a();
  γ2<γ4<T>> b();
  T b(); // Java error, overloaded functions cannot differ only by return type
}
interface¢¢ γ3<T> {¢¢
  T a();
}
interface¢¢ γ4<T> {¢¢
  T b();
}
interface \$ {¢¢
  void \$();
}

new¢¢ ε().a().a().b().b().a().a().\$();
\end{code}
%¢\columnbreak\hspace{-6.25pt}\raisebox{-3.5pt}{\rule{\textwidth}{.4pt}}\leavevmode\newline\vspace{22.5pt}¢
%¢\columnbreak\leavevmode\newline\vspace{22.5pt}¢

Note that even though the program in the listing uses Java syntax, it would
{not} provide the desired result if compiled by a Java compiler. The reason
is that Java does not permit multiple types for sub-expressions.

%Computation starts with expression \cc{\kk{new}~$\varepsilon$()} of type~\cc{$\varepsilon$}. In the first phase,
%types~\cc{γ1} and~\cc{γ2} accumulate~$w=σ₁σ₂⋯σₙ$ in a type,
%\cc{$σₙ$<$σ_{n-1}$'<}…\cc{<$σ₁$'<\$>{}>{}>{}>}; in the second phase, types
%\cc{γ₃} and \cc{γ4} accept~$w^R=σₙσ_{n-1}⋯σ₁$ by removing it, letter by letter,
%from the top of the type. The transition between the phases is done
%non-deterministically: the calls of the first phase, i.e., of~\cc{γ1}
%and~\cc{γ2}, are overloaded, and can either continue the first phase, by
%retuning~\cc{γ1} or~\cc{γ2}, or start the second, by returning \cc{γ₃} or \cc{γ4}
%(residing in the type parameter).

Expression \cc{\kk{new}~$\varepsilon$().a().a().b().b().a().a()} in \cref{listing:palindrome} is
phrased as~$aabbaa$---with this prefix, the center of the word (denoted by `$·$'),
separating~$w$ from~$w^R$, can be in three places:~$aab·$$baa$, in case~$w=aab$,
$aabba·$$a$, in case~$w=aabba$, or~$aabbaa·$, in case~$w=aabbaaσ^*$.
These three possibilities correspond to three types deduced for the expression.
Yet, when reaching method \cc{\$()}, the type checker settles the ambiguity
to the favor of the first option, as only after reading~$ww^R$ type \cc{\$}
with method \cc{\$()} is returned.
As there is exactly one way to type the entire expression, type checking is
successful.

\subsection{On the Complexity of Deep Polyadic Parametric Polymorphism}
\label{section:deep}
We take particular interest in type system~$\PP[deep]$, since it models
generic non-method functions. Also, this type system might be applicable
for the software engineering applications mentioned in \cref{section:zz}.

We don't know the exact complexity class of~$\PP[deep]$, but here are few
comments and observations that might be useful towards characterizing it.

\begin{enumerate}
  \item A tree automaton with~$𝜀$-transitions is even more potent
    than a two-pushdown automaton, which is equivalent to a Turing machine.
    This equivalence does not hold for the tree automaton in point, which
    is real-time.
  \item A direct comparison of our real-time (and hence linear time) tree
    automata to real-time (or linear time) Turing machines is not
    possible, since an elementary operation of tree automata may involve
    transformations of trees whose size may be exponential.

  \item We can still describe an emulation of the computation of \emph{real-time
    Turing machine} (RTM, see \cref{table:automata} above) by a deep tree
    automaton, by breaking the machine's tape into two stacks, and store these
    stacks as branches of the same tree,
    $\|RTM|⊆\A[TA]{deep}$.
    Let RTM$ₙ$ be an RTM equipped with~$n≥0$ linear bounded tapes. A classical
    result of Rabin~\citeyear{Rabin:63} separates the class~$\|RTM|=\|RTM|₁$
    from~$\|RTM|₂$, showing~$|RTM|₁⊊\|RTM|₂$.
    Subsequently, Bruda and Akl \citeyear{Bruda:Akl} generalized Rabin's
    result for any number of tapes, showing that~$\|RTM|ₙ⊊\|RTM|_{n+1}$, for
    all~$n≥1$. Extending the tree automaton emulation of RTMs, to run concurrently on any
    (fixed) number of tapes, we obtain that the entire non-collapsing hierarchy
    of RTMs is contained in \A[TA]{deep}, i.e, that~$\|RTM|ₙ⊆\A[TA]{deep}$ for
    all~$n≥1$.

  \item It does not seem likely that a linear time tree automata can recognize an
    arbitrary context sensitive language, a problem which is known to be PSPACE
    complete~\cite{Karp:1972}. We conjecture that~$\A[TA]{deep}\subsetneq\|CSL|$.

  \item A hint to the complexity of class~$\PP[deep]$ may be found in the fact
    that it is closed under finite intersection and finite union. (The proof
    is by merging the respective tree automata by running their rewrites in tandem on two
    distinct branches of the same tree. The merged automata recognizes
    intersection if there is an accept both branches; it recognizes the
    union, if there is an accept in one of the branches.)

  \item On the other hand, we claim that \A[TA]{deep} is not closed under
    complement (equivalently, set difference):
    Consider (yet again) the language~$aⁿbⁿcⁿ∈\A[TA]{deep}$. If there was an
    automaton that recognizes the complement of the language, it should accept
    the word~$abca$, but reject its prefix~$abc$. Alas, a stateless automaton
    such as ours, can only reject by reaching a configuration where there are
    no further legal transitions, and hence cannot recover from the rejection
    of this prefix.

    We are however able to show that \A[TA]{stateful} is closed under
    complement.

  \item Coquidé et al\@.~\citeyear{coquide1994bottom} discuss tree automata
    models similar to ours, and show that some restrictions on pattern depths and
    signature ranks are interchangeable.
\end{enumerate}

\end{document}